\documentclass[11pt]{article}
\usepackage{amsmath}
\usepackage{amsfonts}
\usepackage{amssymb}
\usepackage{amsthm}
\usepackage{mathtools}
\usepackage{bbm}
\usepackage{graphicx}
\usepackage{enumitem}
\usepackage{booktabs}
\usepackage{multirow}
\usepackage{array}
\usepackage{adjustbox}
\usepackage{subcaption}
\usepackage{natbib}
\usepackage{url}
\usepackage[colorlinks=true,linkcolor=blue,citecolor=blue,urlcolor=blue]{hyperref}
\usepackage{setspace}
\usepackage{titlesec}
\titlespacing*{\section}{0pt}{2.6ex plus .8ex minus .3ex}{1.6ex plus .3ex}
\titlespacing*{\subsection}{0pt}{2.2ex plus .6ex minus .3ex}{1.2ex plus .3ex}
\titleformat*{\section}{\large\bfseries}
\titleformat*{\subsection}{\normalsize\bfseries}
\usepackage{bibspacing}
\allowdisplaybreaks

\newtheorem{theorem}{Theorem}
\newtheorem{lemma}{Lemma}
\newtheorem{corollary}{Corollary}
\newtheorem{assumption}{Assumption}
\newtheorem{example}{Example}
\newtheorem{proposition}{Proposition}
\theoremstyle{definition}

\theoremstyle{plain}
\newtheorem{remark}{Remark}

\newcommand{\E}{\mathbb{E}}
\newcommand{\R}{\mathbb{R}}

\newcommand{\Var}{\mathbb{V}{\rm ar}}
\DeclareMathOperator*{\argmax}{arg\,max}
\DeclareMathOperator*{\argmin}{arg\,min}

\begin{document}

\def\spacingset#1{\renewcommand{\baselinestretch}%
{#1}\small\normalsize} \spacingset{1}

\title{Bayesian Change-Plane Regression}
\author{Yuki Ohnishi and Fan Li\\
  Department of Biostatistics, 
  Yale School of Public Health}
\date{}
\maketitle

\begin{abstract}
Change-plane regression represents treatment effect heterogeneity through an interpretable rule that assigns patients to two groups according to whether a linear score of baseline covariates crosses a threshold. Although easy to communicate, likelihood-based inference for the hard threshold is nonregular, invalidating standard large-sample approximations. We develop a Bayesian framework that replaces the sharp indicator with a smooth probit gate at a declared smoothing scale and treats the smoothed subgroup rule as the reported estimand. At a fixed smoothing scale, the posterior satisfies a misspecified Bernstein--von Mises theorem centered at the smoothed rule. Under a vanishing smoothing schedule, the posterior learns the boundary faster than the parametric rate, while within an explicit schedule window the regression block and treatment-effect contrast satisfy a Bernstein--von Mises theorem centered at the hard-threshold values with known-subgroup oracle efficiency. Thus, effect inference pays no first-order price for regularization. A decision-theoretic reporting protocol with consistency guarantees separates evidence of heterogeneity from reporting a boundary. Computation uses latent-variable augmentation, a great-circle elliptical slice sampler, and a normalized horseshoe prior. Simulations and an application to the PREMIER trial illustrate the method.
\end{abstract}

\spacingset{1.1}

\section{Introduction}\label{sec:intro}

Randomized clinical trials target average treatment effects, yet clinically meaningful benefit often concentrates in a subpopulation defined by baseline characteristics, and guidance asks that subgroup claims rest on prespecified analyses rather than post hoc exploration \citep{ema2019subgroup, ich2019estimands}. A finding must also be communicable, so a clinician can evaluate the rule from a few measurements. Linear-index rules, classifying a patient by whether a single score crosses a cutoff, are richer than one-covariate splits, which are underpowered and ignore joint action, yet more actionable than fully flexible estimates of conditional effects.

We study the change-plane model
\[
Y_i \;=\; W_i^\top\beta_0 \;+\; X_i^\top\gamma_0\,\mathbbm{1}\{Z_i^\top\theta_0\ge 0\} \;+\; \varepsilon_i,
\qquad i=1,\dots,n,
\]
where $W_i$ has baseline effects $\beta_0$, $X_i$ encodes a treatment contrast with subgroup-specific effect $\gamma_0$, and the unit vector $\theta_0$ defines the subgroup through $Z_i^\top\theta_0$. The boundary likelihood is a step function of $\theta$. Scalar-threshold maximizers converge at the $n$ rate with compound Poisson limits \citep{Chan1993}, while the multivariate objective is piecewise constant on the sphere. Moreover, $\theta_0$ is unidentified when $\gamma_0=0$. Standard approximations are therefore fragile where calibration matters most, and two-stage Wald intervals conditional on an estimated boundary understate uncertainty.

We replace the indicator by the probit gate $\Phi(Z_i^\top\theta/\tau)$ and base inference on its deliberately misspecified working likelihood. At fixed $\tau>0$ the posterior targets a Kullback--Leibler pseudo-true parameter defining a $\tau$-regularized rule, which classifies patients near the boundary gradually. We report this rule with $\tau$, while the hard-threshold rule remains the scientific ideal whose distance from the declared target is quantified by our theory.

At fixed $\tau$ we prove posterior contraction and a misspecified Bernstein--von Mises theorem for the gated Gaussian mixture \citep{Kleijn2012}. When $\tau$ vanishes with $n$, Gaussian errors and a positive continuous boundary-score density yield parametric localization for the regression block and faster localization for the direction. On schedules vanishing faster than $n^{-1/2}$ but no faster than essentially $n^{-1}$, the regression-block posterior is centered at the hard-threshold value with known-subgroup precision. As in classical threshold regression \citep{Chan1993,Hansen2000}, fast boundary learning prevents contamination of effect inference. The direction retains a smoothing shift, an asymmetry operationalized by our reporting protocol.

Related work falls into four strands. The first is frequentist threshold and change-plane regression. At the scalar threshold \citep{Chan1993, Hansen2000}, \citet{SeoLinton2007} smoothed the objective for Gaussian limits, extending maximum score \citep{Manski1985, Horowitz1992}. Change-plane work spans identification-loss testing \citep{BanerjeeMcKeague2007, LeeSeoShin2011, Fan2017, Kang2017}, smoothed-indicator estimation \citep{Li2021, Zhang2022}, nonregular $n$-rate inference \citep{Kang2025}, grouped and multi-threshold regimes \citep{Wei2018, YuFan2021}, and margin-condition rates \citep{Mukherjee2023}, leaving boundary uncertainty to nonstandard limits or resampling, and effect inference exposed to post-selection error.

The second is Bayesian and quasi-Bayesian inference for nonregular models. Change-point and threshold-autoregression analyses trace to \citet{CarlinGelfandSmith1992} and \citet{GewekeTerui1993}, with posterior asymptotics \citep{GhosalGhoshSamanta1999}. Building on \citet{IbragimovHasminskii1981}, \citet{ChernozhukovHong2004} gave the scalar-threshold posterior an $n$-rate nonstandard limit, and \citet{Yu2012, Yu2015} proved interval validity and minimax efficiency. Laplace-type estimators \citep{ChernozhukovHong2003}, single-index priors \citep{Antoniadis2004}, and constraint relaxation \citep{Duan2019} complete it. We give the first posterior-asymptotic treatment of the multivariate change-plane boundary, trading the $n$-rate for a jointly Gaussian posterior with tunable, reported bias (Supplementary Corollary~\ref{cor:joint_bvm}).

The third is working likelihoods and Bayesian inference under misspecification, with a frequentist history \citep{White1982} and a Bayesian counterpart spanning posterior limit theory \citep{Kleijn2012, Miller2021}, generalized posteriors \citep{Bissiri2016}, robustness and calibration \citep{GrunwaldVanOmmen2017, Syring2019, Huggins2021}, and a decision-theoretic justification \citep{Mueller2013}. We invert it, engineering misspecification at one interpretable scale and pairing it with theory quantifying its distance to the ideal.

The fourth is subgroup analysis in practice. Single-covariate screening and simple interactions are underpowered and spurious-prone \citep{Wang2007, Lipkovich2017}. Decision-theoretic subgroup-reporting treatments \citep{Sivaganesan2011, BergerWangShen2014,Schnell2016} formalize when to make a claim, an aim our protocol shares on a joint posterior for an explicit boundary. Bayesian additive regression trees and causal forests \citep{Hill2011, HahnMurrayCarvalho2020} estimate effect surfaces that the change-plane posterior complements with an interpretable rule. Selection-aware inference \citep{GuoHe2021} and policy learning with linear rules \citep{Zhao2012owl, KitagawaTetenov2018, AtheyWager2021} target welfare, not inference on the rule.

Our contributions are threefold. First, the model is paired with latent augmentation, an invariant tuning-free great-circle slice sampler, a collapsed update for small $\tau$, a normalized horseshoe prior, and a semiparametric BART extension. Second, the theory spans fixed smoothing and the anisotropic vanishing-smoothing limit, including oracle-precision inference for hard-threshold effects. Third, a decision-theoretic protocol separates clinically meaningful heterogeneity from boundary reporting, with consistency guarantees for the Bayes action and patient-level membership probabilities.


\section{Bayesian Change-Plane Regression}\label{sec:cp_model}
\subsection{Model and identification}\label{sec:model_identification}

Suppose $O_i=(Y_i,W_i,X_i,Z_i)$, $i=1,\dots,n$, are independent and identically distributed, where $Y_i\in\R$, $W_i\in\R^p$ contains baseline regressors, $X_i\in\R^r$ encodes a treatment contrast, and $Z_i\in\R^q$ contains candidate effect modifiers. The hard-threshold model is
\begin{equation}
Y_i=W_i^\top\beta_0+X_i^\top\gamma_0\,\mathbbm 1\{Z_i^\top\theta_0\ge0\}+\varepsilon_i,
\label{eq:true_dgp}
\end{equation}
where $\E[\varepsilon_i\mid W_i,X_i,Z_i]=0$ and $\E[\varepsilon_i^2\mid W_i,X_i,Z_i]=\sigma_0^2\in(0,\infty)$. Rescaling invariance is removed by $\|\theta_0\|_2=1$, placing the direction on $\mathbb S^{q-1}=\{\theta\in\R^q:\|\theta\|_2=1\}$. Write $a=(\beta,\gamma,\sigma^2)\in\mathcal A=\R^p\times\R^r\times(0,\infty)$, $\eta=(a,\theta)$, $a_0=(\beta_0,\gamma_0,\sigma_0^2)$, and $\eta_0=(a_0,\theta_0)$.

The subgroup $\{z:z^\top\theta_0\ge0\}$ is a single interpretable score with a cutoff at zero. We recommend taking $X$ a subvector of the overlapping $W$, so that $\beta_0$ is the global effect of $X$ and $\gamma_0$ the incremental effect inside the subgroup. With binary $X\in\{0,1\}$ and $W=(1,\text{baseline covariates},X)^\top$, the treatment effect is $\delta$ outside and $\delta+\gamma_0$ inside.

The antipodal flip $\theta\mapsto-\theta$ swaps the subgroup and its complement. If the heterogeneity term cannot be absorbed into the baseline, the model is point identified. If it can, the identified set is the two-point flip orbit and the hemisphere convention with positive first coordinate selects one representative. We exclude a true direction on this hemisphere seam and diagnose proximity to it through the sign-invariant summaries of Section~\ref{sec:reporting_protocol}. Supplementary Assumption~\ref{asmp:design} and Proposition~\ref{prop:identification} give the formal conditions and result.

In the remainder of the paper we impose the convention
\begin{equation}
\theta_0\in\mathbb S^{q-1}_+:=\{\theta\in\mathbb S^{q-1}:\theta_1\ge0\},
\label{eq:hemisphere}
\end{equation}
with $\theta_{0,1}>0$. Thus 
\[
\Theta=\R^p\times\R^r\times\mathbb S^{q-1}_+\times(0,\infty)
\] 
is noncompact only in the regression block, as handled in Section~\ref{sec:theorems}. The true parameters are fixed in $n$, with only $\tau_n$ varying. Supplementary Section~\ref{sec:semiparametric_extension_details} replaces $W^\top\beta$ by a BART baseline $\mu_0(W)$ \citep{Chipman2010} as a misspecification safeguard. The main text uses the parametric baseline.

\subsection{The probit-gated working likelihood and the smoothed estimand}\label{sec:working_likelihood}

Likelihood-based inference for \eqref{eq:true_dgp} is nonregular, since the likelihood is a step function in $\theta$, flat between the sample's hyperplane arrangements and discontinuous across them. Rather than confront the resulting nonstandard asymptotics \citep{ChernozhukovHong2004,Yu2012} and computational fragility, we conduct inference under a deliberately smoothed working likelihood. With $\Phi$ and $\phi$ the standard normal distribution and density, fix a smoothing scale $\tau>0$ and define the probit gate $\pi_{\theta,\tau}(z)=\Phi(z^\top\theta/\tau)\in(0,1)$.
Following the data-augmentation tradition of \citet{Albert1993}, introduce latent variables and specify the conditional working model
\begin{gather}
T_i= Z_i^\top\theta+e_i,\qquad e_i\sim\mathcal N(0,\tau^2),\qquad D_i=\mathbbm 1\{T_i\ge0\},
\label{eq:augmentation}\\
Y_i \mid (D_i,W_i,X_i,Z_i) \sim \mathcal N\big(W_i^\top\beta+X_i^\top\gamma D_i,\ \sigma^2\big), \qquad
D_i \mid Z_i \sim \mathrm{Bernoulli}\big(\pi_{\theta,\tau}(Z_i)\big).
\label{eq:working_model}
\end{gather}
Integrating out $(T_i,D_i)$, the working conditional density of $Y_i$ is the two-component Gaussian mixture
\begin{equation}
p_{\eta,\tau}(y\mid w,x,z)
=\big(1-\pi_{\theta,\tau}(z)\big)\,\varphi_{\sigma}(y-w^\top\beta)
+\pi_{\theta,\tau}(z)\,\varphi_{\sigma}(y-w^\top\beta-x^\top\gamma),
\label{eq:working_mix}
\end{equation}
where $\varphi_\sigma(u)=(2\pi\sigma^2)^{-1/2}\exp\{-u^2/(2\sigma^2)\}$, with working conditional mean
$\mu_{\eta,\tau}(w,x,z)=w^\top\beta+x^\top\gamma\,\pi_{\theta,\tau}(z)$. The construction is a Bayesian constraint relaxation \citep{Duan2019}, replacing a sharp indicator by a kernel at scale $\tau$ and restoring smoothness in $\theta$ at the price of deliberate misspecification.

At any fixed $\tau>0$ the inferential target of the working posterior is the Kullback--Leibler pseudo-true parameter
\begin{equation}
\eta^\star(\tau)\in\argmin_{\eta\in\bar\Theta}K_\tau(\eta),\qquad
K_\tau(\eta)=\E_{P_0}\big[-\log p_{\eta,\tau}(Y\mid W,X,Z)\big],
\label{eq:pseudo_true}
\end{equation}
where the prior is supported on the compact effective space $\bar\Theta$ and a minimum exists by Supplementary Lemma~\ref{lem:exist_pseudotrue_fixed_tau_new}. We report $\eta^\star(\tau)$ as the $\tau$-regularized estimand, analogous to a bandwidth-indexed smoothing functional. Section~\ref{sec:asymptotic_tau_new} places it within $O(\tau)$ of the scientific ideal $\eta_0$ and identifies schedules giving exact hard-threshold inference for $a_0$, including $\gamma_0$. Accordingly, $\tau$ is a reported sample-size-informed design choice rather than an estimated parameter, with practical guidance in Section~\ref{sec:choice_tau}.

\subsection{Prior specification}\label{sec:priors}

The default independent priors are Gaussian for $(\beta,\gamma)$, inverse-gamma for $\sigma^2$, and uniform on $\mathbb S^{q-1}_+$ for $\theta$. For theory, the regression priors are truncated to $\|\beta\|_2,\|\gamma\|_2\le10^3$ and $\sigma^2\in[10^{-6},10^6]$, giving compact support $\bar\Theta=\mathcal K_a\times\mathbb S^{q-1}_+$. Compactness is needed because the working risk is noncoercive (Section~\ref{sec:posterior_contraction_BvM_fixed_tau_new}). The sampler uses untruncated conjugate updates, but its posterior does not approach these remote bounds. With many modifiers we instead project an unconstrained normalized horseshoe prior \citep{Carvalho2009} to the hemisphere. Supplementary Section~\ref{sec:horseshoe_construction} and Proposition~\ref{prop:horseshoe_prior} establish its global-scale-free angular central Gaussian law and the local prior positivity used by the theory away from coordinate axes.

\subsection{Posterior computation}\label{sec:inference}

Posterior simulation alternates conditionally conjugate updates for $(D,T,\beta,\gamma,\sigma^2)$ with a slice-sampling update for the boundary direction. The complete cycle, including the horseshoe blocks, is summarized in Supplementary Algorithm~\ref{alg:gibbs}, and we describe here the two components specific to the change-plane structure.

\paragraph{Great-circle slice update for $\theta$.}
Given $T$, $p(\theta\mid T,\cdot)\propto\mathbbm 1\{\theta\in\mathbb S^{q-1}_+\}\exp\{-\sum_i(T_i-Z_i^\top\theta)^2/(2\tau^2)\}$. We update it with a great-circle analogue of elliptical slice sampling \citep{Murray2010,Neal2003}. A tangent-space Gaussian direction defines $\theta(\varphi)=\theta\cos\varphi+u\sin\varphi$, searched by randomized-bracket shrinkage with the hemisphere constraint in the target. For an absorbable specification we instead sample the full sphere and canonicalize each draw, avoiding the spurious seam mode left by hard truncation. Supplementary Algorithm~\ref{alg:gibbs} gives the steps. Its randomized bracket is essential for exchangeability of the cut location, requires no tuning, and avoids the stationarity bias of a fixed antipodal cut.

\begin{proposition}[Invariance of the great-circle slice update]\label{prop:ess_invariance}
The Markov transition defined by the update just described, stated precisely as steps 1--3 of Supplementary Algorithm~\ref{alg:gibbs}, leaves the conditional posterior $p(\theta\mid T,\cdot)$ on $\mathbb S^{q-1}_+$ invariant.
\end{proposition}

The proof in Supplementary Section~\ref{proof:prop:ess_invariance} reduces the update to the invariant elliptical-slice transition of \citet{Murray2010} using the uniform angular reference measure induced by the tangent draw and geodesic flow.

\paragraph{Collapsed update for small $\tau$.}
The conditional step given $T$ has scale $\tau n^{-1/2}$, whereas the marginal posterior scale is $(\tau/n)^{1/2}$, producing a relative $O(\tau^{1/2})$ step. Our collapsed update instead uses
\[
p(\theta\mid Y,a)\propto\mathbbm 1\{\theta\in\mathbb S^{q-1}_+\}\prod_i p_{\eta,\tau}(Y_i\mid W_i,X_i,Z_i)
\]
and then refreshes $(D,T)$ exactly (Supplementary Algorithm~\ref{alg:gibbs}, step 6$'$). This invariant block removes the $O(\tau^{-1/2})$ slowdown at about $2.6$ times the sweep cost. We recommend it throughout. The simulations retain augmented wide-gate fits to document their diagnostic failures.

By default we run four dispersed chains and monitor convergence with rank-normalized $\widehat R$ and effective sample sizes on sign-invariant functionals of $\theta$, such as $(z^\top\theta)^2$ and the entries of $\theta\theta^\top$, alongside $(\beta,\gamma,\sigma^2)$, paying particular attention to the smallest-$\tau$ settings. The experiments of Section~\ref{sec:simulations} follow this workflow.
\section{Asymptotic Theory}
\label{sec:theorems}

At fixed $\tau$, the posterior contracts at the parametric rate around $\eta^\star(\tau)$ and admits a misspecified Bernstein--von Mises (BvM) approximation \citep{Kleijn2012,Bissiri2016}. The main contribution is the vanishing-smoothing analysis $\tau=\tau_n\downarrow0$. Boundary information accumulates at rate $n/\tau_n$, so the boundary and regression blocks are normalized by $(n/\tau_n)^{1/2}$ and $n^{1/2}$. Up to polylogarithmic factors, with precise forms in the theorems, the relevant window is
\begin{equation}
\frac{\log n}{n}\;\ll\;\tau_n\;\ll\;\frac{1}{\sqrt n}
\label{eq:window}
\end{equation}
the regression-block posterior is centered at $a_0$ with known-subgroup oracle precision, while the faster boundary block retains a smoothing shift. Section~\ref{sec:reporting_protocol} operationalizes this division between exact effect inference and regularized boundary inference.

Let $P_0$ be the true law of $O=(Y,W,X,Z)$, $\E_{P_0}$ its expectation, and $\mathcal D_n=(O_1,\dots,O_n)$. Supplementary Assumption~\ref{asmp:A_fixed_tau}\ref{subasmp:A_fixed_chart} supplies a local $C^2$ chart $\vartheta\mapsto\theta(\vartheta)$ for the hemisphere, with nondegenerate Jacobian. Write $\tilde\eta=(a,\vartheta)\in\R^d$, $d=p+r+q$, and use subscripts $a,\vartheta$ for its blocks. Also let $\ell_\tau(\eta;O)=\log p_{\eta,\tau}(Y\mid W,X,Z)$ and $\tilde\ell_\tau(\tilde\eta;O)=\ell_\tau((a,\theta(\vartheta));O)$. The fixed-prior posterior is
\begin{equation*}
    \Pi^{(\tau)}(A\mid\mathcal D_n)
    =
    \frac{\int_A \prod_{i=1}^n p_{\eta,\tau}(Y_i\mid W_i,X_i,Z_i)\,\Pi(d\eta)} {\int_\Theta \prod_{i=1}^n p_{\eta,\tau}(Y_i\mid W_i,X_i,Z_i)\,\Pi(d\eta)}.
\end{equation*}
Supplementary Assumption~\ref{asmp:design} is maintained. For sequences, $x_n\lesssim y_n$ means $x_n\le Cy_n$ for $C$ independent of $n$, and $x_n\asymp y_n$ means both directions. The symbols $\|\cdot\|_{\mathrm{TV}},\Rightarrow,\|\cdot\|_{\psi_1},\|\cdot\|_{\psi_2}$ denote total variation, weak convergence, and sub-exponential and sub-Gaussian Orlicz norms \citep{Vershynin_2018}.

\subsection{Fixed smoothing scale}
\label{sec:posterior_contraction_BvM_fixed_tau_new}

Under the five conditions in Supplementary Assumption~\ref{asmp:A_fixed_tau}, the compact-prior posterior contracts at the parametric rate around $\eta^\star(\tau)$ and admits a total-variation misspecified BvM approximation (Supplementary Theorem~\ref{thm:mis_contraction_fixed_tau_new}). This verifies the conditions of \citet{Kleijn2012} for the gated Gaussian mixture. Supplementary Lemma~\ref{lem:exist_pseudotrue_fixed_tau_new} shows that $K_\tau$ is finite, continuous, and minimized over $\bar\Theta$ but is noncoercive on $\Theta$, so compactness cannot be removed. We assume the minimizer is unique with standard local regularity and write its chart image as $\tilde\eta^\star$. Fixed-$\tau$ credible sets require an information equality for exact calibration. Supplementary Remark~\ref{rem:sandwich_note} gives the sandwich caveat and its disappearance for the regression block as $\tau$ vanishes.

\begin{remark}[Scope of the theory]
\label{rem:scope_theory}
The theory covers compact regression-block priors and boundary priors with continuous positive local density, including the uniform law and the normalized horseshoe away from coordinate axes (Supplementary Proposition~\ref{prop:horseshoe_prior}). The sampler and protocol still apply, without contraction or BvM guarantees, to two settings used later. These are a BART baseline and the sparse high-dimensional regime, where the pseudo-true direction lies on axes at which the horseshoe envelope diverges.
\end{remark}

\subsection{Anisotropic theory under vanishing smoothing}
\label{sec:asymptotic_tau_new}

We now let $\tau=\tau_n\downarrow0$ and relate the working posterior to the hard-threshold parameter $\eta_0$. Two obstacles distinguish this limit from fixed-$\tau$ theory. Only observations in the shrinking layer $|Z^\top\theta_0|=O(\tau)$ inform the boundary at leading order, so the layer's probability and its ability to distinguish the two mixture components determine the information rate. Local curvature is also insufficient by itself because the changing working risk could develop competing minima away from $\eta_0$. Our conditions separate these local and global requirements. Throughout this subsection $U_\theta=Z^\top\theta$, $N_0\subset\mathbb S^{q-1}_+$ is a closed neighborhood of $\theta_0$ contained, with its chart preimage, in the hemisphere interior, and $\mathcal A_\delta\subset\mathrm{int}\,\mathcal K_a$ is a fixed compact neighborhood of $a_0$.

The local conditions make the boundary layer regular and informative (Supplementary Assumption~\ref{asmp:V}). Gaussian errors provide the change-of-measure identities used to compare the true component with the working mixture (\ref{subasmp:V_gauss}), while bounded design and nondegenerate oracle designs provide uniform derivative bounds and positive regression-block information (\ref{subasmp:V_bounded}). The margin condition gives $U_\theta$ a positive continuous density near zero with continuous conditional moments (\ref{subasmp:V_margin}), so a layer of width $\tau$ contains probability of order $\tau$. The wedge condition converts a directional perturbation into sign disagreement of order $\|\theta-\theta_0\|$ (\ref{subasmp:V_wedge}). Finally, positive boundary information ensures that observations in this layer distinguish the mixture components in every tangent direction (\ref{subasmp:V_boundary_info}). Its weight $\chi^2_{\mathrm B}$, defined in Supplementary Equation~\eqref{eq:chisq_weight}, integrates their squared normalized discrepancy through the layer and is positive exactly when $x^\top\gamma_0\neq0$. Together these conditions produce boundary information of order $\tau^{-1}$, bounded cross-information, and order-one regression information, which is the anisotropy driving the results below.

The remaining conditions turn this local geometry into a global posterior statement. An anchor confines all small-$\tau$ minimizers to $\mathcal A_\delta\times N_0$ (\ref{subasmp:V_anchor}), and seam identifiability rules out near-minima between the outer risk-localization region and the inner quadratic region (\ref{subasmp:V_seam}). These are computable population conditions on $K_\tau$, with the seam playing the role of multivariate profile identifiability in threshold regression \citep{Hansen2000}. Local prior positivity then supplies mass near $\tilde\eta_0$ (\ref{subasmp:V_prior}). Supplementary Assumption~\ref{asmp:V} states all conditions formally. Supplementary Section~\ref{sec:pseudo_true_discussion} verifies the margin, information, and wedge conditions analytically for a canonical design class and checks the anchor and seam by quadrature for the simulation designs.

This is the generic margin regime $\alpha=1$ \citep{Tsybakov2005,Audibert2007,Mukherjee2023}, arising when the boundary-score density is bounded and nonvanishing at the cutoff. Vanishing-density geometries remain open (Supplementary Remark~\ref{rem:nonregular_rates}). Gaussianity is essential to our change-of-measure argument. Boundedness should relax to sub-Gaussian designs by truncation, though we do not prove that extension.

The proof route uses two lemmas in Supplementary Section~\ref{sec:proofs_tau_n_new}. The smoothed and hard risks differ uniformly by $O(\tau)$, while hard excess risk is quadratic in $a-a_0$ and linear in the boundary error. This first localizes minimizers to $\|a^\star-a_0\|\lesssim\sqrt\tau$ and $\|\theta^\star-\theta_0\|\lesssim\tau$ (Supplementary Lemma~\ref{lem:risk_localization}). On that neighborhood, boundary information is $O(\tau^{-1})$, regression information is $O(1)$, and cross-information and the boundary score at $\eta_0$ are bounded (Supplementary Lemma~\ref{lem:aniso_information}). A Newton step then gives $O(\tau)$ bias in both blocks. The boundary score need not vanish. The rate comes from the diverging boundary information.

\begin{proposition}[The pseudo-true path]
\label{prop:eta_star_bias_rate_new}
Under Supplementary Assumption~\ref{asmp:V} there are $\tau_1>0$ and $C<\infty$ such that for all $\tau\in(0,\tau_1]$ any minimizer $\eta^\star(\tau)$ of $K_\tau$ satisfies 
\begin{equation*}
    \|a^\star(\tau)-a_0\|\le C\tau, \qquad \|\theta^\star(\tau)-\theta_0\|\le C\tau.
\end{equation*}
Moreover $\eta^\star(\tau)$ is eventually unique: for small $\tau$ the risk $K_\tau$ has exactly one stationary point in $\mathcal A_\delta\times N_0$, and $\tau\mapsto\tilde\eta^\star(\tau)$ is continuous.
\end{proposition}

Thus the regularized rule is within $C\tau$ of the hard-threshold rule in both blocks. The proof gives a boundary constant driven by asymmetry of the discrimination profile, and no general improvement follows without symmetry.

We can now state the two main theorems. Let $\tau_n\downarrow0$, abbreviate $\eta_n^\star=\eta^\star(\tau_n)$, $\tilde\ell_n=\tilde\ell_{\tau_n}$ and $K_n=K_{\tau_n}$, and define
\[
D_n=\mathrm{diag}\big(n^{-1/2}I_{p+r+1},\ (\tau_n/n)^{1/2}I_{q-1}\big).
\]

\begin{theorem}[Vanishing $\tau_n$: anisotropic posterior contraction]
\label{thm:tau_contraction_new}
Let Supplementary Assumption~\ref{asmp:V} hold and let $\tau_n=n^{-\varrho}$ for some $\varrho\in(0,1)$. Then there exists $M_0<\infty$ such that
\[
\Pi^{(\tau_n)}\Big(\|a-a_n^\star\|\ge M_0\,\tfrac{\log n}{\sqrt n}
\ \ \text{or}\ \
\|\theta-\theta_n^\star\|\ge M_0\,\sqrt{\tau_n}\,\tfrac{\log n}{\sqrt n}
\,\Big|\,\mathcal D_n\Big)
\xrightarrow[n\to\infty]{P_0}0 .
\]
Combining with Proposition~\ref{prop:eta_star_bias_rate_new}, the posterior locates the hard-threshold parameter at the rates
\[
\|a-a_0\|\lesssim \tfrac{\log n}{\sqrt n}+\tau_n,
\qquad
\|\theta-\theta_0\|\lesssim \tau_n+\sqrt{\tau_n}\,\tfrac{\log n}{\sqrt n},
\]
with posterior probability tending to one. The same conclusions hold for general sequences $\tau_n\downarrow0$ with $n\tau_n/(\log n)^{6}\to\infty$.
\end{theorem}

The boundary layer yields effective sample size $n/\tau_n$, while total boundary accuracy is dominated by bias $\tau_n$. As $\varrho\uparrow1$ the error approaches the $n^{-1}$ scalar-threshold rate \citep{Chan1993,Hansen2000,ChernozhukovHong2004}, with Gaussian rather than hard-threshold nonstandard shape \citep{KimPollard1990,SeoLinton2007} and without preliminary estimation or sample splitting.

\begin{theorem}[Bernstein--von Mises for the regression block at the hard-threshold target]
\label{thm:effect_bvm}
Let Supplementary Assumption~\ref{asmp:V} hold and let $\tau_n=n^{-\varrho}$ for some $\varrho\in(1/2,1)$, a polynomial schedule interior to the window \eqref{eq:window}. General sequences with $n\tau_n/(\log n)^{6}\to\infty$ and $\sqrt n\,\tau_n\to0$ are covered as well. Let $\Pi_a^{(\tau_n)}(\cdot\mid\mathcal D_n)$ denote the marginal posterior of the regression block and set $h_a=\sqrt n\,(a-a_0)$. Then
\[
\sup_{B\in\mathcal B(\R^{p+r+1})}
\Big|\Pi_a^{(\tau_n)}(h_a\in B\mid\mathcal D_n)-\mathcal N\big(\mathcal I_0^{-1}\Delta_{n,a},\ \mathcal I_0^{-1}\big)(B)\Big|
\xrightarrow[n\to\infty]{P_0}0,
\]
where $\Delta_{n,a}=n^{-1/2}\sum_{i=1}^n\dot\ell_{a}(O_i)$ is the score of the oracle model, the Gaussian regression model \eqref{eq:true_dgp} with the true subgroup indicator $\mathbbm 1\{Z^\top\theta_0\ge0\}$ treated as known, and $\mathcal I_0$ is the corresponding Fisher information, $\mathcal I_0=\mathrm{blockdiag}\Big(\sigma_0^{-2}\,\E\big[V_0V_0^\top\big],\ (2\sigma_0^4)^{-1}\Big)$, $V_0=\big(W^\top,\ \mathbbm 1\{Z^\top\theta_0\ge0\}X^\top\big)^\top$.
Moreover $\Delta_{n,a}\Rightarrow\mathcal N(0,\mathcal I_0)$ under $P_0$, so that for any coordinate functional $c^\top a$ the central $1-\alpha$ credible interval for $c^\top a$ has frequentist coverage converging to $1-\alpha$.
\end{theorem}

This is the paper's central affirmative result. The nonempty window contains $\tau_n\asymp n^{-2/3}$ and removes the first-order cost of misspecification for $\gamma_0$. After norming, cross-information is $O(\sqrt{\tau_n})$, so boundary uncertainty decouples from $a$, while $\sqrt n\tau_n\to0$ removes transported bias. The lower condition $n\tau_n\gg\log n$ preserves a quadratic approximation before the likelihood reaches the nonregular hard-threshold regime \citep{ChernozhukovHong2004,Yu2012}. Supplementary Corollary~\ref{cor:joint_bvm} gives the joint Gaussian approximation at the pseudo-true center.

\begin{remark}[The boundary block keeps its shift]
\label{rem:boundary_shift}
No analogue of Theorem~\ref{thm:effect_bvm} holds for $\theta$ centered at $\theta_0$, since in the norming $(n/\tau_n)^{1/2}$ the smoothing shift is of order $(n\tau_n)^{1/2}\to\infty$ on the window, so the boundary posterior is inference for the $\tau_n$-regularized boundary $\theta^\star(\tau_n)$, within $O(\tau_n)$ of $\theta_0$. The asymmetry is intrinsic, since removing the shift with a Gaussian limit would need $\tau_n=o((\tau_n/n)^{1/2})$, which is impossible, and exact hard-threshold boundary inference requires nonstandard limit theory \citep{ChernozhukovHong2004,Yu2012,SeoLinton2007}. The reporting protocol of Section~\ref{sec:reporting_protocol} instead propagates the regularized-boundary uncertainty into covariate-level membership statements.
\end{remark}

Supplementary Remarks~\ref{rem:rate_comparison} and \ref{rem:nonregular_rates} relate the window to the $n$ rate of hard-threshold estimation \citep{Kang2025} and treat margin exponents $\alpha\neq1$ and purely discrete $Z$, where the direction is set identified.

\subsection{Role and practical choice of the smoothing scale}
\label{sec:choice_tau}

The analyst fixes and reports one $\tau$ per data set, while $\tau_n$ describes its scaling. With standardized boundary covariates, the gate moves from $0.05$ to $0.95$ across a score width $3.29\tau$, giving $\tau$ a gray-zone interpretation. Because the window is an asymptotic sufficient condition with nonquantitative constants, we use the following theory-motivated heuristic, validated in Section~\ref{sec:simulations}. Center at $\tau=c n^{-2/3}$ for $c\in[0.5,2]$. Prespecify a sensitivity grid spanning the lower edge $\log n/n$ through a value above the window near $n^{-1/3}$, and run the collapsed update with diagnostics throughout. Also prespecify an aggregation rule requiring the reporting decision to agree over all in-window values, treating conclusions supported only by heavy smoothing as exploratory. For $n=500$--$1{,}000$, the center is about $0.005$--$0.03$ and the upper edge is $0.03$--$0.045$, which motivates the application grid. Conclusions should rest on well-diagnosed in-window fits.
\section{A Framework for Decision-Theoretic Reporting}
\label{sec:reporting_protocol}

In change-plane models the boundary direction $\theta$ is learned only through treatment-effect heterogeneity, so when $\gamma_0=0$ (with $X$ encoding the treatment contrast) the boundary is unidentified and when $\|\gamma_0\|$ is small the likelihood in $\theta$ is nearly flat. The theory of Section~\ref{sec:theorems}, like its frequentist counterparts \citep[e.g.,][]{Fan2017,Kang2025}, excludes this regime by a signal-strength condition, yet subgroup analysis must confront it, since whether clinically meaningful heterogeneity exists is usually the primary question. Weak identification then yields a diffuse, prior-driven posterior for $\theta$, so point summaries should not be reported as a subgroup rule absent evidence of heterogeneity. We therefore adopt a decision-theoretic protocol separating heterogeneity evidence from boundary reporting and establish its asymptotic operating characteristics, which to our knowledge give the first formal guarantee for a subgroup-reporting rule in the change-plane setting, valid under the same signal-strength condition.

\paragraph{Step 1. Define clinical heterogeneity.}
Let $\Delta(\gamma)$ be a one-dimensional, decision-relevant heterogeneity contrast, equal to $\gamma$ under binary treatment, to $x_{\rm clin}^\top\gamma$ for a prespecified contrast with multivariate $X$, or to $\|\gamma\|_2$ for overall magnitude. Fix a clinically meaningful threshold $\delta>0$, for instance a minimal clinically important difference, and set $H_\delta=\{|\Delta(\gamma)|\ge\delta\}$, both $\Delta$ and $\delta$ being substantive inputs fixed before analysis.

\paragraph{Step 2: a Bayes rule for whether to report a boundary.}
With actions $\mathsf{a}_0$ (``do not report a subgroup boundary'') and $\mathsf{a}_1$ (``report a subgroup boundary''), and losses $L(\mathsf{a}_1,\gamma)=c_{\rm FP}\mathbbm 1\{H_\delta^c\}$ and $L(\mathsf{a}_0,\gamma)=c_{\rm FN}\mathbbm 1\{H_\delta\}$ penalizing a false claim and a missed heterogeneity, the Bayes action minimizes posterior expected loss and takes the form
\begin{equation}
    \text{choose }\mathsf{a}_1 \iff
\Pi^{(\tau)}(H_\delta\mid \mathcal D_n)\ >\ \frac{c_{\rm FP}}{c_{\rm FP}+c_{\rm FN}}=:p_{\mathrm{report}}.
\label{eq:bayes_action}
\end{equation}
This specializes Bayesian subgroup decision theory \citep{Sivaganesan2011,BergerWangShen2014} to the change-plane setting. Since regulators penalize false subgroup claims more heavily than missed heterogeneity \citep{ema2019subgroup,Dane2019}, we use $p_{\mathrm{report}}=0.9$ as a baseline, with $[0.85,0.95]$ recommended in confirmatory settings and lower thresholds for exploratory ones, and examine sensitivity in Section~\ref{sec:sensitivity_analysis_delta_p}.

\paragraph{Step 3(a): reporting under $\mathsf{a}_0$.}
If $\mathsf{a}_0$ is selected, report posterior summaries of $\Delta(\gamma)$, the tail probability $\Pi(|\Delta(\gamma)|\ge\delta\mid\mathcal D_n)$ and a credible interval, with an explicit statement that $\theta$ is weakly identified. Any numerical summary of $\theta$ here reflects prior regularization, not evidence, and must not be read as a data-supported subgroup rule.

\paragraph{Step 3(b): reporting under $\mathsf{a}_1$.}
If $\mathsf{a}_1$ is selected, report the subgroup rule with uncertainty summaries respecting the spherical geometry and the sign convention. Let $M=\E[\theta\theta^\top\mid\mathcal D_n]$ be the posterior second-moment matrix of the boundary direction, sign-invariant and so unaffected by proximity to the hemisphere seam.

Under the rank-one projection loss $L(\theta,\theta')=\|\theta\theta^\top-\theta'\theta'^\top\|_F^2$ the Bayes point estimator is the posterior principal direction
\begin{equation}
\label{eq:posterior_principal_direction}
    \hat\theta \in \argmax_{\theta\in\mathbb S^{q-1}_+}\ \theta^\top M\theta,
\end{equation}
a leading eigenvector of $M$ (Rayleigh--Ritz), signed to the hemisphere convention, with the intrinsic Fr\'echet mean \citep{BhattacharyaPatrangenaru2003} and the radial projection of the componentwise mean coinciding asymptotically except in the weakly identified regime of Step 3(a). Since $\mathrm{tr}(M)=1$, the leading eigenvalue $\lambda_{\max}(M)\in[1/q,1]$ measures directional concentration, near $1/q$ under near-uniform dispersion, and we recommend reporting it next to $\hat\theta$.

For membership uncertainty, the posterior hard-membership probability of a covariate vector $z$ is $q(z):=\Pi^{(\tau)}\big(z^\top\theta\ge 0\mid\mathcal D_n\big)$, estimated by $B^{-1}\sum_{b=1}^B\mathbbm 1\{z^\top\theta^{(b)}\ge0\}$ over $B$ posterior draws, with $\bar q=n^{-1}\sum_iq(Z_i)$ summarizing the sample. Reporting $q(\cdot)$ propagates boundary uncertainty to the covariate level without a brittle hard classification.

Finally, report posterior summaries of $\Delta(\gamma)$, the implied effects inside and outside the subgroup when $X$ is included in $W$, and the posterior probability of clinically meaningful benefit within the subgroup.

\subsection{Asymptotic operating characteristics}
\label{sec:protocol_theory}

All protocol components are functionals of one joint posterior, so within the signal-strength regime covered by Section~\ref{sec:theorems} the contraction theory delivers guarantees for the protocol itself, with $\Delta^\star(\tau)=\Delta(\gamma^\star(\tau))$.

\begin{proposition}[Consistency of the reporting rule and of membership probabilities]
\label{prop:protocol_consistency}
Let the fixed-$\tau$ conditions of Supplementary Theorem~\ref{thm:mis_contraction_fixed_tau_new} hold at fixed $\tau>0$, and let $\Delta(\cdot)$ be continuous.
\begin{enumerate}[label=(\alph*),nosep]
\item\label{prop:pc_report} If $|\Delta^\star(\tau)|>\delta$ then $\Pi^{(\tau)}(H_\delta\mid\mathcal D_n)\xrightarrow{P_0}1$, so the protocol selects $\mathsf{a}_1$ with probability tending to one. If $|\Delta^\star(\tau)|<\delta$ it selects $\mathsf{a}_0$ with probability tending to one, for any fixed $p_{\mathrm{report}}\in(0,1)$.
\item\label{prop:pc_membership} For any $z$ with $z^\top\theta^\star(\tau)\neq0$,
$q(z)\xrightarrow{P_0}\mathbbm 1\{z^\top\theta^\star(\tau)>0\}$.
\item\label{prop:pc_window} If in addition Supplementary Assumption~\ref{asmp:V} holds and $\tau=\tau_n$ obeys the window \eqref{eq:window} (in the precise form of Theorem~\ref{thm:effect_bvm}), then (a) holds with $\Delta^\star(\tau)$ replaced by the hard-threshold contrast $\Delta(\gamma_0)$, and (b) holds with $\theta^\star(\tau)$ replaced by $\theta_0$, for any fixed $z$ off the true boundary. Under $\gamma_0$ with $|\Delta(\gamma_0)|=\delta$ exactly, and provided $\Delta$ is continuously differentiable at $\gamma_0$ with $\nabla\Delta(\gamma_0)\neq0$ (automatic for linear contrasts $\Delta(\gamma)=x_{\rm clin}^\top\gamma$, and satisfied by $\Delta=\|\cdot\|_2$ at any $\gamma_0\neq0$), the asymptotic reporting probability is governed by the Gaussian limit of Theorem~\ref{thm:effect_bvm}: $\Pi^{(\tau_n)}(H_\delta\mid\mathcal D_n)$ converges in distribution to $\Phi\big(\mathcal Z\big)$ for a standard Gaussian $\mathcal Z$, so the rule is non-degenerate at the clinical boundary.
\end{enumerate}
\end{proposition}

The proof (Supplementary Section~\ref{proof:prop:protocol_consistency}) shows the report decision and the membership statements are consistent for a declared estimand at fixed $\tau$ and for the hard-threshold quantities on the window, with the transport error from $\eta^\star(\tau)$ to $\eta_0$ controlled by Proposition~\ref{prop:eta_star_bias_rate_new}. Part~\ref{prop:pc_window} motivates a bias-aware margin treating effects within $O(\tau)\cdot\|\partial\Delta\|$ of $\delta$ as boundary cases, and the $\tau$-grid analysis of Section~\ref{sec:choice_tau} reveals when a decision is this fragile. The coherence and propagation advantages of the joint posterior over the frequentist test-then-report workflow are discussed in Supplementary Section~\ref{sec:frequentist_comparison}, with the empirical comparison in Section~\ref{sec:simulations}.
\section{Simulation Studies}
\label{sec:simulations}

We evaluate the method through a simulation suite built around the asymptotic theory of Section~\ref{sec:theorems}, one canonical data generating process with controlled departures, each study targeting a named theoretical object.
The canonical data generating process is
\begin{equation}
Y_i = W_i^\top\beta_0 + \gamma_0\,X_i\,\mathbbm 1\{Z_i^\top\theta_0\ge0\} + \varepsilon_i,
\label{eq:sim_dgp}
\end{equation}
with $W_i=(1,W_{i2},\ldots,W_{i5})^\top$, $W_{i,2:5}\sim\mathcal N_4(0,I_4)$, $Z_i=W_i$, $X_i\sim\mathrm{Bernoulli}(0.5)$, and $\varepsilon_i\sim\mathcal N(0,1)$. The fixed parameters are
\[
\begin{gathered}
\beta_0=(1.0,-0.3,-0.4,-0.65,0.4)^\top,\qquad \gamma_0=2,\\
\theta_0=\theta^*/\|\theta^*\|_2,\qquad \theta^*=(1.0,-0.8,0.7,1.2,-1.5)^\top,
\end{gathered}
\]
unless stated otherwise, using the closed-halfspace convention of Section~\ref{sec:cp_model}. The Gaussian score has a bounded positive density near zero. Supplementary Section~\ref{sec:pseudo_true_discussion} verifies the anchor and seam and computes the pseudo-true path by quadrature. A skewed companion replaces $W_{i5}$ by a centered unit-variance exponential variable to expose the order-$\tau_n$ shift. We vary $n\in\{250,500,1000,2000,4000\}$, schedules $\tau_n=cn^{-\rho}$ for $\rho\in\{1/3,1/2,2/3,5/6,1\}$ and $c\in\{0.5,1,2\}$, and fixed values $\tau\in\{0.1,0.035,0.01\}$. Thus $\rho\in\{2/3,5/6\}$ lies inside the window, $\rho\le1/2$ above it, and $\rho=1$ below its lower edge up to a logarithm. Each cell has $500$ replicates, four dispersed chains, the collapsed update, and diagnostics on sign-invariant boundary functionals. Augmented wide-gate fits document the failure mode of Section~\ref{sec:sim_effect_block}.

\subsection{Regression-block inference and the smoothing window}
\label{sec:sim_effect_block}

This study probes Theorem~\ref{thm:effect_bvm} directly, and its predictions hold across the grid (Figure~\ref{fig:sim_suite}(a,b), full schedule grid in Supplementary Table~\ref{tab:sim_gamma_coverage}). Over every cell, including the fixed values up to $\tau=0.1$ and the schedules above the window, empirical coverage of the central $95\%$ credible interval for $\gamma$ lies between $0.920$ and $0.974$, and relative to the binomial standard error at the nominal level, $0.0097$ for $500$ replicates, $74$ of $75$ cells lie within three standard errors of $0.95$. The one exception is the undercoverage cell $0.920$, the canonical design at fixed $\tau=0.035$ and $n=500$, at $3.1$ standard errors, the kind of extreme expected among $75$ cells, while the largest overcoverage cell, $0.974$ for the skewed design at $\rho=2/3$ and $n=4000$, sits at $2.5$. The absolute bias of the posterior mean never exceeds $0.02$, and the ratio of the posterior standard deviation of $\gamma$ to the known-subgroup oracle standard deviation falls from at most $1.033$ at $n=250$ to at most $1.010$ at $n=4000$, reaching $1.003$ on the interior schedule, within one percent of the oracle efficiency of Theorem~\ref{thm:effect_bvm}. The constants $c\in\{0.5,2\}$ leave these conclusions unchanged (coverage $0.934$ to $0.966$), and both window edges are conservative for this design. Above the window the transported smoothing bias is real but small, the posterior mean bias at fixed $\tau=0.1$ settling at $+0.008$ to $+0.010$ for $n\ge1000$, matching the quadrature pseudo-true shift $\gamma^\star(0.1)-\gamma_0=0.010$ of Supplementary Section~\ref{sec:pseudo_true_discussion} within Monte Carlo error. The window of Theorem~\ref{thm:effect_bvm} is therefore sufficient but not necessary for accurate regression-block inference, the sharp object above it being the pseudo-true path, computable by quadrature at design time.

A deliberately adversarial stress design shows the window binding exactly where that path predicts. Found by a quadrature search over the canonical model class for the configuration that maximizes the transported bias, it retains the canonical structure but places a two-component covariate mixture so that boundary mass concentrates asymmetrically inside the smoothing gray zone (Supplementary Section~\ref{sec:stress_design}). Coverage of the $95\%$ interval for $\gamma$ stays nominal on the interior schedule, from $0.934$ to $0.956$, but falls under deliberate oversmoothing to $0.902$ at $(n,\tau)=(4000,0.1)$, $0.866$ at $(4000,0.15)$, and $0.850$ at $(8000,0.1)$, as the fixed gate holds the transported bias nearly constant while the interval contracts with $n$. That bias tracks the quadrature pseudo-true shift $\gamma^\star(\tau)-\gamma_0$ computed before the Monte Carlo at every cell, for example $0.0451$ against the predicted $0.0469$ at $(4000,0.1)$ and $0.0548$ against $0.0598$ at $(4000,0.15)$, so the window binds precisely where the precomputable pseudo-true path says it will (Supplementary Table~\ref{tab:sim_stress}).

\begin{figure}[!htbp]
  \centering
  \includegraphics[width=\textwidth]{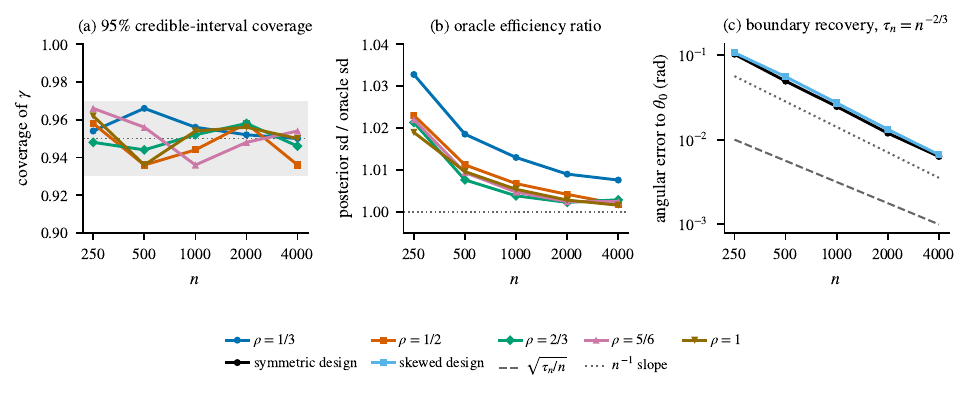}
  \caption{The simulation suite against the theory. (a) Empirical coverage of the central 95\% credible interval for $\gamma$ under the schedules $\tau_n=n^{-\rho}$, with $500$ replicates per cell. The band marks two Monte Carlo standard errors. (b) Posterior over oracle standard deviation of $\gamma$, lines and markers as in the legend of (a). (c) Angular error to $\theta_0$ at the interior schedule, with the localization scale $\sqrt{\tau_n/n}$ and an $n^{-1}$ guide.}
  \label{fig:sim_suite}
\end{figure}

The uniformity just described required the collapsed update. An earlier augmented-sampler run of the wide-gate cells failed convergence diagnostics with spurious attenuation, the collapsed rerun restores $\widehat R=1.000$ and the operating characteristics above, and the full account is in Supplementary Section~\ref{sec:sim_diagnostics}.
\subsection{Boundary recovery and the reporting protocol}
\label{sec:sim_boundary}
\label{sec:sim_protocol}

Boundary recovery probes Theorem~\ref{thm:tau_contraction_new} and Proposition~\ref{prop:eta_star_bias_rate_new} (Figure~\ref{fig:sim_suite}(c), numerical grid in Supplementary Section~\ref{sec:additional_simulation}). At the interior schedule the mean angular error to $\theta_0$ falls from $0.103$ radians at $n=250$ to $0.0063$ at $n=4000$, an empirical scaling of almost exactly $n^{-1}$, while subgroup misclassification falls from $2.7$ to $0.16$ percent and $\lambda_{\max}(M)$ rises to $0.99994$. The error sits within a factor of roughly six to ten of the localization scale $\sqrt{\tau_n/n}$, consistent with the theorem's logarithmic factors. The angular error to the pseudo-true direction $\theta^\star(\tau_n)$ is indistinguishable from the error to $\theta_0$ in both designs, so at accessible sample sizes and smoothing scales the total boundary error is dominated by posterior spread rather than the smoothing shift. That shift is a deterministic population quantity, quantified by the quadrature protocol of Supplementary Section~\ref{sec:pseudo_true_discussion}, where the skewed design exhibits it with a fitted log-log slope of $0.86$ in $\tau$.

We evaluated the reporting protocol at $\delta=1$ and $p_{\mathrm{report}}=0.9$ over six true contrasts from $\gamma_0=0$ to $2$ and $n\in\{500,1000,2000\}$, with $500$ replicates per cell. Table~\ref{tab:sim_protocol} reports the representative sample size $n=1000$. The protocol never reports a subclinical contrast, reports with probability $0.078$ at the clinical threshold, close to the limit $1-p_{\mathrm{report}}=0.10$ of Proposition~\ref{prop:protocol_consistency}, and reports every contrast above the threshold. Its credible intervals retain coverage between $0.942$ and $0.964$ across the signal grid. The \citet{Fan2017} test rejects from probability $0.626$ at $\gamma_0=0.25$ to probability one from $\gamma_0=0.5$ onward. These rejections below $\delta$ are not false positives, since the test targets any heterogeneity rather than clinically meaningful heterogeneity. The limitation appears in the subsequent plug-in interval, whose coverage falls to $0.554$ at the clinical threshold and to $0.096$ at $\gamma_0=2$, while the joint posterior continues to quantify the effect honestly after the decision.

\begin{table}[!htbp]
  \centering
  \caption{Reporting and interval operating characteristics under the canonical design \eqref{eq:sim_dgp} at the interior schedule with $n=1000$ and $500$ replicates, using $\delta=1$ and $p_{\mathrm{report}}=0.9$. The proposed columns give the probability of reporting a clinically meaningful boundary and empirical coverage of the central $95\%$ credible interval for $\gamma$. The \citet{Fan2017} columns give the rejection rate for the null of no heterogeneity and empirical coverage of the plug-in $95\%$ confidence interval for $\gamma$. Since the two decision rules target different hypotheses, rejection below $\delta$ is not a false positive for the Fan test.}
  \label{tab:sim_protocol}
  \begin{adjustbox}{max width=0.9\textwidth}
  \begin{tabular}{c cc cc}
    \toprule
    & \multicolumn{2}{c}{Proposed posterior} & \multicolumn{2}{c}{Test then plug-in} \\
    \cmidrule(lr){2-3} \cmidrule(lr){4-5}
    $\gamma_0$ & $P(\mathrm{report})$ & Coverage & Reject rate & Plug-in CI coverage \\
    \midrule
    $0$    & $0.000$ & $0.960$ & $0.034$ & $0.352$ \\
    $0.25$ & $0.000$ & $0.954$ & $0.626$ & $0.976$ \\
    $0.5$  & $0.000$ & $0.954$ & $1.000$ & $0.930$ \\
    $1$    & $0.078$ & $0.964$ & $1.000$ & $0.554$ \\
    $1.5$  & $1.000$ & $0.956$ & $1.000$ & $0.242$ \\
    $2$    & $1.000$ & $0.942$ & $1.000$ & $0.096$ \\
    \bottomrule
  \end{tabular}
  \end{adjustbox}
\end{table}

The full sample-size grid reinforces this separation of scientific targets (Supplementary Section~\ref{sec:additional_simulation}). The protocol produces no false reports at any subclinical contrast, while at the knife edge its report probabilities are $0.106$, $0.078$, and $0.094$ across the three sample sizes, with Brier scores $0.33$ to $0.35$ bracketing the value $1/3$ implied by the uniform limit. At $\gamma_0\in\{1.5,2\}$ the report probability is $0.996$ to $1.000$. The stability metric $\lambda_{\max}(M)$ rises from $0.29$ under the canonical null to $0.9998$ at the largest contrast, and the membership probabilities $q(z)$ are calibrated to within $0.03$. Under the harder absorbable specification with a homogeneous benefit and no heterogeneity, the protocol produces at most one false report in $1500$ replicates, even though the posterior can carve noise subgroups and $\Pi(|\gamma|\ge1)$ lies between $0.65$ and $0.67$ (Supplementary Table~\ref{tab:sim_protocol_null}). The demanding reporting bar, rather than the tail probability alone, holds the false-report rate between $0.000$ and $0.002$. Across the full grid the Fan plug-in interval covers only $34$ percent under the null and deteriorates to $0.078$ at $\gamma_0=1.5$ with $n=2000$.

\subsection{Comparisons}
\label{sec:sim_comparisons}

We compare the proposed posterior with four alternatives, each isolating a specific question. The exact hard-threshold posterior, fit under the same priors with the indicator gate $\mathbbm 1\{Z^\top\theta\ge0\}$ in place of the probit gate, isolates what the smoothing buys. A smoothed frequentist M estimator of the same gated objective with plug-in sandwich intervals isolates the value of full posterior uncertainty propagation, and a two-stage refit that estimates the boundary then fits the effect by ordinary least squares with a Wald interval, with and without bootstrap calibration, isolates the cost of treating the estimated boundary as fixed. The plug-in test-then-report workflow of \citet{Fan2017}, implemented through the \texttt{subdetect} package, is the nearest off-the-shelf frequentist competitor.

\begin{table}[!htbp]
  \centering
  \caption{Comparison of methods at $n=1000$ under the canonical design at the interior schedule ($\tau_n=n^{-2/3}$), $500$ replicates. Bias and RMSE of the point estimate of $\gamma$, empirical coverage and mean length of nominal $95\%$ intervals, and mean angular error of the boundary estimate to $\theta_0$. The bootstrap two-stage does not produce a boundary estimate beyond the two-stage one. The replicate set is independent of the coverage grid underlying the other tables, so small differences from the full-grid values, for example coverage $0.944$ here against $0.952$ in Supplementary Table~\ref{tab:sim_gamma_coverage}, reflect Monte Carlo error rather than a methodological discrepancy.}
  \label{tab:sim_comparisons}
  \begin{adjustbox}{max width=\textwidth}
  \begin{tabular}{lrrrrr}
    \toprule
    Method & Bias & RMSE & Coverage & CI length & Angular error\\
    \midrule
    Smoothed posterior (this paper) & $-0.002$ & $0.075$ & $0.944$ & $0.285$ & $0.024$\\
    Hard-threshold posterior        & $-0.003$ & $0.075$ & $0.944$ & $0.285$ & $0.024$\\
    Smoothed MLE, sandwich CI       & $-0.185$ & $0.319$ & $0.608$ & $0.354$ & $0.283$\\
    Two-stage refit, Wald CI        & $-0.161$ & $0.288$ & $0.606$ & $0.288$ & $0.268$\\
    Two-stage refit, bootstrap CI   & $-0.269$ & $0.384$ & $0.810$ & $0.485$ & ---\\
    \citet{Fan2017} plug-in         & $-0.473$ & $0.512$ & $0.116$ & $0.499$ & $0.297$\\
    \bottomrule
  \end{tabular}
  \end{adjustbox}
\end{table}

Table~\ref{tab:sim_comparisons} reports the results at the representative sample size $n=1000$. The smoothed and exact hard-threshold posteriors are nearly indistinguishable, with bias $-0.002$ and $-0.003$, coverage $0.944$ for both, and identical interval lengths and boundary errors, and this agreement holds across the grid ($n=250$ to $4000$, coverage $0.930$ to $0.948$). Nothing is lost against the exact comparator at the interior schedule, and the agreement is itself informative, since the smoothed family is the object for which the guarantees of Section~\ref{sec:theorems} hold uniformly over the window and the exact posterior is its $\tau\to0$ limit, so the theory developed here is the available explanation for the exact posterior's regular behavior. The practical case for smoothing therefore rests not on a coverage gap but on the diagnosable computation, the wide-gate pathology of Section~\ref{sec:sim_effect_block} being real and documented, and on the Gaussian boundary summaries. Every plug-in workflow undercovers. The smoothed maximum likelihood estimator with sandwich intervals is attenuated ($-0.18$ at $n=1000$, $-0.09$ at $n=4000$) and covers $0.57$ to $0.77$. A dedicated diagnostic (Supplementary Section~\ref{sec:mle_diagnostic}) attributes this gap entirely to optimization, since the five-start quasi-Newton search lands in an inferior basin in $60$ percent of replicates while one oracle start restores bias $+0.002$, matching the posterior on the identical likelihood. The advantage is therefore computational rather than inferential. The two-stage refit covers $0.56$ to $0.78$ with bootstrap calibration recovering only part of the gap ($0.81$ to $0.90$) while the attenuation persists, and the plug-in interval of \citet{Fan2017} is severely biased at every $n$ ($-0.44$ to $-0.48$) with coverage collapsing to $0.012$ at $n=4000$ (two degenerate-subgroup replicates at $n=250$ were caught by the harness and recorded as failures). We emphasize that \texttt{subdetect} is a sound testing procedure whose authors do not advocate the plug-in interval, which we include because it is what an analyst seeking an effect estimate after a significant test would naively form, its failure being precisely the post-selection error the joint posterior avoids.

\subsection{Robustness, high dimension, and scalability}
\label{sec:simulation_variable_selection}

Three departures from the canonical design test robustness and scale, a nonlinear baseline mean, a boundary index augmented with $50$ pure-noise covariates under the normalized horseshoe prior on $\theta$, and per-sweep timing with and without the collapsed update. Under the nonlinear baseline the parametric coverage for $\gamma$ falls to $0.155$ at $n=4000$ while the BART variant holds $0.92$ to $0.96$, the horseshoe delivers coverage $0.93$ to $0.96$ on the active boundary, both variants sitting outside the formal theory (Remark~\ref{rem:scope_theory}), and the collapsed update costs $2.3$ to $2.6$ times the augmented sweep. Full grids are in Supplementary Table~\ref{tab:sim_robust_timing}. A fourth departure fits the absorbable specification of Section~\ref{sec:model_identification} itself, the canonical design augmented with a global treatment effect $\delta_0=1$ at $\gamma_0=2$ and with opposite-signed effects ($\delta_0=1$, $\gamma_0=-2$), where coverage of the $95\%$ interval runs from $0.922$ to $0.968$ for $\gamma$ and from $0.936$ to $0.964$ for $\delta$ across the grid, with absolute bias at most $0.026$ for $\gamma$ and $0.021$ for $\delta$ and the fixed $\tau=0.1$ cells showing the same small transported bias as the canonical design (Supplementary Section~\ref{sec:absorbable_results}).

\section{Empirical Data Example}
\label{sec:analysis}
The data come from the PREMIER randomized clinical trial \citep{Svetkey2005}, a multi-site study of lifestyle modification for blood pressure control among $N=810$ generally healthy adults with above-optimal blood pressure through stage 1 hypertension, none taking antihypertensive medication at baseline. Participants were randomized to ``Advice Only,'' a multicomponent lifestyle intervention, or the same intervention augmented with counseling toward a structured dietary component, and we combine the two active groups as the treatment arm. The primary outcome is change in systolic blood pressure from baseline to six months, and because race, age, baseline hypertension status, and sex are clinically relevant candidate effect modifiers, the trial is a natural setting for detecting whether intervention benefit concentrates in a data-learned subgroup. We report a two-part analysis kept separate throughout, the primary analysis using continuous boundary covariates, standardized age and baseline systolic blood pressure with the binary indicators female, African American, and baseline hypertension status, so the boundary score has a continuous component and the continuous-margin condition of Supplementary Assumption~\ref{asmp:V}\ref{subasmp:V_margin} is plausible. The secondary analysis uses the purely binary profile covariates of the original design and illustrates the discrete-$Z$ set-identified regime discussed in Supplementary Remark~\ref{rem:nonregular_rates}.

\subsection{Primary analysis (continuous boundary covariates)}
\label{sec:analysis_result}
We follow Section~\ref{sec:reporting_protocol} with $Z_i=(1,\mathrm{age}_i,\mathrm{SBP}_i,\mathrm{female}_i,\mathrm{AA}_i,\mathrm{HTN}_i)^\top$ and $W_i=(Z_i^\top,X_i)^\top$, where $X_i$ is treatment. Under this absorbable specification, $\delta$, $\gamma$, and $\delta+\gamma$ are the effects outside, between subgroups, and inside. We fit parametric and BART baselines with normalized horseshoe priors on $\theta$ (Supplementary Section~\ref{sec:semiparametric_extension_details}). The prespecified grid is $\tau\in\{0.005,0.01,0.02,0.035,0.05,0.1\}$. For $n=810$, the approximate window $0.008$--$0.035$ places $\{0.01,0.02\}$ inside, $0.035$ at the upper edge, $0.005$ just below, and $\{0.05,0.1\}$ above. We set $\Delta(\gamma)=\gamma$, $H_\delta=\{|\gamma|\ge3\}$ mmHg, and $p_{\mathrm{report}}=0.9$, where the subscript in $H_\delta$ labels the protocol threshold rather than the treatment coefficient. The clinical threshold was fixed before analysis because PREMIER reported reductions of $1$--$6$ mmHg \citep{Svetkey2005}, while a meta-analysis of one million adults in 61 studies associated a $2$ mmHg usual-pressure reduction with about seven percent lower ischaemic heart disease mortality and ten percent lower stroke mortality in middle age \citep{Lewington2002}. A boundary is reported only if every in-window fit selects $\mathsf a_1$. Arm-specific contrasts are sensitivity analyses.

The protocol outcome is unambiguous and the two baselines agree on it. The posterior probability of the clinically meaningful event is $0.43$ at $\tau=0.02$ under both baselines, ranges from $0.41$ to $0.47$ over the in-window values, and stays far short of $p_{\mathrm{report}}=0.9$ at every grid value, so the Bayes action is $\mathsf a_0$ throughout, no pooled boundary is reported, and following Step~3(a) we report the effect summaries and the heterogeneity evidence without interpreting any point summary of $\theta$. What the data identify tightly is the treatment effect inside the learned subgroup. At $\tau=0.02$ the identified functional $\delta+\gamma$ has posterior mean $-4.31$ with standard deviation $1.32$ and $95\%$ credible interval $(-7.09,-1.90)$ under the parametric baseline and $-4.37$ ($1.33$) with interval $(-7.16,-2.01)$ under BART, and its posterior mean lies between $-4.31$ and $-4.43$ at every grid value under both baselines with standard deviation $1.26$ to $1.45$ (Figure~\ref{fig:premier_stability}(a)), clearly beyond the $3$ mmHg threshold. The decomposition of this effect is weakly identified. At $\tau=0.02$ the contrast $\gamma$ has posterior mean $-1.28$ with standard deviation $3.81$ and interval $(-9.19,6.23)$ under the parametric baseline and $-1.32$ ($4.31$) under BART, the outside effect $\delta$ is $-3.03$ ($3.43$) and $-3.05$ ($3.93$), and the posterior correlation between the two components is $-0.94$ under the parametric baseline and $-0.95$ under BART, so the sum is pinned down while its split into a global effect and a subgroup increment is not, the real-data face of the weak identification the protocol is designed for. The boundary posterior is correspondingly diffuse, with $\lambda_{\max}(M)$ between $0.26$ and $0.32$ across the pooled grid against the dispersion floor $1/q\approx0.17$ and average hard membership $\bar q$ between $0.63$ and $0.71$. The substantive conclusion is a clinically meaningful blood pressure benefit among treated participants assigned to the learned subgroup side, with no reportable evidence that the benefit differs across any covariate-defined boundary. For the identified sum, four dispersed chains give rank-normalized $\widehat R$ at most $1.012$ with bulk effective sample sizes of at least $289$ for the pooled fits and $231$ for the arm fits reported below, while the split components mix more slowly near the diffuse boundary, with in-window bulk effective sizes of $110$ to $339$ after four pooled fits, the parametric at $\tau\in\{0.05,0.1\}$ and BART at $\tau\in\{0.035,0.1\}$, were extended to $60000$ iterations, slow mixing along the genuinely weakly identified direction rather than a sampler failure (Supplementary Section~\ref{sec:additional_analysis_results}). An earlier run of the oversmoothed fits with the augmented sampler failed these diagnostics, with $\widehat R$ up to $1.40$ under the BART baseline, a real-data instance of the wide-gate mixing pathology of Section~\ref{sec:simulations} that the collapsed update resolves, and every value reported here uses the collapsed update. These fits impose the hemisphere convention by truncation, and a canonicalized full-sphere refit of the central fit leaves the reporting functional unchanged up to Monte Carlo error ($\Pi(H_\delta\mid\mathcal D_n)=0.45$ against $0.43$) while widening the convention-dependent marginals of the individual components, consistent with the weak identification of the split.

The pooled contrast averages two different active interventions, and a sensitivity analysis fixed in advance of these analyses repeats the primary analysis under the parametric baseline within each arm against advice only (Supplementary Table~\ref{tab:premier_arms}). Both arms return $\mathsf a_0$ at every grid value. In the established recommendations arm the identified effect $\delta+\gamma$ at $\tau=0.02$ is $-4.32$ with interval $(-9.84,-1.47)$, $\Pi(H_\delta\mid\mathcal D_n)$ is $0.53$ to $0.56$ across the grid, and the boundary is diffuse with $\lambda_{\max}(M)$ near $0.24$. In the arm augmented with the structured dietary component the identified effect at $\tau=0.02$ is $-4.79$ with interval $(-7.99,-0.72)$, the outside effect $\delta$ is $-3.36$ with interval $(-8.14,0.74)$, $\Pi(H_\delta\mid\mathcal D_n)$ rises to $0.72$ to $0.75$, and the boundary posterior shows a suggestively concentrated direction with $\lambda_{\max}(M)$ between $0.61$ and $0.66$, yet the evidence stays short of the reporting bar, exactly the pattern the protocol is built to decline to over-read. The arm-specific samples are roughly one third smaller and the reporting inputs were fixed in advance for the pooled contrast, so these are sensitivity findings rather than confirmatory claims.

\begin{table*}
  \centering
  \caption{Primary continuous-covariate analysis of PREMIER at the central in-window smoothing value $\tau=0.02$ for $n=810$. Columns report the posterior probability of the clinically meaningful heterogeneity event $H_\delta=\{|\gamma|\ge3\}$, the Bayes action under $p_{\mathrm{report}}=0.9$, the boundary stability $\lambda_{\max}(M)$, the average hard membership $\bar q$, and posterior means with $95\%$ credible intervals for the outside effect $\delta$, the heterogeneity contrast $\gamma$, and the identified effect inside the subgroup $\delta+\gamma$. Both baselines select $\mathsf a_0$.}
  \begin{adjustbox}{width=\textwidth}
  \begin{tabular}{c rrrr ccc}
    \toprule
      \textbf{Model}
      & $\Pi(H_\delta\mid \mathcal D_n)$ & \textup{Bayes action} & $\lambda_{\max}(M)$
      & $\bar{q}$ & $\delta$ (95\% CI) & $\gamma$ (95\% CI)
      & $\delta+\gamma$ (95\% CI)  \\
    \midrule
    Parametric & 0.43 & $\mathsf{a}_0$ & 0.26 & 0.67 & $-3.03$ $(-10.09,\,5.18)$ & $-1.28$ $(-9.19,\,6.23)$ & $-4.31$ $(-7.09,\,-1.90)$  \\
    BART       & 0.43 & $\mathsf{a}_0$ & 0.26 & 0.66 & $-3.05$ $(-9.04,\,6.61)$ & $-1.32$ $(-10.46,\,5.42)$ & $-4.37$ $(-7.16,\,-2.01)$  \\
    \bottomrule
  \end{tabular}
  \end{adjustbox}
  \label{tab:contrast_decision_rule_result}
\end{table*}

\begin{figure}[!htbp]
  \centering
  \includegraphics[width=\textwidth]{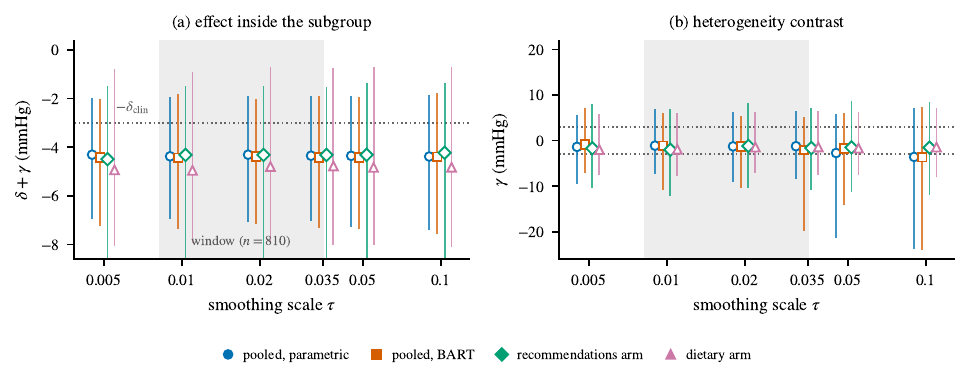}
  \caption{Stability of the PREMIER effect summaries across the smoothing grid. Panel (a) shows the identified effect inside the subgroup $\delta+\gamma$ and panel (b) the heterogeneity contrast $\gamma$, for the pooled parametric, pooled BART, established recommendations arm, and dietary arm fits. Markers show posterior means and vertical bars the central 95\% credible intervals, dodged horizontally for legibility, and all markers are open because the reporting rule selects $\mathsf{a}_0$ at every value. Dotted horizontal lines mark the clinical threshold, $-3$ mmHg in (a) and $\pm3$ mmHg in (b), and the shaded band marks the window for $n=810$. The identified effect in (a) is stable and clinically relevant across the grid, while the contrast in (b) is diffuse.}
  \label{fig:premier_stability}
\end{figure}

\subsection{Secondary analysis and sensitivity}
\label{sec:sensitivity_analysis_delta_p}
The secondary analysis uses the purely binary profile covariates of the original design, following \citet{Svetkey2005} in taking self-reported race, sex, hypertension status, and age above fifty with their two-way interactions, under the same absorbable specification with the treatment indicator in $W$. Because every component of $Z$ is a binary indicator or an interaction of indicators, the boundary score $Z^\top\theta$ takes at most $16$ distinct values, the direction is only set-identified, and the analysis is inference on the induced partition under $\tau$-regularization rather than on a unique direction (Supplementary Remark~\ref{rem:nonregular_rates}). The conclusions match the primary analysis. At $\tau=0.035$ the contrast $\gamma$ has posterior mean $-0.70$ with standard deviation $7.18$ under the parametric baseline and $-0.50$ ($3.92$) under BART, with tail probabilities $0.52$ and $0.48$, while the identified effect $\delta+\gamma$ is $-4.03$ ($1.47$) and $-4.04$ ($1.46$) and stays between $-4.0$ and $-4.3$ across the smoothing values examined, with standard deviation $1.2$ to $1.5$. The tail probability is $0.57$ and $0.65$ at $\tau=0.01$ and $0.72$ and $0.75$ at $\tau=0.1$ under the two baselines, never approaching the bar, so the Bayes action is $\mathsf a_0$ everywhere and no secondary boundary is reported. At $\tau=0.035$ the principal boundary directions of both baselines load on the product of the female and African American indicators, the leading coordinate shifting with the smoothing value, but with $\lambda_{\max}(M)$ below $0.37$ everywhere, against a dispersion floor near $0.09$ for the eleven binary and interaction covariates, they are not interpreted. Table~\ref{tab:premier_profile_summary} instead makes the boundary uncertainty visible at the patient-profile level through posterior hard-membership probabilities $q(z)$. The eight profiles cover $72.5$ percent of the sample, and their memberships are compressed between $0.53$ and $1.00$. Only younger non-African American men without baseline hypertension receive near-certain membership, no profile is confidently assigned to the complement, and the two baselines agree closely on every profile. Full decision summaries, model-dependent boundary loadings, and smoothing-sensitivity results are in Supplementary Section~\ref{sec:additional_analysis_results}.

\begin{table}[!htbp]
  \centering
  \caption{Eight most common observed clinical profiles in PREMIER and their posterior hard-membership probabilities $q(z)$ at $\tau=0.035$ in the secondary binary analysis. The eight profiles cover $72.5$ percent of the sample.}
  \begin{adjustbox}{width=0.9\textwidth}
  \begin{tabular}{l l l l r r r}
    \toprule
    \textbf{Age 50+} & \textbf{Female} & \textbf{African American} & \textbf{Baseline HTN} & \textbf{$n$ (\%)} & \textbf{$q(z)$ Parametric} & \textbf{$q(z)$ BART} \\
    \midrule
    No  & Yes & No  & No  & 103 (12.7) & 0.76 & 0.75 \\
    Yes & Yes & No  & No  & 81 (10.0)  & 0.70 & 0.68 \\
    No  & Yes & Yes & No  & 77 (9.5)   & 0.60 & 0.53 \\
    No  & No  & No  & No  & 76 (9.4)   & 1.00 & 1.00 \\
    Yes & No  & No  & No  & 73 (9.0)   & 0.78 & 0.74 \\
    Yes & Yes & No  & Yes & 73 (9.0)   & 0.67 & 0.66 \\
    Yes & No  & No  & Yes & 58 (7.2)   & 0.70 & 0.68 \\
    Yes & Yes & Yes & No  & 46 (5.7)   & 0.58 & 0.53 \\
    \bottomrule
  \end{tabular}
  \end{adjustbox}
  \label{tab:premier_profile_summary}
\end{table}

For the reporting thresholds, neither decision is fragile. The pooled primary tail probability at $\tau=0.02$ is $0.43$ under both baselines, so at the $3$ mmHg threshold the $\mathsf a_0$ action would reverse only for an evidence threshold below about $0.43$, well outside any conventional choice, and the secondary tail probabilities move between $0.48$ and $0.75$ across the smoothing values examined without approaching the bar, so no smoothing value alone reverses the secondary decision either. The protocol makes the location of the decision surface explicit rather than hiding it, and full sensitivity tables over the reporting preferences and the smoothing grid for the secondary analysis are in Supplementary Section~\ref{sec:additional_analysis_results}.

\section{Discussion}
\label{sec:conclusion}
This paper developed a Bayesian framework for change-plane regression organized around an explicitly declared smoothed estimand, replacing the brittle hard-threshold likelihood by a probit-gated working likelihood at a reported scale $\tau$ and treating the $\tau$-regularized subgroup rule $\eta^\star(\tau)$ as a legitimate estimand whose relation to the hard-threshold ideal is quantified. Under the anisotropic theory of Section~\ref{sec:theorems} the pseudo-true path lies within $O(\tau)$ of $\eta_0$, the posterior recovers the hard-threshold boundary at rate $\tau_n$ up to logarithmic factors, and on the window \eqref{eq:window} the regression-block posterior satisfies a Bernstein--von Mises theorem centered at the hard-threshold value with the known-subgroup oracle information (Theorem~\ref{thm:effect_bvm}), so contrast inference is asymptotically exact while the boundary keeps its smoothing shift at its faster scale. The reporting protocol of Section~\ref{sec:reporting_protocol} operationalizes this division of labor with consistency guarantees for the report decision and the covariate-level membership probabilities (Proposition~\ref{prop:protocol_consistency}). The framework is modular rather than tied to the Gaussian baseline emphasized here, and the same augmentation device combines with semiparametric baselines, sparsity priors for high-dimensional effect modifiers, and outcome-specific likelihoods for generalized linear \citep{Huang2021}, survival \citep{Wei2018}, and other structured settings \citep[e.g.,][]{Zhang2022}.

Several open problems remain, the first boundary geometry beyond the generic margin, where margin exponents $\alpha\neq1$ change the anisotropic rescaling and purely discrete $Z$ makes the direction set-identified, so the natural target becomes the induced partition rather than a unique direction. The second is posterior asymptotics for the semiparametric BART baseline, for which we currently offer only the working-model construction of Supplementary Section~\ref{sec:semiparametric_extension_details}. The third is weak-identification asymptotics for the reporting protocol when $\gamma_0$ is local to zero rather than fixed. The fourth is extending the theory to generalized linear \citep{Huang2021} and survival \citep{Wei2018} outcomes while preserving the interpretability of the protocol, as in related structured settings \citep{Zhang2022}.

\section*{Acknowledgment}
Research in this article was supported by the United States National Institutes of Health (NIH), National Heart, Lung, and Blood Institute (NHLBI, grant number R01-HL168202). All statements in this report, including its findings and conclusions, are solely those of the authors and do not necessarily represent the views of the NIH.

\bibliographystyle{plainnat}
\bibliography{literature}

\newpage
\appendix
\setcounter{page}{1}
\setcounter{section}{0}
\renewcommand{\thesection}{S\arabic{section}}
\renewcommand{\theequation}{S\arabic{equation}}
\renewcommand{\thefigure}{S\arabic{figure}}
\renewcommand{\thetable}{S\arabic{table}}
\renewcommand{\thetheorem}{S\arabic{theorem}}
\renewcommand{\thelemma}{S\arabic{lemma}}
\renewcommand{\thecorollary}{S\arabic{corollary}}
\renewcommand{\theproposition}{S\arabic{proposition}}
\renewcommand{\theassumption}{S\arabic{assumption}}
\renewcommand{\thedefinition}{S\arabic{definition}}
\renewcommand{\theremark}{S\arabic{remark}}
\renewcommand{\theexample}{S\arabic{example}}
\renewcommand{\appendixname}{Supplementary Material}

\section{Algorithm Details and Semiparametric Extension}
\label{sec:semiparametric_extension_details}

This section specifies the complete posterior sampling cycle summarized in Section~\ref{sec:inference}, and develops the semiparametric extension of the baseline mean referenced in Section~\ref{sec:model_identification}.

\subsection{The complete Gibbs cycle}
\label{sec:algorithm_details}

We first record every conditional update of the sampler for the parametric baseline of Section~\ref{sec:cp_model}; the modification for a nonparametric baseline is given in Supplementary Section~\ref{sec:generic_post_sample}. Throughout, $\varphi_\sigma(u)=(2\pi\sigma^2)^{-1/2}\exp\{-u^2/(2\sigma^2)\}$ and $\pi_{\theta,\tau}(z)=\Phi(z^\top\theta/\tau)$ as in Section~\ref{sec:working_likelihood}.

\newcounter{algsweep}
\providecommand{\theHalgsweep}{\arabic{algsweep}}
\refstepcounter{algsweep}\label{alg:gibbs}%
\paragraph{Algorithm~\thealgsweep\ (one MCMC sweep, parametric baseline).}
One sweep of the sampler performs, in order:
\begin{enumerate}[leftmargin=*,nosep]
\item \textbf{Update $D_i$ given $(\beta,\gamma,\sigma^2,\theta)$.}
For each $i$, the full conditional is Bernoulli with success probability
\begin{equation*}
\Pr(D_i=1\mid\cdot)=
\frac{\pi_{\theta,\tau}(Z_i)\,\varphi_\sigma\!\big(Y_i-W_i^\top\beta-X_i^\top\gamma\big)}
{\{1-\pi_{\theta,\tau}(Z_i)\}\,\varphi_\sigma\!\big(Y_i-W_i^\top\beta\big)+\pi_{\theta,\tau}(Z_i)\,\varphi_\sigma\!\big(Y_i-W_i^\top\beta-X_i^\top\gamma\big)}.
\end{equation*}
\item \textbf{Update $T_i$ given $(D_i,\theta)$.}
For each $i$, sample
\begin{align*}
    T_i\mid(D_i=1,\theta,Z_i)\sim\mathcal{TN}(Z_i^\top\theta,\tau^2,0,\infty), \quad
    T_i\mid(D_i=0,\theta,Z_i)\sim\mathcal{TN}(Z_i^\top\theta,\tau^2,-\infty,0),
\end{align*}
where $\mathcal{TN}(m,s^2,l,u)$ denotes the normal distribution with location parameter $m$ and variance $s^2$ truncated to $(l,u)$. The second argument is the variance $\tau^2$, not the standard deviation.
\item \textbf{Update $\beta$ given $(D,\gamma,\sigma^2)$.}
The residuals $Y_i-X_i^\top\gamma D_i$ follow the Gaussian regression $Y_i-X_i^\top\gamma D_i=W_i^\top\beta+\varepsilon_i$, $\varepsilon_i\sim\mathcal N(0,\sigma^2)$. Under the prior $\beta\sim\mathcal N(m_\beta,V_\beta)$, the full conditional is $\beta\mid\cdot\sim\mathcal N(m_{\beta\mid\cdot},V_{\beta\mid\cdot})$ with
\begin{align*}
    V_{\beta\mid\cdot}=\Big(V_\beta^{-1}+\sigma^{-2}\sum_{i=1}^n W_iW_i^\top\Big)^{-1},\qquad
    m_{\beta\mid\cdot}=V_{\beta\mid\cdot}\Big(V_\beta^{-1}m_\beta+\sigma^{-2}\sum_{i=1}^n W_i\big(Y_i-X_i^\top\gamma D_i\big)\Big).
\end{align*}
\item \textbf{Update $\gamma$ given $(D,\beta,\sigma^2)$.}
Let $\widetilde X_i:=D_iX_i\in\R^r$ and define the working regression $Y_i-W_i^\top\beta=\widetilde X_i^\top\gamma+\varepsilon_i$, $\varepsilon_i\sim\mathcal N(0,\sigma^2)$.
Under a Gaussian prior $\gamma\sim\mathcal N(m_\gamma,V_\gamma)$, the full conditional is $\gamma\mid\cdot \sim \mathcal N(m_{\gamma\mid\cdot},V_{\gamma\mid\cdot})$ with
\begin{align*}
    V_{\gamma\mid\cdot} =\Big(V_\gamma^{-1}+\sigma^{-2}\sum_{i=1}^n \widetilde X_i\widetilde X_i^\top\Big)^{-1},\qquad
    m_{\gamma\mid\cdot} = V_{\gamma\mid\cdot}\Big(V_\gamma^{-1}m_\gamma+\sigma^{-2}\sum_{i=1}^n \widetilde X_i \big(Y_i-W_i^\top\beta\big)\Big).
\end{align*}
In all experiments we use $m_\beta=0$, $V_\beta=100\,I_p$ and $m_\gamma=0$, $V_\gamma=100\,I_r$.
\item \textbf{Update $\sigma^2$ given $(D,\beta,\gamma)$.}
With an inverse-gamma prior $\sigma^2\sim\mathrm{InvGamma}(\sigma_a,\sigma_b)$, the full conditional is
\begin{equation*}
\sigma^2\mid\cdot \sim \mathrm{InvGamma}\!\left(\sigma_a+\frac n2,\ \sigma_b+\frac12\sum_{i=1}^n\big(Y_i-W_i^\top\beta-\widetilde X_i^\top\gamma\big)^2\right).
\end{equation*}
We fix $\sigma_a=2.0$, $\sigma_b=1.0$ throughout.
\item \textbf{Update $\theta$ given $T$.}
Under the uniform hemisphere prior, the conditional log target is $\ell_T(\theta)=-\tfrac{1}{2\tau^2}\sum_{i=1}^n(T_i-Z_i^\top\theta)^2$ truncated to $\mathbb S^{q-1}_+$, and $\theta$ is updated by the great-circle analogue of elliptical slice sampling \citep{Murray2010,Neal2003}. Given the current $\theta$, one transition performs the following three steps.

When the heterogeneity term is $W$-absorbable, for instance when the treatment is a column of $W$ so that the baseline carries an ungated treatment effect $\delta$, the flip orbit $(\theta,\gamma,\delta)\mapsto(-\theta,-\gamma,\delta+\gamma)$ is an exact reparameterization, and the posterior on the full sphere consists of two mirror-image modes. Truncating the prior to the hemisphere removes the antipodal mode but not the shoulder of its basin, which survives inside the constraint as a spurious local mode near the seam, and with strong heterogeneity dispersed chains can stick there, a failure the rank-normalized diagnostics flag immediately. We therefore impose the sign convention \eqref{eq:hemisphere} by canonicalization rather than truncation in the absorbable case. The sampler runs on the full sphere with the constraint factor removed, and each retained draw with $\theta_1<0$ is mapped to its orbit representative, $\theta\mapsto-\theta$, $\gamma\mapsto-\gamma$, $\delta\mapsto\delta+\gamma$. The two implementations target the same identified quantities, the reporting functional and the inside effect $\delta+\gamma$ are invariant to the choice by construction, and the canonicalized sampler traverses the mirror modes along the great circles that connect them. Dispersed default initialization places chains at $\gamma$ near zero with a random direction, inside a near-homogeneous basin of the absorbable posterior in which the gate is unused, and chains can fail to leave it within the burn-in. The fits therefore initialize each chain from one of the top directions of a profile scan over $128$ random hemisphere directions, distinct across chains, with the regression block initialized by least squares given the induced partition, after which rank-normalized $\widehat R$ for $\gamma$ stays below $1.05$ in all but three of the six thousand absorbable effect-study replicates, each at the smallest sample size $n=250$. The simulations of Section~\ref{sec:simulations} under the absorbable specification and the diagnostics reported there use this canonicalized update.
\begin{enumerate}[label=\textup{(\arabic*)},nosep,leftmargin=2.2em]
\item Draw $\nu\sim\mathcal N(0,I_q)$ and project onto the tangent space at $\theta$, $\nu_\perp=\nu-(\nu^\top\theta)\theta$, $u=\nu_\perp/\|\nu_\perp\|_2$, so that $u^\top\theta=0$ and $\theta(\varphi)=\theta\cos\varphi+u\sin\varphi$ traces the great circle through $\theta$ in direction $u$, with $\theta(0)=\theta$ and $\|\theta(\varphi)\|_2=1$ for all $\varphi$.
\item Draw the slice level $\log y=\ell_T(\theta)+\log U$ with $U\sim\mathrm{Unif}(0,1)$.
\item Draw $\varphi\sim\mathrm{Unif}(0,2\pi)$, initialize the bracket $(\varphi_{\min},\varphi_{\max})=(\varphi-2\pi,\varphi)$, and propose $\theta(\varphi)$. If $\theta(\varphi)\in\mathbb S^{q-1}_+$ and $\ell_T(\theta(\varphi))\ge\log y$, accept. Otherwise shrink the bracket toward zero (setting $\varphi_{\min}\leftarrow\varphi$ if $\varphi<0$ and $\varphi_{\max}\leftarrow\varphi$ otherwise), redraw $\varphi\sim\mathrm{Unif}(\varphi_{\min},\varphi_{\max})$, and repeat.
\end{enumerate}
These are the steps 1--3 referenced by Proposition~\ref{prop:ess_invariance} and its proof, with the randomized bracket of \citet{Murray2010} making the update tuning-free and exactly invariant. Under the normalized-horseshoe prior, this step is replaced by step~7 below.
\item[6$'$.] \textbf{Collapsed variant: joint update of $(\theta,D,T)$.}
First draw $\theta$ from the collapsed conditional
$p(\theta\mid Y,\beta,\gamma,\sigma^2)\propto\mathbbm 1\{\theta\in\mathbb S^{q-1}_+\}\prod_{i=1}^n p_{\eta,\tau}(Y_i\mid W_i,X_i,Z_i)$,
with $(D,T)$ integrated out analytically via the mixture \eqref{eq:working_mix}; the same great-circle slice scheme applies verbatim with $\ell_T$ replaced by the observed-data log-likelihood. Then refresh $(D_i,T_i)$ from their exact conditionals given the new $\theta$: $D_i$ from the Bernoulli of step~1 and $T_i$ from the truncated normal of step~2. Because this is a blocked draw from the joint conditional of $(\theta,D,T)$, invariance of the Gibbs scheme is preserved. The collapsed variant was originally motivated by the small-$\tau$ regime, where the augmented update of step~6 has relative step size of order $\tau^{1/2}$ (Section~\ref{sec:inference}). We recommend it at every smoothing scale (see Section~\ref{sec:simulations} for the diagnostic comparison with the augmented update).
\item[7.] \textbf{Horseshoe blocks (normalized-horseshoe prior).}
When the prior of Section~\ref{sec:priors} is used, $\theta=\theta(\nu)$ with $\nu\in\R^q$, $\nu\mid\lambda,\lambda_0\sim\mathcal N(0,\Sigma_\lambda)$, $\Sigma_\lambda=\lambda_0^2\,\mathrm{diag}(\lambda_1^2,\dots,\lambda_q^2)$. The update has two parts.
\begin{enumerate}[label=(\alph*),nosep]
\item Elliptical slice update of $\nu$. Conditional on $(\lambda_0,\lambda_1,\dots,\lambda_q)$ and $T$, the target for $\nu$ is the Gaussian prior $\mathcal N(0,\Sigma_\lambda)$ times the likelihood $L(\nu)=\exp\{\ell_T(\theta(\nu))\}$, which depends on $\nu$ only through $\theta(\nu)$. This is exactly the setting of standard elliptical slice sampling \citep{Murray2010}, and we use it as published, in particular with the standard randomized bracket. (In the collapsed variant, $\ell_T$ is again replaced by the observed-data log-likelihood.)
\item Makalic--Schmidt updates of the scales. Using the inverse-gamma auxiliary representation of the half-Cauchy priors \citep{MakalicSchmidt2016}, the four full conditionals are, for $j=1,\dots,q$,
\begin{align*}
\lambda_j^2\mid\nu_j,\xi_j,\lambda_0^2 &\sim \mathrm{InvGamma}\!\Big(1,\ \frac{1}{\xi_j}+\frac{\nu_j^2}{2\lambda_0^2}\Big),
&
\xi_j\mid\lambda_j^2 &\sim \mathrm{InvGamma}\!\Big(1,\ 1+\frac{1}{\lambda_j^2}\Big),
\\
\lambda_0^2\mid\nu,\lambda_{1:q}^2,\xi_0 &\sim \mathrm{InvGamma}\!\Big(\frac{q+1}{2},\ \frac{1}{\xi_0}+\frac12\sum_{j=1}^q\frac{\nu_j^2}{\lambda_j^2}\Big),
&
\xi_0\mid\lambda_0^2 &\sim \mathrm{InvGamma}\!\Big(1,\ 1+\frac{1}{\lambda_0^2}\Big).
\end{align*}
By Proposition~\ref{prop:horseshoe_prior}\ref{prop:hs_acg}, the global scale $\lambda_0$ does not affect the induced law of $\theta$; it is retained as a purely computational parameter that improves the conditioning of the elliptical slice update.
\end{enumerate}
\end{enumerate}
Steps 1--7 (with step~6 replaced by 6$'$ in the small-$\tau$ regime and by step~7(a) under the horseshoe prior) constitute one systematic-scan sweep; each component update leaves the joint posterior invariant.

\subsection{Semiparametric baseline}
\label{sec:semiparametric}

A practically important extension of the framework replaces the linear baseline $W^\top\beta$ in \eqref{eq:true_dgp} with an unknown regression surface $\mu_0(W)$. The motivation is robustness: misspecification of the baseline mean can distort subgroup detection, because the gate $\mathbbm 1\{Z^\top\theta_0\ge0\}$ is identified through contrasts in $Y$ that must be separated from systematic variation explained by $W$. Concretely, consider the hard-threshold model
$Y=\mu_0(W)+X^\top\gamma_0\,\mathbbm 1\{Z^\top\theta_0\ge0\}+\varepsilon$,
where $\|\theta_0\|_2=1$, $\mu_0$ is an unknown function on the support of $W$, and $\varepsilon$ is mean-zero noise as in Section~\ref{sec:cp_model}; this is a widely adopted form of the change-plane regression model in the literature \citep[e.g.,][]{Fan2017}.
We retain the probit-gated data augmentation at smoothing scale $\tau>0$ and posit the working model
\begin{equation}
\label{eq:semiparametric_working_model}
    Y\mid(D,W,X,Z) \sim \mathcal N\big(\mu(W)+X^\top\gamma D,\sigma^2\big), \quad
    D\mid Z \sim \mathrm{Bernoulli}\big(\pi_{\theta,\tau}(Z)\big),
\end{equation}
equivalently $T_i\mid(Z_i,\theta)\sim\mathcal N(Z_i^\top\theta,\tau^2)$ and $D_i=\mathbbm 1\{T_i\ge0\}$, where $\mu(\cdot)$ is endowed with a flexible prior, $\theta\in\mathbb S^{q-1}_+$, and $\sigma^2>0$. Under \eqref{eq:semiparametric_working_model}, $\mu(W)$ absorbs complex main effects of the baseline covariates, while treatment-effect heterogeneity is still summarized by the low-dimensional contrast $X^\top\gamma$ and the partition induced by $\theta$. Integrating out $(T_i,D_i)$ yields the mixture representation
$p(y\mid w,x,z)
=\big(1-\pi_{\theta,\tau}(z)\big)\,\varphi_{\sigma}(y-\mu(w))
+\pi_{\theta,\tau}(z)\,\varphi_{\sigma}(y-\mu(w)-x^\top\gamma)$,
the semiparametric analogue of \eqref{eq:working_mix}. Our default specification for $\mu(W)$ is a Bayesian additive regression trees (BART) prior \citep{Chipman2010}, whose use as a flexible baseline in causal applications is well established \citep{Hill2011}; the specification is detailed in Supplementary Section~\ref{sec:bart_specification}.

The extension preserves the main advantages of the approach: the subgroup boundary remains interpretable through the linear score $Z^\top\theta$, the heterogeneous effect is summarized by $X^\top\gamma$, and the nonparametric component $\mu(W)$ provides robustness to baseline misspecification. Because posterior inference is joint, uncertainty in $\mu(W)$ is automatically propagated into posterior summaries for $(\theta,\gamma)$, yielding coherent credible intervals and posterior probabilities for subgroup effects even when the baseline mean is complex.

\subsection{Generic posterior sampling with a nonparametric baseline}
\label{sec:generic_post_sample}

Posterior computation alternates between the latent gate variables $(D,T)$, the parametric block $(\gamma,\sigma^2,\theta)$, and the nonparametric baseline $\mu$. Write $\mu_i=\mu(W_i)$ and define the residualized outcomes $Y_i^\dagger:=Y_i-X_i^\top\gamma\,D_i$. A single MCMC sweep proceeds as follows.

\begin{enumerate}[leftmargin=*,nosep]
\item \textbf{Update $D_i$ given $(\mu,\gamma,\sigma^2,\theta)$.}
As in step~1 of Algorithm~\ref{alg:gibbs} with $W_i^\top\beta$ replaced by $\mu_i$:
\begin{equation*}
\Pr(D_i=1\mid\cdot)=
\frac{\pi_{\theta,\tau}(Z_i)\,\varphi_\sigma\!\big(Y_i-\mu_i-X_i^\top\gamma\big)}
{\{1-\pi_{\theta,\tau}(Z_i)\}\,\varphi_\sigma\!\big(Y_i-\mu_i\big)+\pi_{\theta,\tau}(Z_i)\,\varphi_\sigma\!\big(Y_i-\mu_i-X_i^\top\gamma\big)}.
\end{equation*}
\item \textbf{Update $T_i$ given $(D_i,\theta)$.}
Exactly as in step~2 of Algorithm~\ref{alg:gibbs}: a normal with location parameter $Z_i^\top\theta$ and variance $\tau^2$, truncated to $[0,\infty)$ if $D_i=1$ and to $(-\infty,0)$ if $D_i=0$.
\item \textbf{Update $\gamma$ given $(D,\mu,\sigma^2)$.}
The conjugate Gaussian update of step~4 of Algorithm~\ref{alg:gibbs} applied to the working regression $Y_i-\mu_i=\widetilde X_i^\top\gamma+\varepsilon_i$ with $\widetilde X_i=D_iX_i$; that is, $\gamma\mid\cdot\sim\mathcal N(m_{\gamma\mid\cdot},V_{\gamma\mid\cdot})$ with
\begin{align*}
    V_{\gamma\mid\cdot} =\Big(V_\gamma^{-1}+\sigma^{-2}\sum_{i=1}^n \widetilde X_i\widetilde X_i^\top\Big)^{-1},\qquad
    m_{\gamma\mid\cdot} = V_{\gamma\mid\cdot}\Big(V_\gamma^{-1}m_\gamma+\sigma^{-2}\sum_{i=1}^n \widetilde X_i (Y_i-\mu_i)\Big),
\end{align*}
and $m_\gamma=0$, $V_\gamma=100\,I_r$ throughout.
\item \textbf{Update $\sigma^2$ given $(D,\mu,\gamma)$.}
The inverse-gamma update of step~5 of Algorithm~\ref{alg:gibbs} with residuals $Y_i-\mu_i-\widetilde X_i^\top\gamma$; we fix $\sigma_a=2.0$, $\sigma_b=1.0$ (but see the remark in Supplementary Section~\ref{sec:bart_specification} for the BART case).
\item \textbf{Update $\theta$ given $T$.}
This step is unchanged by the semiparametric baseline. Conditional on $T$, the relevant likelihood term is Gaussian, $T_i\mid(Z_i,\theta)\sim\mathcal N(Z_i^\top\theta,\tau^2)$, with $\theta\in\mathbb S^{q-1}_+$. Under the uniform prior on $\mathbb S^{q-1}_+$ we use the great-circle slice sampler of Section~\ref{sec:inference}, with log target proportional to $-\frac{1}{2\tau^2}\sum_{i=1}^n(T_i-Z_i^\top\theta)^2$ and proposals restricted to the hemisphere via truncation of the log target outside $\mathbb S^{q-1}_+$; under the normalized-horseshoe prior of Section~\ref{sec:priors}, the elliptical-slice and Makalic--Schmidt blocks of step~7 of Algorithm~\ref{alg:gibbs} apply unchanged. The collapsed variant (step~6$'$) likewise applies verbatim with the mixture density evaluated at $\mu_i$ in place of $W_i^\top\beta$.
\item \textbf{Update $\mu$ given $(Y^\dagger,W,\sigma^2)$.}
Given $(D,\gamma)$, the model reduces to a nonparametric Gaussian regression of $Y^\dagger$ on $W$, $Y_i^\dagger=\mu(W_i)+\varepsilon_i$, $\varepsilon_i\sim\mathcal N(0,\sigma^2)$,
so the $\mu$-update depends only on the chosen prior for $\mu(\cdot)$; the BART update is given in Supplementary Section~\ref{sec:bart_specification}.
\end{enumerate}

\subsection{BART specification}
\label{sec:bart_specification}

\subsubsection{Sum-of-trees representation}
BART \citep{Chipman2010} represents the baseline mean as a sum of $m$ regression trees, $\mu(W)=\sum_{t=1}^m g(W;\mathcal T_t,\mathcal M_t)$,
where each tree $\mathcal T_t$ defines a partition of the predictor space and $\mathcal M_t$ collects terminal-node means. Given $(\mathcal T_t,\mathcal M_t)_{t=1}^m$, the baseline vector is $f_i=\mu(W_i)=\sum_{t=1}^m g(W_i;\mathcal T_t,\mathcal M_t)$, and $Y_i^\dagger\mid(\mu,\sigma^2)\sim\mathcal N(\mu(W_i),\sigma^2)$.

\subsubsection{Posterior inference}
The BART update uses the standard Bayesian backfitting strategy, conditional on $(Y^\dagger,W,\sigma^2)$:
\begin{enumerate}[leftmargin=*,nosep]
\item For $t=1,\dots,m$, form the partial residuals $R_i^{(t)}=Y_i^\dagger-\sum_{\ell\neq t} g(W_i;\mathcal T_\ell,\mathcal M_\ell)$.
\item \textbf{Update the tree structure $\mathcal T_t$.}
Propose a local modification of $\mathcal T_t$ (e.g., grow or prune a terminal node) and accept or reject the move by a Metropolis--Hastings step based on the integrated likelihood obtained by analytically integrating out terminal-node means under their Gaussian prior.
\item \textbf{Update terminal-node means $\mathcal M_t$.}
Conditional on $\mathcal T_t$, each terminal node $b$ collects observations $\{i:W_i\in b\}$. Under a Gaussian leaf prior $\mu_b\sim\mathcal N(0,\sigma_\mu^2)$ and Gaussian noise variance $\sigma^2$, the full conditional is
\begin{align*}
\mu_b\mid\cdot &\sim \mathcal N(m_{b\mid\cdot},V_{b\mid\cdot}),\qquad
V_{b\mid\cdot}=\Big(\sigma_\mu^{-2}+\frac{n_b}{\sigma^2}\Big)^{-1},\qquad
m_{b\mid\cdot}=V_{b\mid\cdot}\frac{1}{\sigma^2}\sum_{i:W_i\in b} R_i^{(t)},
\end{align*}
where $n_b$ is the number of observations in node $b$.
\item After updating $(\mathcal T_t,\mathcal M_t)$, refresh the fitted values of tree $t$ and proceed to $t+1$.
\end{enumerate}
After all trees are updated, set $\mu(W_i)=\sum_{t=1}^m g(W_i;\mathcal T_t,\mathcal M_t)$ and return to the generic sampler steps for $(D,T,\gamma,\sigma^2,\theta)$.

\subsubsection{Prior specifications}
The BART prior is characterized by: (i) the number of trees $m$; (ii) a depth-penalizing splitting rule $\Pr(\text{split at depth }d)=\alpha_{\mathrm{tree}}(1+d)^{-\beta_{\mathrm{tree}}}$, which favors shallow trees and regularizes $\mu(\cdot)$; (iii) a choice of splitting variable and cutpoint distribution (typically uniform among available predictors and observed cutpoints); and (iv) a Gaussian prior on leaf means, $\mu_b\sim\mathcal N(0,\sigma_\mu^2)$. Following standard scaling heuristics, $\sigma_\mu$ is chosen as $\sigma_\mu=c/(k\sqrt m)$ after putting $Y^\dagger$ on an interpretable scale, where $k>0$ controls shrinkage and $c$ is proportional to a marginal scale of $Y^\dagger$ (e.g., its empirical standard deviation). These defaults yield a weakly informative prior that keeps individual trees small while allowing the ensemble to represent complex regression surfaces. We fix $(m, \alpha_{\mathrm{tree}}, \beta_{\mathrm{tree}}, k, c)=(200, 0.95, 2.0, 2.0, 1.0)$ in our simulations and analysis for stable mixing.

\begin{remark}[A single source of truth for $\sigma^2$]
\label{rem:sigma_single_source}
Under the BART specification of $\mu$, the residual variance $\sigma^2$ in \eqref{eq:semiparametric_working_model} is the same object as the residual variance maintained by BART internally; the $\sigma^2$-update in step~4 of the generic sampler is then the BART scaled-inverse-$\chi^2$ update of \citet{Chipman2010}, and a single source of truth for $\sigma^2$ must be enforced at the implementation level, meaning the variance is updated exactly once per sweep and shared by the regression and tree blocks, to avoid double counting.
\end{remark}

\section{The normalized horseshoe prior}\label{sec:horseshoe_construction}

When $Z$ contains many candidate effect modifiers, an unregularized direction is pulled toward noise coordinates. To learn a sparse boundary while preserving the interpretable score $z^\top\theta$, we replace the uniform prior by a normalized horseshoe prior \citep{Carvalho2009}. Introduce an unconstrained $\nu\in\R^q$ with
\begin{equation}
\nu_j \mid \lambda_j,\lambda_0 \sim \mathcal N\big(0,\lambda_0^2\lambda_j^2\big),\qquad
\lambda_0\sim\mathcal C^+(0,1),\qquad \lambda_j\sim\mathcal C^+(0,1)\quad (j=1,\dots,q),
\label{eq:horseshoe}
\end{equation}
where $\mathcal C^+(0,1)$ is the standard half-Cauchy distribution, and define
\[
\theta(\nu)=\mathrm{sign}(\nu_1)\,\nu/\|\nu\|_2\in\mathbb S^{q-1}_+,\qquad \mathrm{sign}(0):=1.
\]
This construction follows projected and normalized priors for directional parameters \citep{Presnell1998,Antoniadis2004}. Proposition~\ref{prop:horseshoe_prior} shows the induced direction is angular central Gaussian, free of the global scale $\lambda_0$, with a hemisphere density continuous and positive off the coordinate axes and an envelope proportional to $\prod_j|\theta_j|^{-1}$ that spikes along sparse directions. The prior-positivity condition of Section~\ref{sec:theorems} then holds at any direction with no zero coordinate, covering the horseshoe used in our high-dimensional simulations and application whenever the pseudo-true direction is non-sparse. The half-Cauchy scales admit the inverse-gamma representation $\lambda_j^2\mid\xi_j\sim\mathrm{IG}(1/2,1/\xi_j)$, $\xi_j\sim\mathrm{IG}(1/2,1)$ \citep{MakalicSchmidt2016}, yielding closed-form updates.

\subsection{Normalized horseshoe prior, statement and proof}\label{proof:prop:horseshoe_prior}

\begin{proposition}[Induced prior on the hemisphere]\label{prop:horseshoe_prior}
Let $\nu$ follow \eqref{eq:horseshoe} and $\theta=\theta(\nu)$.
\begin{enumerate}[label=(\alph*),nosep]
\item\label{prop:hs_acg} Conditional on $(\lambda_1,\dots,\lambda_q)$, the law of $\theta$ does not depend on the global scale $\lambda_0$: the direction of $\mathcal N(0,s^2\Lambda)$ follows the angular central Gaussian law $\mathrm{ACG}(\Lambda)$, $\Lambda=\mathrm{diag}(\lambda_1^2,\dots,\lambda_q^2)$, for every $s>0$. Sparsity of the direction is therefore driven entirely by the local scales, and $\lambda_0$ is a purely computational parameter.
\item\label{prop:hs_density} The marginal law of $\theta$ has a density $\pi_{\mathrm{HS}}$ with respect to the surface measure on $\mathbb S^{q-1}_+$ satisfying, for a constant $C_q<\infty$,
\[
0<\pi_{\mathrm{HS}}(\theta)\le C_q\prod_{j=1}^q|\theta_j|^{-1}\quad\text{for all }\theta\text{ with }\theta_j\neq0\ \forall j,
\]
and $\pi_{\mathrm{HS}}$ is continuous and strictly positive on $\{\theta\in\mathbb S^{q-1}_+:\theta_j\neq0\ \forall j\}$. In particular, the local prior-positivity condition of Assumption~\ref{asmp:A_fixed_tau}\ref{subasmp:A_fixed_prior} holds at any $\theta^\star$ with no zero coordinate.
\end{enumerate}
\end{proposition}

\begin{proof}
Part (a). Fix $\Lambda=\mathrm{diag}(\lambda_1^2,\dots,\lambda_q^2)$ with all $\lambda_j>0$ and $s>0$, and let $\nu\sim\mathcal N(0,s^2\Lambda)$. For a Borel set $A$ of the unit sphere, passing to polar coordinates $\nu=ru$, $r>0$, $u\in\mathbb S^{q-1}$, with $d\nu=r^{q-1}\,dr\,\sigma_{q-1}(du)$,
\[
P\big(\nu/\|\nu\|\in A\big)
=\int_A\bigg[\int_0^\infty\frac{r^{q-1}}{(2\pi s^2)^{q/2}\det(\Lambda)^{1/2}}
\exp\Big\{-\frac{r^2\,u^\top\Lambda^{-1}u}{2s^2}\Big\}dr\bigg]\sigma_{q-1}(du).
\]
Substituting $t=r^2u^\top\Lambda^{-1}u/(2s^2)$, the inner integral equals $\tfrac12\Gamma(q/2)(2s^2)^{q/2}(u^\top\Lambda^{-1}u)^{-q/2}$; multiplying by the prefactor $(2\pi s^2)^{-q/2}\det(\Lambda)^{-1/2}$, every occurrence of $s$ cancels and the product is exactly $\frac{\Gamma(q/2)}{2\pi^{q/2}}\det(\Lambda)^{-1/2}(u^\top\Lambda^{-1}u)^{-q/2}$. Hence the direction $\nu/\|\nu\|$ has the angular central Gaussian density
\begin{equation}
f_{\mathrm{ACG}}(u;\Lambda)=\frac{\Gamma(q/2)}{2\pi^{q/2}}\,\det(\Lambda)^{-1/2}\,\big(u^\top\Lambda^{-1}u\big)^{-q/2},
\qquad u\in\mathbb S^{q-1},
\label{eq:acg_density}
\end{equation}
free of $s$. The hemisphere folding $\theta(\nu)=\mathrm{sign}(\nu_1)\nu/\|\nu\|$ maps $u$ and $-u$ to the same point and $f_{\mathrm{ACG}}(\cdot;\Lambda)$ is antipodally symmetric, so the induced density on $\mathbb S^{q-1}_+$ is $2f_{\mathrm{ACG}}$, still free of $s$. Taking $s=\lambda_0$ and integrating over its prior law leaves the conditional-on-$(\lambda_1,\dots,\lambda_q)$ law of $\theta$ unchanged, proving (a).

Part (b), upper bound. By the arithmetic--geometric mean inequality, for $u$ with all $u_j\neq0$,
\[
u^\top\Lambda^{-1}u=\sum_{j=1}^q\frac{u_j^2}{\lambda_j^2}
\ \ge\ q\prod_{j=1}^q\Big(\frac{u_j^2}{\lambda_j^2}\Big)^{1/q},
\qquad\text{hence}\qquad
\big(u^\top\Lambda^{-1}u\big)^{-q/2}\le q^{-q/2}\prod_{j=1}^q\frac{\lambda_j}{|u_j|}.
\]
Substituting into \eqref{eq:acg_density}, the factors $\prod_j\lambda_j$ cancel against $\det(\Lambda)^{-1/2}=\prod_j\lambda_j^{-1}$:
\[
f_{\mathrm{ACG}}(u;\Lambda)\ \le\ \frac{\Gamma(q/2)}{2\pi^{q/2}}\,q^{-q/2}\prod_{j=1}^q|u_j|^{-1}
\ =:\ \tfrac12 C_q\prod_{j=1}^q|u_j|^{-1},
\]
uniformly in $\Lambda$. Averaging over the half-Cauchy law of the local scales therefore yields
$\pi_{\mathrm{HS}}(u)=2\,\E_\Lambda[f_{\mathrm{ACG}}(u;\Lambda)]\le C_q\prod_j|u_j|^{-1}<\infty$ whenever no coordinate vanishes.

Part (b), positivity. Let $L=\{\lambda_j\in[1,2],\ j=1,\dots,q\}$, an event of probability $p_L=\big(\mathbb P(\mathcal C^+(0,1)\in[1,2])\big)^q>0$. On $L$, $\det(\Lambda)^{-1/2}\ge2^{-q}$ and $u^\top\Lambda^{-1}u\le\sum_ju_j^2=1$, so $f_{\mathrm{ACG}}(u;\Lambda)\ge\Gamma(q/2)\pi^{-q/2}2^{-q-1}$ for every $u\in\mathbb S^{q-1}$. Hence
$\pi_{\mathrm{HS}}(u)\ge2p_L\,\Gamma(q/2)\pi^{-q/2}2^{-q-1}>0$ everywhere on the hemisphere.

Part (b), continuity. Fix $u^\circ$ with all coordinates nonzero and a closed neighborhood $\mathcal U\subset\mathbb S^{q-1}_+$ of $u^\circ$ on which $\min_j|u_j|\ge c_\circ>0$. For each fixed $\Lambda$, $u\mapsto f_{\mathrm{ACG}}(u;\Lambda)$ is continuous on $\mathcal U$, and the family is dominated there by the constant $\tfrac12C_qc_\circ^{-q}$, which is integrable with respect to the (probability) law of $\Lambda$. By dominated convergence, $\pi_{\mathrm{HS}}$ is continuous on $\mathcal U$. Combining the three displays, at any $\theta^\star$ with no zero coordinate, $\pi_{\mathrm{HS}}$ is continuous and strictly positive on a neighborhood, so the local prior-positivity condition of Assumption~\ref{asmp:A_fixed_tau}\ref{subasmp:A_fixed_prior} holds (after composing with the smooth chart, which preserves continuity and positivity of densities).
\end{proof}

\section{Proofs for Section~\ref{sec:cp_model}}\label{sec:proofs_sec2}

\subsection{Formal identification statements}

The design conditions and identification result summarized in Section~\ref{sec:model_identification} are stated here. The heterogeneity term is $W$-absorbable if there exists $b_0\in\R^p$ with $X^\top\gamma_0=W^\top b_0$ almost surely, which holds automatically when $X$ is a subvector of $W$, and nowhere $W$-absorbable if for every $Z$-measurable event $E$ with $P_0(E)>0$ no $b\in\R^p$ satisfies $X^\top\gamma_0=W^\top b$ almost surely on $E$. In the absorbable case the flip acts through the flip map
\[
F(\beta,\gamma,\theta,\sigma^2)=(\beta+b_0,\,-\gamma,\,-\theta,\,\sigma^2),
\]
which replaces the subgroup by its complement and adjusts the global effect. Finally, let
\[
\mathcal E(\theta,\theta')=\big\{\mathbbm 1\{Z^\top\theta\ge0\}\neq\mathbbm 1\{Z^\top\theta'\ge0\}\big\}
\]
be the event that two candidate boundaries classify $Z$ differently.

\begin{assumption}[Design nondegeneracy]\label{asmp:design}
\leavevmode
\begin{enumerate}[label=(\roman*),nosep]
\item\label{asmp:design_nohyper} $P_0(Z^\top\theta=0)=0$ for every $\theta\in\mathbb S^{q-1}$, and $P_0\{\mathcal E(\theta,\theta')\}>0$ for every pair $\theta\neq\theta'\in\mathbb S^{q-1}$ with $\theta\ne-\theta'$.
\item\label{asmp:design_condcov} For every $Z$-measurable event $E$ with $P_0(E)>0$, the matrices $\E_{P_0}[WW^\top\mid E]$ and $\E_{P_0}[XX^\top\mid E]$ are finite and nonsingular; if the heterogeneity term is nowhere $W$-absorbable, the joint matrix $\E_{P_0}\big[(W^\top,X^\top)^\top(W^\top,X^\top)\mid E\big]$ is nonsingular as well.
\item\label{asmp:design_signal} $\gamma_0\neq0$ and $0<P_0(Z^\top\theta_0\ge0)<1$.
\end{enumerate}
\end{assumption}

Part \ref{asmp:design_nohyper} says that no candidate hyperplane carries probability mass and that distinct non-antipodal boundaries induce different partitions, which holds whenever $Z$ has $q-1$ coordinates with a joint Lebesgue density given the remaining one. Part \ref{asmp:design_condcov} prevents any $Z$-defined subpopulation from collapsing the regression design, and when $X$ is a subvector of $W$ the joint matrix is singular by construction, which is why the blockwise form is the right requirement in the absorbable case. Part \ref{asmp:design_signal} is the signal-strength condition under which the boundary is learnable, since at $\gamma_0=0$ the likelihood is constant in $\theta$. Section~\ref{sec:reporting_protocol} treats inference when the data cannot rule out this weakly identified regime.

\begin{proposition}[Identification]\label{prop:identification}
Let Assumption~\ref{asmp:design} hold, and let $\eta=(\beta,\gamma,\theta,\sigma^2)$, with nondegenerate gate $0<P_0(Z^\top\theta\ge0)<1$, and $\eta_0$ induce the same conditional mean and variance of $Y$ given $(W,X,Z)$ under model \eqref{eq:true_dgp}. (Directions with degenerate gates reduce the change-plane term to a global regression term and are observationally redundant reparameterizations of the regression block; in case (a) below the conclusion holds without excluding them.)
\begin{enumerate}[label=(\alph*),nosep]
\item If the heterogeneity term is nowhere $W$-absorbable, then $\eta=\eta_0$: the model is point identified on $\mathcal A\times\mathbb S^{q-1}$.
\item If the heterogeneity term is $W$-absorbable, then $\eta\in\{\eta_0,F(\eta_0)\}$: the identified set is the two-point flip orbit. If in addition $\theta_{0,1}\neq0$, the convention $\theta_1\,\mathrm{sign}(\theta_{0,1})>0$ selects a unique representative.
\end{enumerate}
\end{proposition}

Proposition~\ref{prop:identification} replaces the informal ``without loss of generality'' sign conventions common in this literature, which are valid only in case (b). In case (a), restricting $\theta$ to a half-sphere could exclude the truth.

\subsection{Proof of Proposition~\ref{prop:identification}}\label{proof:prop:identification}

Throughout the proof, ``a.s.''\ means $P_0$-almost surely, $d_\theta=d_\theta(Z)=\mathbbm 1\{Z^\top\theta\ge0\}$, and we restrict attention to candidate directions with nondegenerate gates, $0<P_0(Z^\top\theta\ge0)<1$: a direction with $d_\theta\in\{0,1\}$ a.s.\ reduces the change-plane term to a global regression term, so any such parameter is an observationally redundant reparameterization of the regression block rather than a competing subgroup model, and we exclude it from the parameter class (the true parameter is in the class by Assumption~\ref{asmp:design}\ref{asmp:design_signal}). The flip map $F$ and its witness $b_0$ are as defined above.

\begin{proof}
Matching conditional variances gives $\sigma^2=\sigma_0^2$ at once. Matching conditional means gives
\begin{equation}
W^\top(\beta-\beta_0)+X^\top\gamma\,d_\theta-X^\top\gamma_0\,d_{\theta_0}=0
\qquad\text{a.s.}
\label{eq:mean_match}
\end{equation}
For $a,b\in\{0,1\}$ define $E_{ab}=\{d_{\theta_0}=a,\ d_\theta=b\}$. On any $E_{ab}$ with $P_0(E_{ab})>0$, \eqref{eq:mean_match} reads
\begin{equation}
W^\top(\beta-\beta_0)+X^\top(\gamma b-\gamma_0 a)=0\qquad\text{a.s.\ on }E_{ab}.
\label{eq:event_relation}
\end{equation}
We use repeatedly: if $R$ is a random vector with $\E[RR^\top\mid E]$ nonsingular and $v^\top R=0$ a.s.\ on $E$ with $P_0(E)>0$, then $v^\top\E[RR^\top\mid E]v=0$, hence $v=0$.

Case (a), nowhere $W$-absorbable. Here Assumption~\ref{asmp:design}\ref{asmp:design_condcov} provides nonsingularity of the joint conditional second moment of $R=(W^\top,X^\top)^\top$ on every positive-probability $Z$-event, so \eqref{eq:event_relation} forces
\begin{equation}
\beta=\beta_0
\quad\text{and}\quad
\gamma b=\gamma_0 a
\qquad\text{whenever }P_0(E_{ab})>0.
\label{eq:joint_force}
\end{equation}
Suppose first $\theta=\theta_0$. Then $E_{00}=\{d_{\theta_0}=0\}$ and $E_{11}=\{d_{\theta_0}=1\}$ have positive probability by Assumption~\ref{asmp:design}\ref{asmp:design_signal} (note $P_0(Z^\top\theta_0=0)=0$ by \ref{asmp:design_nohyper}, so the two events partition the space up to a null set), and \eqref{eq:joint_force} with $(a,b)=(1,1)$ gives $\gamma=\gamma_0$; hence $\eta=\eta_0$. Next suppose $\theta=-\theta_0$. Then $d_\theta=1-d_{\theta_0}$ a.s., $E_{10}=\{d_{\theta_0}=1\}$ has positive probability, and \eqref{eq:joint_force} with $(a,b)=(1,0)$ gives $\gamma_0=0$, contradicting \ref{asmp:design_signal}. Finally suppose $\theta\notin\{\theta_0,-\theta_0\}$. By \ref{asmp:design_nohyper}, $P_0(E_{01}\cup E_{10})=P_0\{\mathcal E(\theta,\theta_0)\}>0$. If $P_0(E_{01})>0$, then \eqref{eq:joint_force} with $(a,b)=(0,1)$ gives $\gamma=0$ and $\beta=\beta_0$; substituting into \eqref{eq:mean_match} yields $X^\top\gamma_0d_{\theta_0}=0$ a.s., and on $\{d_{\theta_0}=1\}$ (positive probability) the $X$-block nonsingularity in \ref{asmp:design_condcov} forces $\gamma_0=0$, a contradiction. If instead $P_0(E_{10})>0$, then \eqref{eq:joint_force} with $(a,b)=(1,0)$ forces $\gamma_0=0$ directly, again a contradiction. Hence $\theta=\theta_0$ and the model is point identified. Finally, the parenthetical in the proposition holds in this case. If $d_\theta$ is degenerate, say $d_\theta=1$ a.s., then \eqref{eq:mean_match} reads $W^\top(\beta-\beta_0)+X^\top\gamma-X^\top\gamma_0 d_{\theta_0}=0$ a.s., and the joint nonsingularity forces $\beta=\beta_0$ and $\gamma=0$ on $\{d_{\theta_0}=0\}$, after which $X^\top\gamma_0=0$ a.s.\ on $\{d_{\theta_0}=1\}$ forces $\gamma_0=0$, contradicting \ref{asmp:design_signal}. The case $d_\theta=0$ a.s.\ is symmetric, so no degenerate-gate parameter reproduces the conditional law in case (a).

Case (b), $W$-absorbable, $X^\top\gamma_0=W^\top b_0$ a.s. Now only the blockwise nonsingularity in \ref{asmp:design_condcov} is available (the joint matrix is singular whenever $X$ is a subvector of $W$).

Step 1, $\theta=\theta_0$. On $E_{00}$ (positive probability), \eqref{eq:event_relation} reads $W^\top(\beta-\beta_0)=0$ a.s., so $\beta=\beta_0$ by $W$-block nonsingularity. On $E_{11}$, \eqref{eq:event_relation} then reads $X^\top(\gamma-\gamma_0)=0$ a.s., so $\gamma=\gamma_0$ by $X$-block nonsingularity. Hence $\eta=\eta_0$.

Step 2, $\theta=-\theta_0$. Then $d_\theta=1-d_{\theta_0}$ a.s. On $E_{10}=\{d_{\theta_0}=1\}$, \eqref{eq:event_relation} reads $W^\top(\beta-\beta_0)-X^\top\gamma_0=0$, i.e., by absorbability, $W^\top(\beta-\beta_0-b_0)=0$ a.s., forcing $\beta=\beta_0+b_0$. On $E_{01}=\{d_{\theta_0}=0\}$, \eqref{eq:event_relation} reads $W^\top(\beta-\beta_0)+X^\top\gamma=W^\top b_0+X^\top\gamma=0$ a.s., i.e., $X^\top(\gamma+\gamma_0)=0$ a.s.\ by absorbability, forcing $\gamma=-\gamma_0$. Hence $\eta=F(\eta_0)$, and conversely $F(\eta_0)$ reproduces the conditional law of $\eta_0$ by construction, so both orbit points are in the identified set.

Step 3, $\theta\notin\{\theta_0,-\theta_0\}$ is impossible. Suppose first that $P_0(E_{00})>0$ and $P_0(E_{11})>0$. As in Step 1, $E_{00}$ forces $\beta=\beta_0$. If $P_0(E_{01})>0$, the relation there reads $X^\top\gamma=0$, so $\gamma=0$; then $E_{11}$ reads $X^\top(0-\gamma_0)=0$, forcing $\gamma_0=0$, a contradiction. If $P_0(E_{10})>0$, the relation there reads $-X^\top\gamma_0=W^\top b_0=0$ a.s.\ on $E_{10}$, forcing $b_0=0$ by $W$-block nonsingularity, whence $X^\top\gamma_0=0$ a.s.\ and $\gamma_0=0$ by $X$-block nonsingularity, a contradiction. Since $P_0(E_{01}\cup E_{10})>0$ by \ref{asmp:design_nohyper}, one of the two contradictions obtains.

It remains to rule out the degenerate containment configurations in which $P_0(E_{00})=0$ or $P_0(E_{11})=0$ (which can occur when $Z$ contains an intercept). If both vanish, then $d_\theta=1-d_{\theta_0}$ a.s., so $\theta$ and $-\theta_0$ induce the same partition a.s.; but $\theta\neq-\theta_0$ and $\theta\neq\theta_0$, so \ref{asmp:design_nohyper} applied to the pair $(\theta,-\theta_0)$ gives $P_0\{\mathcal E(\theta,-\theta_0)\}>0$, a contradiction. Suppose $P_0(E_{00})=0$ but $P_0(E_{11})>0$, so $\{d_{\theta_0}=0\}\subseteq\{d_\theta=1\}$ a.s.\ and $E_{01}=\{d_{\theta_0}=0\}$ has positive probability by \ref{asmp:design_signal}. If also $P_0(E_{10})=0$, then $\{d_{\theta_0}=1\}\subseteq\{d_\theta=1\}$ a.s., whence $d_\theta=1$ a.s.: the gate of $\theta$ is degenerate, and $\theta$ is excluded from the parameter class by the convention fixed at the start of the proof. If $P_0(E_{10})>0$, then the argument of the preceding paragraph applies verbatim with $E_{00}$ replaced as follows: $E_{01}$ and $E_{11}$ give the two relations $W^\top(\beta-\beta_0)+X^\top\gamma=0$ and $W^\top(\beta-\beta_0)+X^\top(\gamma-\gamma_0)=0$ on their respective events, and $E_{10}$ gives $W^\top(\beta-\beta_0)-X^\top\gamma_0=0$; from $E_{10}$ and absorbability, $\beta=\beta_0+b_0$; substituting into the $E_{11}$ relation gives $W^\top b_0+X^\top(\gamma-\gamma_0)=X^\top(\gamma-\gamma_0+\gamma_0)=X^\top\gamma=0$ a.s.\ on $E_{11}$, so $\gamma=0$; the $E_{01}$ relation then reads $W^\top b_0=X^\top\gamma_0=0$ a.s.\ on $E_{01}$, forcing $\gamma_0=0$, a contradiction. In the mirror case $P_0(E_{11})=0$, $P_0(E_{00})>0$, we have $\{d_{\theta_0}=1\}\subseteq\{d_\theta=0\}$ a.s., and $E_{10}=\{d_{\theta_0}=1\}$ has positive probability by \ref{asmp:design_signal}. If $P_0(E_{01})=0$ then $d_\theta=0$ a.s.\ and $\theta$ is excluded as a degenerate gate. If $P_0(E_{01})>0$, then $E_{00}$ forces $\beta=\beta_0$ ($W$-block); $E_{01}$ then reads $X^\top\gamma=0$ a.s., forcing $\gamma=0$ ($X$-block); and $E_{10}$ reads $-X^\top\gamma_0=0$ a.s.\ on $E_{10}$, forcing $\gamma_0=0$, a contradiction. This exhausts all configurations and completes the proof.
\end{proof}

\begin{remark}[The hemisphere seam]\label{rem:identification_seam}
When $\theta_{0,1}=0$, the two points of the flip orbit both lie on the seam $\{\theta_1=0\}$ of the closed hemisphere, and the convention \eqref{eq:hemisphere} fails to select between them: coordinatewise posterior summaries of $\theta$ can then be corrupted by sign switching between the two equivalent representations. Sign-invariant functionals are unaffected: $\theta\theta^\top$, the posterior second-moment matrix $M$, its eigenstructure (hence $\hat\theta$ up to sign and $\lambda_{\max}(M)$), and the pair of membership probabilities $\{q(z),1-q(z)\}$ are identical at the two orbit points. In practice a posterior for $\theta_1$ concentrating near zero is the diagnostic that the seam case is in play, and reporting should then rely on the sign-invariant summaries of Section~\ref{sec:reporting_protocol}.
\end{remark}

\subsection{Proof of Proposition~\ref{prop:ess_invariance}}\label{proof:prop:ess_invariance}

\begin{proof}
Write $f(\theta)=\mathbbm 1\{\theta\in\mathbb S^{q-1}_+\}\exp\{\ell_T(\theta)\}$ for the unnormalized target with respect to the surface measure $\sigma_{q-1}$ on $\mathbb S^{q-1}$, and let
\[
T^1=\big\{(\theta,u):\theta\in\mathbb S^{q-1},\ u\in\mathbb S^{q-1},\ u^\top\theta=0\big\}
\]
be the unit tangent bundle, equipped with the measure $\mu(d\theta,du)=\sigma_{q-1}(d\theta)\,\sigma_{q-2,\theta}(du)$, where $\sigma_{q-2,\theta}$ is the surface measure on the unit sphere of the tangent space at $\theta$. Define the extended target $\rho(d\theta,du)\propto f(\theta)\,\mu(d\theta,du)$; its $\theta$-marginal is the target posterior.

Step 1 (the auxiliary draw). Given $\theta$, step 1 of the algorithm sets $u=\nu_\perp/\|\nu_\perp\|$ with $\nu\sim\mathcal N(0,I_q)$ and $\nu_\perp=(I_q-\theta\theta^\top)\nu$. Since $\nu_\perp$ is a standard Gaussian on the $(q-1)$-dimensional tangent space, its normalization is uniform on the unit sphere of that space, i.e., $u\mid\theta\sim\sigma_{q-2,\theta}/|\sigma_{q-2,\theta}|$. Hence step 1 is an exact conditional refresh of $u$ under $\rho$, and leaves $\rho$ invariant.

Step 2 (the geodesic rotation preserves $\mu$). For $\psi\in\R$ define
$R_\psi(\theta,u)=(\theta\cos\psi+u\sin\psi,\ -\theta\sin\psi+u\cos\psi)$,
a bijection of $T^1$ (the great-circle, or geodesic, flow). We claim $\mu\circ R_\psi^{-1}=\mu$. Represent $\mu$ (normalized) as the law of the pair obtained by Gram--Schmidt from two independent standard Gaussians $G_1,G_2$ in $\R^q$:
$\theta=G_1/\|G_1\|$, $u=P_\theta G_2/\|P_\theta G_2\|$ with $P_\theta=I_q-\theta\theta^\top$; this law is exactly uniform-$\theta$ with uniform tangent direction, as in Step 1. Condition on the two-dimensional plane $V=\mathrm{span}(G_1,G_2)$ (almost surely of dimension $2$): given $V$, the pair $(\theta,u)$ is an orthonormal $2$-frame of $V$. We claim its conditional law is the Haar-uniform law on the frame manifold of $V$, which is invariant under the right $O(2)$-action. Indeed, Gram--Schmidt is equivariant under left multiplication by orthogonal maps, and for any fixed $Q\in O(q)$ preserving $V$ the Gaussian pair $(QG_1,QG_2)$ has the same law as $(G_1,G_2)$ and spans the same plane; hence the conditional law of the frame given $V$ is invariant under the left action of the orthogonal group $O(V)\cong O(2)$ of the plane. That action is free and transitive on the orthonormal $2$-frames of $V$, so the conditional law is the unique invariant (Haar) probability on this $O(2)$-torsor; by compactness and unimodularity of $O(2)$, the Haar law is invariant under the right action as well. The map $R_\psi$ is precisely right multiplication of the frame $(\theta,u)$ by the rotation matrix $\big(\begin{smallmatrix}\cos\psi&-\sin\psi\\ \sin\psi&\cos\psi\end{smallmatrix}\big)\in O(2)$, and it leaves $V$ unchanged. Hence the conditional law of $(\theta,u)$ given $V$ is $R_\psi$-invariant, and integrating over $V$ proves the claim.

Step 3 (the angular slice update). Fix the current $(\theta,u)$ and consider the orbit $\psi\mapsto R_\psi(\theta,u)$, which is $2\pi$-periodic and passes through the current state at $\psi=0$. By Step 2, the extended target $\rho$ restricted to the orbit has density proportional to $\psi\mapsto f(\theta(\psi))$ with respect to the uniform (Lebesgue) measure on the angle; in particular the hemisphere indicator is part of the orbit density, vanishing where $\theta(\psi)\notin\mathbb S^{q-1}_+$. Steps 2--3 of the algorithm are then exactly the elliptical-slice transition of \citet[Section 2.3]{Murray2010} applied to this orbit density: a slice level $y$ uniform on $(0,f(\theta(0)))$; an initial angle $\varphi$ uniform on $(0,2\pi)$ defining the bracket $(\varphi-2\pi,\varphi)$, equivalently the circle cut at the initial proposal, and shrinkage toward $0$ upon rejection, where a proposal is rejected exactly when its orbit density lies below $y$ (the hemisphere violation being the case $f=0<y$). \citet{Murray2010} prove that this transition satisfies detailed balance with respect to the orbit density for any nonnegative integrable orbit density; the argument depends only on the density values along the orbit and is therefore unaffected by the presence of the indicator factor. (The randomized location of the bracket cut is essential to that proof: the probability of generating a given sequence of intermediate angles from the current point to an accepted point equals that of the reversed sequence only when the cut is exchangeable between the two endpoints. A deterministic cut at the antipode of the current state breaks this exchangeability, and the resulting transition is not exactly invariant; we verified the failure numerically on asymmetric bimodal circle targets, which is why the randomized bracket is part of the algorithm specification.)

Step 4 (termination). At $\varphi\in\{0,2\pi\}$ the proposal equals the current state, which satisfies $f(\theta)>0$ and $f(\theta)\ge y$ almost surely (as $y<f(\theta)$ with probability one). The current state lies in the open hemisphere $\theta_1>0$ almost surely (the seam $\{\theta_1=0\}$ is a null set of the target), so $\theta(\varphi)\in\mathbb S^{q-1}_+$ for all $|\varphi|$ small; combined with continuity of $\ell_T$ and of $\theta(\cdot)$, which give $\ell_T(\theta(\varphi))>\log y$ near $\varphi=0$, the acceptable set $\{\varphi:f(\theta(\varphi))\ge y\}$ contains an open neighborhood of $0$ (mod $2\pi$) almost surely. Each rejection strictly shrinks the bracket toward $0$ while keeping $0$ (mod $2\pi$) in its closure, and the redrawn angle is uniform on the current bracket; a standard argument \citep[Section 4.1]{Neal2003} shows the procedure terminates with probability one.

Combining: step 1 preserves $\rho$; steps 2--3 preserve $\rho$ (detailed balance on each orbit, orbits partitioning $T^1$ up to null sets, by Steps 2--3); hence the composition preserves $\rho$, and its $\theta$-marginal chain leaves the target posterior invariant.
\end{proof}

\subsection{Discussion and verification of the assumptions}\label{sec:pseudo_true_discussion}

This section collects the primitive content behind the boundary regularity conditions. We first record three observations on the uniqueness and anchoring conditions, Assumption~\ref{asmp:A_fixed_tau}\ref{subasmp:A_fixed_ident} and Assumption~\ref{asmp:V}\ref{subasmp:V_anchor}. We then verify the local items of Assumption~\ref{asmp:V} for a canonical design, describe a reproducible quadrature protocol for the two global items, and close with the boundary cases in which the conditions fail.

(i) Local uniqueness is automatic. Under Assumption~\ref{asmp:A_fixed_tau}\ref{subasmp:A_fixed_curv}, any stationary point of $K_\tau$ in $B_\delta(\tilde\eta^\star)$ is a strict local minimizer with positive-definite Hessian, and two distinct stationary points cannot coexist in a ball on which $\lambda_{\min}(I_\tau)\ge c_I>0$ (the gradient is strictly monotone along the connecting segment). What is genuinely assumed is therefore only global uniqueness and well-separation.

(ii) Global uniqueness connects to identification as $\tau\downarrow0$. By Lemma~\ref{lem:risk_localization}(a) and the Dini argument in the proof of Theorem~\ref{thm:tau_contraction_new} (Zone I), $K_\tau\to K_\infty$ uniformly on $\bar\Theta$. The minimizers of the hard-threshold risk $K_\infty$ over the nondegenerate-gate class are exactly the identified set of Proposition~\ref{prop:identification}, the singleton $\{\eta_0\}$ under the hemisphere convention with $\theta_{0,1}>0$. Consequently, for all sufficiently small $\tau$, every near-minimizer of $K_\tau$ lies in any prescribed neighborhood of $\eta_0$, and Assumption~\ref{asmp:V}\ref{subasmp:V_anchor} should be read as quantifying ``sufficiently small $\tau$'' for the analyst's chosen neighborhood, not as an additional identification condition. At fixed moderate $\tau$, by contrast, uniqueness of $\eta^\star(\tau)$ is a genuine regularity condition that can fail on contrived designs, for instance designs symmetric under a rotation that maps two candidate boundaries onto each other, exactly as uniqueness of pseudo-true parameters can fail in any misspecified model \citep{White1982}.

(iii) The conditions are checkable by quadrature. For a posited design, meaning a distribution for $(W,X,Z)$ together with the hard-threshold parameters, $K_\tau(\eta)$ is a smooth integral functional that can be evaluated by one-dimensional quadrature over the boundary score combined with Monte Carlo over the remaining covariates, and its minimizers located numerically. We recommend this as a design-stage diagnostic when the change-plane posterior is part of a prespecified analysis plan, since it verifies uniqueness and well-separation directly and produces the pseudo-true path $\tau\mapsto\eta^\star(\tau)$ against which the smoothing bias of Proposition~\ref{prop:eta_star_bias_rate_new} can be computed for the design at hand. The remainder of this section makes both the design class and the protocol explicit.

\paragraph{A canonical design class.} Consider the design in which $Z=(1,Z_2,\dots,Z_q)^\top$ has continuous part $(Z_2,\dots,Z_q)$ distributed with a density $p_Z$ that is bounded, continuous, and strictly positive on a bounded open set $\mathcal Z\subset\R^{q-1}$ and zero outside it, for instance a vector of standard Gaussian coordinates truncated to a box. Let $X\in\{0,1\}$ be a Bernoulli treatment indicator drawn independently of $Z$ with $0<P_0(X=1)<1$, let $W$ contain $Z$ and $X$ so that $X$ is a subvector of $W$ in the recommended absorbable specification, let the error be $\mathcal N(0,\sigma_0^2)$, let $\gamma_0\neq0$, and let $\theta_0$ lie in the interior of $\mathbb S^{q-1}_+$ with $\theta_{0,2:q}\neq0$, so that the gate is nondegenerate. We verify the local items of Assumption~\ref{asmp:V} in turn. The two global items, the anchor \ref{subasmp:V_anchor} and the seam \ref{subasmp:V_seam}, are not implied by this local structure alone and are addressed by the quadrature protocol below.

(i) holds by construction, since the error is $\mathcal N(0,\sigma_0^2)$.

(ii) Boundedness holds because $\mathcal Z$ is bounded and $X\in\{0,1\}$, so $\|W\|\vee\|X\|\vee\|Z\|\le C_B$ almost surely. For the oracle designs, the two subgroup events $\{U_{\theta_0}\ge0\}$ and $\{U_{\theta_0}<0\}$ are $Z$-measurable with probability in $(0,1)$ under the nondegenerate gate, so Assumption~\ref{asmp:design}\ref{asmp:design_condcov} supplies nonsingular conditional second moments of $(W,X)$ on each. The matrix $\E[(W^\top,X^\top D)(W^\top,X^\top D)^\top]$ is then positive definite for $D\in\{\mathbbm 1\{U_{\theta_0}\ge0\},1-\mathbbm 1\{U_{\theta_0}\ge0\}\}$, the interaction $DX$ contributing a direction independent of the main effects because the nondegenerate gate makes $D$ nonconstant.

(iii) Write $U_\theta=\theta_1+Z_{2:q}^\top\theta_{2:q}$. For $\theta$ near $\theta_0$ we have $\theta_{2:q}\neq0$, and the linear change of variables along the direction $\theta_{2:q}$ gives the score density
\[
f_\theta(u)=\frac{1}{\|\theta_{2:q}\|}\int_{\{z\in\mathcal Z:\;\theta_1+z^\top\theta_{2:q}=u\}} p_Z(z)\,d\mathcal H^{q-2}(z),
\]
the integral of $p_Z$ over the affine section of $\mathcal Z$ cut by $\{U_\theta=u\}$ against $(q-2)$-dimensional Hausdorff measure. Because $p_Z$ is continuous, bounded, and strictly positive on the open set $\mathcal Z$, and the boundary section at $u=0$ meets the interior of $\mathcal Z$ in a set of positive $(q-2)$-measure under the nondegenerate gate, $f_\theta$ is continuous in $(u,\theta)$ and satisfies $0<c_f\le f_\theta(u)\le C_f<\infty$ for $u$ near $0$ uniformly over $\theta$ near $\theta_0$. The uniform conditional moment continuity required in \ref{subasmp:V_margin} follows from the same representation, since $\E[g(W,X,Z)\mid U_\theta=u]$ is a ratio of section integrals of $g\,p_Z$ and $p_Z$, both continuous in $(u,\theta)$ by continuity and boundedness of $p_Z$ and of the polynomial $g$ on the bounded support, with $X$ entering through its independent Bernoulli moments.

(iv) Recall $\Pi_{\theta_0}=I_q-\theta_0\theta_0^\top$ projects onto the tangent space $T_{\theta_0}=\{v:v^\top\theta_0=0\}$, on which $\Pi_{\theta_0}v=v$. Fix a tangent direction $v\neq0$. Since $X$ is independent of $Z$ and $X^\top\gamma_0=\gamma_0 X$ with $\gamma_0\neq0$, the weight $\chi^2_{\mathrm B}(W,X)$ of \eqref{eq:chisq_weight} is strictly positive on the treated units $\{X=1\}$, which retain the full conditional law of $Z$ on the boundary section. Hence
\[
v^\top\Sigma_Z(\theta_0)\,v=\E\big[\chi^2_{\mathrm B}(W,X)\,(v^\top Z)^2\mid U_{\theta_0}=0\big]
\]
vanishes only if $v^\top Z=0$ almost surely on the section $S=\{z:z_1=1,\ z^\top\theta_0=0\}$ restricted to the support. On $S$ the map $z\mapsto v^\top z$ is affine, and the section carries positive $(q-2)$-density on a relatively open subset of $S$, so $v^\top z\equiv0$ on all of $S$. The linear part of $S$ is $\{w:w_1=0,\ w^\top\theta_0=0\}$, whose orthogonal complement is $\mathrm{span}\{e_1,\theta_0\}$, so $v=\alpha e_1+\beta\theta_0$ for scalars $\alpha,\beta$. Evaluating $v^\top z^\star=0$ at any $z^\star\in S$, where $z^\star_1=1$ and $z^{\star\top}\theta_0=0$, forces $\alpha=0$, and then $v^\top\theta_0=\beta=0$, so $v=0$, a contradiction. Thus $\Sigma_Z(\theta_0)$ is positive definite on $T_{\theta_0}$. The second matrix in \ref{subasmp:V_boundary_info} reduces, for the scalar Bernoulli treatment with $\gamma_0\neq0$, to the positive scalar $\E[X\mid U_{\theta_0}=0]=P_0(X=1)>0$. What is genuinely assumed here is the positivity and continuity of $p_Z$ on the section together with $\gamma_0\neq0$. Positive definiteness on the tangent space is then derived, not posited, and it is exactly the statement that no tangent direction is annihilated by the boundary section, which carries positive density mass. The section is $(q-2)$-dimensional here rather than $(q-1)$-dimensional because the intercept coordinate fixes $z_1=1$.

(v) The sign-disagreement wedge $\mathcal E(\theta,\theta_0)$ is contained in $\{|U_{\theta_0}|\le C_B\|\theta-\theta_0\|\}$, because $U_\theta-U_{\theta_0}=Z^\top(\theta-\theta_0)$ and disagreement forces $|U_{\theta_0}|\le|U_\theta-U_{\theta_0}|\le\|Z\|\,\|\theta-\theta_0\|$. The density bound $f_{\theta_0}\le C_f$ then gives the upper bound $P_0\{\mathcal E(\theta,\theta_0)\}\le 2C_fC_B\|\theta-\theta_0\|$. For the matching lower bound, a fixed positive fraction of the boundary section has $|Z^\top(\theta-\theta_0)|\ge c_1\|\theta-\theta_0\|$ with the disagreeing sign, because the projected covariate $Z^\top(\theta-\theta_0)/\|\theta-\theta_0\|$ has a nondegenerate law on the section by the positivity of $p_Z$, whence $P_0\{\mathcal E(\theta,\theta_0)\}\ge c_f c_1\|\theta-\theta_0\|$. These are the two-sided bounds of \ref{subasmp:V_wedge}, with $\|\theta-\theta_0\|$ comparable to the angle between the two directions near $\theta_0$.

(vii) The truncated default priors, that is Gaussian blocks for $(\beta,\gamma)$, an inverse-gamma for $\sigma^2$, and the uniform law on $\mathbb S^{q-1}_+$ for $\theta$, all restricted to $\bar\Theta$, are supported on $\bar\Theta$ and have a continuous positive density at the interior point $\tilde\eta_0$ in the sense of Assumption~\ref{asmp:A_fixed_tau}\ref{subasmp:A_fixed_prior}, so \ref{subasmp:V_prior} holds by the analyst's choice.

The two remaining items, the anchor \ref{subasmp:V_anchor} and the seam \ref{subasmp:V_seam}, constrain $K_\tau$ globally over $\mathcal A_\delta\times N_0$ and cannot be reduced to the density and moment structure above. We verify them numerically for a posited design.

\paragraph{Quadrature verification of the anchor and seam conditions.} The anchor \ref{subasmp:V_anchor} and the seam \ref{subasmp:V_seam} are population conditions on the computable functional $K_\tau$, checkable to any desired resolution for a posited design. We describe the protocol in five steps.

(a) The inner integral. Under \ref{subasmp:V_gauss} the true conditional law of $Y$ given $(W,X,Z)$ is $\mathcal N(\mu_0,\sigma_0^2)$ with $\mu_0=W^\top\beta_0+X^\top\gamma_0\,\mathbbm 1\{U_{\theta_0}\ge0\}$, so
\[
K_\tau(\eta)=\E_{P_0}\big[\,m\big(U_\theta/\tau,\ A(\eta;W,X),\ \mu_0\big)\big],
\qquad
A(\eta;W,X)=\big(W^\top\beta,\ W^\top\beta+X^\top\gamma,\ \sigma^2\big),
\]
where $A(\eta;W,X)$ collects the working mean pair and variance, $\mu_0$ is the fixed true mean, and
\[
m(s,A,\mu_0)=\E_{Y\sim\mathcal N(\mu_0,\sigma_0^2)}\Big[-\log\big\{(1-\Phi(s))\,\varphi_{\sigma}(Y-A_1)+\Phi(s)\,\varphi_{\sigma}(Y-A_2)\big\}\Big]
\]
is the cross entropy of the true Gaussian conditional at $\eta_0$ against the working probit mixture at $\eta$. The inner expectation is a one-dimensional Gaussian integral of the log of a two-component Gaussian mixture, and the substitution $Y=\mu_0+\sqrt2\,\sigma_0\,t$ turns it into a Gauss--Hermite quadrature $m\approx\pi^{-1/2}\sum_k \omega_k\big\{-\log[(1-\Phi(s))\varphi_{\sigma}(\mu_0+\sqrt2\sigma_0 t_k-A_1)+\Phi(s)\varphi_{\sigma}(\mu_0+\sqrt2\sigma_0 t_k-A_2)]\big\}$ with nodes $t_k$ and weights $\omega_k$, whose error is controlled by the node count.

(b) The outer integral. The integrand depends on $(W,X,Z)$ only through $X$, the working linear predictors, and the pair of scores $(U_\theta,U_{\theta_0})=(Z^\top\theta,Z^\top\theta_0)$. Two reductions are convenient. Conditioning on $U_\theta=u$ makes the outer expectation a one-dimensional quadrature over $u$ against $f_\theta(u)$, with an inner conditional Monte Carlo over the remaining covariate directions and over $(W,X)$. Alternatively one fixes a single large design sample $\{(W_j,X_j,Z_j)\}_{j=1}^N$ once and approximates $K_\tau(\eta)\approx N^{-1}\sum_{j=1}^N m(U_{\theta,j}/\tau,\,A(\eta;W_j,X_j),\,\mu_{0,j})$, reusing the same sample and the same Gauss--Hermite nodes for every $\eta$. These common random numbers make the excess risk $K_\tau(\eta)-K_\tau(\eta_0)$ and its shell minima low variance across $\eta$.

(c) The anchor. For each $\tau$ on a logarithmic grid, minimize $K_\tau$ over a coarse global grid on $\mathbb S^{q-1}_+\times\mathcal A_\delta$ and confirm that every near minimizer, meaning every grid point within a fixed tolerance of the grid minimum, lies in the declared neighborhood $\mathcal A_\delta\times N_0$. A confirmation across the $\tau$ grid certifies \ref{subasmp:V_anchor} for the design at the grid resolution.

(d) The seam. For the same $\tau$ grid, evaluate the normalized excess risk $(K_\tau(\eta)-K_\tau(\eta_0))/\tau$ on the shells $\{\eta\in\mathcal A_\delta\times N_0:\,C_1\tau\le\|\theta-\theta_0\|\le c_0\tau\}$ and record its infimum over each shell. A uniform positive lower bound $c_s$ across the grid certifies \ref{subasmp:V_seam} with that constant. The common random numbers of step (b) are essential here, since the certificate is a small difference divided by the small quantity $\tau$.

(e) The pseudo-true path. Trace $\tau\mapsto\eta^\star(\tau)$ by warm-started local minimization down the $\tau$ grid, initializing each solve at the previous solution, and report the boundary bias $\|\theta^\star(\tau)-\theta_0\|$ and the regression-block bias $\|a^\star(\tau)-a_0\|$ against the $O(\tau)$ bound of Proposition~\ref{prop:eta_star_bias_rate_new}. Regressing $\log\|\theta^\star(\tau)-\theta_0\|$ on $\log\tau$ estimates the boundary bias exponent, which the proposition predicts to equal one.

We ran this protocol on the canonical simulation design of Section~\ref{sec:simulations}, with a fixed design sample of size $2\times10^{5}$, $64$ Gauss--Hermite nodes, and a logarithmic grid of twelve values $\tau\in[0.002,0.2]$. The anchor holds at every $\tau$. The best risk value found by $64$ dispersed multistarts always lies in the declared neighborhood of $\eta_0$, and the best local minimum outside that neighborhood exceeds the global minimum by at least $0.032$ on the risk scale. The minimal seam ratio, $c_s=\min_\tau\ \inf_{C_1\tau\le\|\theta-\theta_0\|\le c_0\tau}(K_\tau(\eta)-K_\tau(\eta_0))/\tau$, is $0.029$ over the grid, attained at the smallest $\tau$ and bounded away from zero as \ref{subasmp:V_seam} requires. The pseudo-true path behaves as Proposition~\ref{prop:eta_star_bias_rate_new} predicts. The regression block converges with fitted log-log slope $1.04$, the bias of $\gamma^\star(\tau)$ falling from $1.9\times10^{-2}$ at $\tau=0.2$ to $1.8\times10^{-4}$ at $\tau=0.002$. The boundary shift is numerically negligible for this design, below $2.2\times10^{-3}$ radians over the whole grid and at the resolution of the design sample. This does not contradict the proposition, whose order $\tau$ statement is an upper bound. The leading constant of the boundary shift is driven by the asymmetry of the boundary-layer discrimination profile, and for the elliptically symmetric Gaussian design that asymmetry vanishes. A companion design that replaces the last boundary covariate by a centered exponential variable with unit variance makes the profile asymmetric, and the same protocol then exhibits the order $\tau$ boundary shift. For the skewed companion the shift falls from $6.1\times10^{-3}$ radians at $\tau=0.2$ to the design-sample resolution near $10^{-4}$ at the smallest grid values, with fitted log-log slope $0.86$, the anchor again holds with minimal gap $0.038$, and the seam ratio is bounded below by $0.027$ over the range where the normalized difference exceeds the Monte Carlo resolution of the protocol. Section~\ref{sec:simulations} uses this skewed companion wherever the smoothing shift itself is the object of study.

\paragraph{When the conditions fail.} Two boundary cases delimit the scope of the theory.

\begin{example}[Discrete effect modifiers]
Let every coordinate of $Z$ be binary, so $Z$ takes finitely many values and the score $U_\theta=Z^\top\theta$ has a discrete law. Then \ref{subasmp:V_margin} fails, since $U_\theta$ has no density and the wedge probability $P_0\{\mathcal E(\theta,\theta_0)\}$ jumps as $\theta$ crosses the finitely many hyperplanes separating the atoms of $Z$. The boundary direction is only set identified, its identified set being the collection of half-spaces that induce the same partition of the support of $Z$. The natural report is then the induced partition together with the membership summaries $q(z)=\Pi^{(\tau)}(z^\top\theta\ge0\mid\mathcal D_n)$ at the support points, as in the discrete-$Z$ analysis of Section~\ref{sec:analysis}. The anisotropic rate theory does not apply and the problem is the set-identified one flagged in Remark~\ref{rem:nonregular_rates}.
\end{example}

\begin{example}[Vanishing heterogeneity]
Let $\gamma_0=0$. Then $X^\top\gamma_0=0$ almost surely, the inner second moment in \eqref{eq:chisq_weight} vanishes, and $\chi^2_{\mathrm B}\equiv0$, so $\Sigma_Z(\theta_0)=0$ and \ref{subasmp:V_boundary_info} fails. The working likelihood is constant in $\theta$ and the boundary is unidentified, consistent with the signal-strength requirement of Assumption~\ref{asmp:design}\ref{asmp:design_signal}. This is not a defect to assume away but the regime that the reporting protocol of Section~\ref{sec:reporting_protocol} is designed to survive, since it bases the reporting decision on the heterogeneity evidence $\Pi(H_\delta\mid\mathcal D_n)$ and reports the boundary only when the data rule out $\gamma_0=0$.
\end{example}

A third, intermediate case is a design whose score density vanishes at the boundary, $f_{\theta_0}(u)\asymp|u|^{\alpha-1}$ with margin exponent $\alpha>1$. Here \ref{subasmp:V_margin} holds only in its degenerate form, the boundary information degrades to $I_{\tau,\vartheta\vartheta}\asymp\tau^{\alpha-2}$, and for $\alpha>2$ it no longer diverges, which changes the rescaling regime. This is the nonregular margin class of Remark~\ref{rem:nonregular_rates}, left open here.

\section{Proof of Proposition~\ref{prop:protocol_consistency}}\label{proof:prop:protocol_consistency}

\begin{proof}
Part (a). Suppose $|\Delta^\star(\tau)|>\delta$. By continuity of $\Delta$ there is $\rho>0$ with
$\{\gamma:\|\gamma-\gamma^\star\|\le\rho\}\subseteq\{\gamma:|\Delta(\gamma)|>\delta\}$, hence
$\Pi^{(\tau)}(H_\delta\mid\mathcal D_n)\ge\Pi^{(\tau)}(\|\eta-\eta^\star\|\le\rho\mid\mathcal D_n)\xrightarrow{P_0}1$
by Theorem~\ref{thm:mis_contraction_fixed_tau_new}\ref{thm1:contraction}. Since $p_{\mathrm{report}}<1$, the Bayes rule \eqref{eq:bayes_action} selects $\mathsf{a}_1$ with $P_0$-probability tending to one. If $|\Delta^\star(\tau)|<\delta$, the same argument applied to the open set $\{|\Delta(\gamma)|<\delta\}$ gives $\Pi^{(\tau)}(H_\delta\mid\mathcal D_n)\xrightarrow{P_0}0<p_{\mathrm{report}}$, so $\mathsf{a}_0$ is selected with probability tending to one.

Part (b). If $z^\top\theta^\star>0$, then $\{\theta:\|\theta-\theta^\star\|<z^\top\theta^\star/\|z\|\}\subseteq\{\theta:z^\top\theta>0\}$, so $q(z)=\Pi^{(\tau)}(z^\top\theta\ge0\mid\mathcal D_n)\ge\Pi^{(\tau)}(\|\theta-\theta^\star\|<z^\top\theta^\star/\|z\|\mid\mathcal D_n)\xrightarrow{P_0}1$; symmetrically $q(z)\xrightarrow{P_0}0$ when $z^\top\theta^\star<0$.

Part (c). Under Assumption~\ref{asmp:V} and the window, Proposition~\ref{prop:eta_star_bias_rate_new} gives $\|\gamma^\star(\tau_n)-\gamma_0\|\le C\tau_n\to0$ and $\|\theta^\star(\tau_n)-\theta_0\|\le C\tau_n\to0$. In the strict cases $|\Delta(\gamma_0)|\ne\delta$, continuity of $\Delta$ implies that for all large $n$ the sign of $|\Delta^\star(\tau_n)|-\delta$ agrees with that of $|\Delta(\gamma_0)|-\delta$ with a margin bounded away from zero, and the arguments of parts (a)--(b) apply verbatim along the sequence, with Theorem~\ref{thm:mis_contraction_fixed_tau_new} replaced by Theorem~\ref{thm:tau_contraction_new} for the contraction statements (the contraction radii $\tau_n+\log n/\sqrt n\to0$ are eventually smaller than any fixed $\rho$); likewise $q(z)\to\mathbbm 1\{z^\top\theta_0>0\}$ for any fixed $z$ off the true boundary.

For the knife-edge case, assume without loss of generality $\Delta(\gamma_0)=+\delta$ and let $\Delta$ be continuously differentiable at $\gamma_0$ with gradient $g_0=\nabla\Delta(\gamma_0)\ne0$ (for linear contrasts $\Delta(\gamma)=x_{\rm clin}^\top\gamma$, $g_0=x_{\rm clin}$ exactly). Write $v^2=\big(g_0^\top\,[\,\mathcal I_0^{-1}\,]_{\gamma\gamma}\,g_0\big)>0$, the limiting posterior variance of $\sqrt n\,\Delta(\gamma)$ implied by Theorem~\ref{thm:effect_bvm} and the delta method (the posterior of $\sqrt n(\gamma-\gamma_0)$ converges in total variation to a Gaussian, and $\Delta$ is smooth, so the posterior of $\sqrt n\{\Delta(\gamma)-\delta\}$ converges in total variation, in $P_0$-probability, to $\mathcal N(g_0^\top m_n,\,v^2)$ with $m_n=[\mathcal I_0^{-1}\Delta_{n,a}]_\gamma$). Hence
\[
\Pi^{(\tau_n)}\big(\Delta(\gamma)\ge\delta\mid\mathcal D_n\big)
=1-\Phi\Big(\frac{0-g_0^\top m_n}{v}\Big)+o_{P_0}(1)
=\Phi\Big(\frac{g_0^\top m_n}{v}\Big)+o_{P_0}(1),
\]
while the opposite tail $\Pi^{(\tau_n)}(\Delta(\gamma)\le-\delta\mid\mathcal D_n)\to0$ in $P_0$-probability since $-\delta$ is at distance $2\delta>0$ from the localization point. By Theorem~\ref{thm:effect_bvm}, $\Delta_{n,a}\Rightarrow\mathcal N(0,\mathcal I_0)$, so $g_0^\top m_n/v\Rightarrow\mathcal N(0,1)$, and by the continuous mapping theorem
$\Pi^{(\tau_n)}(H_\delta\mid\mathcal D_n)\Rightarrow\Phi(\mathcal Z)$, $\mathcal Z\sim\mathcal N(0,1)$,
a nondegenerate (uniform, in fact, since $\Phi(\mathcal Z)\sim\mathrm{Unif}(0,1)$) limit: at the clinical boundary the reporting probability converges to $P(\Phi(\mathcal Z)>p_{\mathrm{report}})=1-p_{\mathrm{report}}$, which is the protocol's exact analogue of size at the boundary of a one-sided test.
\end{proof}

\section{Relation to frequentist test-then-report workflows}
\label{sec:frequentist_comparison}

This section expands the comparison summarized at the end of the reporting protocol of Section~\ref{sec:reporting_protocol}.
A frequentist analogue tests for heterogeneity and then reports an estimated boundary upon rejection \citep[e.g.,][]{Fan2017,LeeSeoShin2011}, coupling two nonstandard steps, inference on $\gamma$ with a weakly identified nuisance boundary and inference on $\theta$ after selection. Treating the estimated boundary as fixed after a pretest yields anti-conservative uncertainty statements \citep{Kang2017}, which we document in Section~\ref{sec:simulations}. The program is not infeasible, since debiased post-selection inference is available \citep{GuoHe2021} and the policy-learning literature gives linear treatment rules with welfare guarantees under asymmetric costs \citep{KitagawaTetenov2018,AtheyWager2021,Zhao2012owl}. The Bayesian protocol instead offers coherence and propagation. Its heterogeneity evidence, effect summaries, boundary summary, stability diagnostic, and membership probabilities are all marginals of a single joint posterior, so boundary uncertainty propagates into every reported quantity. The threshold $\delta$ and the cost asymmetry enter the rule explicitly rather than through a test level, and the machinery applies unchanged in the high-dimensional settings of Section~\ref{sec:priors} where classical post-selection theory is unavailable.
\section{Proofs for Section~\ref{sec:posterior_contraction_BvM_fixed_tau_new}: Fixed Smoothing Scale}
\label{sec:proofs_fixed_tau_new}

\subsection{Fixed smoothing scale: conditions and statement}

Because remote regions of the noncompact $\Theta$ are pathological under misspecification and irrelevant to practice, we place the regression block on a fixed compact set,
\[
\bar\Theta \;=\; \mathcal K_a\times\mathbb S^{q-1}_+,
\qquad
\mathcal K_a=\big\{(\beta,\gamma,\sigma^2):\|\beta\|\le R_\beta,\ \|\gamma\|\le R_\gamma,\ \sigma^2\in[\underline\sigma^2,\overline\sigma^2]\big\},
\]
with $0<\underline\sigma^2<\overline\sigma^2<\infty$ and $R_\beta,R_\gamma<\infty$ fixed by the analyst, possibly astronomically large. Compact support licenses clean uniform deviation bounds, and extension to full support requires misspecification-robust tail tests \citep{Kleijn2006} orthogonal to our contribution.

The fixed-$\tau$ analysis rests on five standard regularity conditions, stated in full as Assumption~\ref{asmp:A_fixed_tau}. They ask that $\eta^\star$ be the unique, well-separated minimizer of $K_\tau$ in the interior of $\bar\Theta$, that the hemisphere admit a $C^2$ chart at $\theta^\star$, that the working information $I_\tau$ be continuous and nondegenerate near $\tilde\eta^\star$, that the prior have a continuous positive local density at $\tilde\eta^\star$, and that the design and errors be sub-Gaussian. The information and prior conditions are local and standard, and prior positivity holds for the truncated default priors of Section~\ref{sec:priors} and, by Proposition~\ref{prop:horseshoe_prior}, for the normalized-horseshoe prior whenever $\theta^\star$ has no zero coordinate. The substantive condition is the working-model identification requirement of a unique well-separated minimizer, and Supplementary Section~\ref{sec:pseudo_true_discussion} shows that local uniqueness follows from the information condition while relating global uniqueness to Proposition~\ref{prop:identification}.

\begin{theorem}[Fixed $\tau$: posterior contraction and misspecified Bernstein--von Mises]
\label{thm:mis_contraction_fixed_tau_new}
Let Assumption~\ref{asmp:A_fixed_tau} hold and fix $\tau>0$. Then:
\begin{enumerate}[label=(\alph*),nosep]
\item\label{thm1:contraction} there exists $M_0<\infty$ such that
\[
\Pi^{(\tau)}\Big(\|\tilde\eta-\tilde\eta^\star\|\ge M_0\sqrt{\log n/n}\,\Big|\,\mathcal D_n\Big)\xrightarrow[n\to\infty]{P_0}0;
\]
\item\label{thm1:bvm} with $h=\sqrt n(\tilde\eta-\tilde\eta^\star)$,
$\Delta_n=n^{-1/2}\sum_{i=1}^n\nabla\tilde\ell_\tau(\tilde\eta^\star;O_i)$ and $I_\tau^\star=I_\tau(\tilde\eta^\star)$,
\[
\sup_{B\in\mathcal B(\R^d)}
\Big|\Pi^{(\tau)}(h\in B\mid\mathcal D_n)-\mathcal N\big(I_\tau^{\star-1}\Delta_n,\ I_\tau^{\star-1}\big)(B)\Big|
\xrightarrow[n\to\infty]{P_0}0;
\]
\item\label{thm1:rate} consequently, for every sequence $M_n\to\infty$,
$\Pi^{(\tau)}\big(\|\tilde\eta-\tilde\eta^\star\|\ge M_nn^{-1/2}\mid\mathcal D_n\big)\xrightarrow{P_0}0$.
\end{enumerate}
\end{theorem}

The proof (Supplementary Section~\ref{sec:proofs_fixed_tau_new}) has four steps. A uniform Bernstein bound with the well-separation in \ref{subasmp:A_fixed_ident} and a denominator lower bound removes posterior mass at fixed distances. A peeling argument over dyadic annuli, using the quadratic separation from \ref{subasmp:A_fixed_curv}, contracts the posterior to radius $M_0\sqrt{\log n/n}$. A uniform local asymptotic normality expansion then delivers the Laplace approximation and the total-variation BvM of part~\ref{thm1:bvm}, and part~\ref{thm1:rate} follows because the limiting Gaussian mass beyond $M_n$ vanishes.

\subsection{Preliminaries}

Throughout this section $\tau>0$ is fixed; all constants may depend on $(\tau,p,r,q,\kappa,\mathcal K_a)$ but never on $n$ or on the parameter point unless indicated. Recall $\bar\Theta=\mathcal K_a\times\mathbb S^{q-1}_+$, the log-likelihood ratio
\[
r_\eta(O)=\log\frac{p_{\eta,\tau}(Y\mid W,X,Z)}{p_{\eta^\star,\tau}(Y\mid W,X,Z)},
\qquad
R_n(\eta)=\sum_{i=1}^n r_\eta(O_i),
\]
and that $\eta^\star$ is the (unique, interior, well-separated) minimizer of $K_\tau$ over $\bar\Theta$ fixed by Assumption~\ref{asmp:A_fixed_tau}\ref{subasmp:A_fixed_ident}. We write $\mathbb G_n g=n^{-1/2}\sum_{i=1}^n\{g(O_i)-\E_{P_0}g(O)\}$ for the empirical process and $\mathbb P_ng=n^{-1}\sum_ig(O_i)$. Since $\E_{P_0}r_\eta=-\{K_\tau(\eta)-K_\tau(\eta^\star)\}=:-\Delta K(\eta)\le0$,
\begin{equation}
R_n(\eta)=-n\,\Delta K(\eta)+\sqrt n\,\mathbb G_n r_\eta .
\label{eq:Rn_decomp}
\end{equation}

For reference we restate here the fixed-$\tau$ regularity conditions of Section~\ref{sec:posterior_contraction_BvM_fixed_tau_new}, in force throughout this section.

\begin{assumption}[Fixed $\tau$: regularity]
\label{asmp:A_fixed_tau}
\leavevmode
\begin{enumerate}[label=(\roman*),nosep]
\item \label{subasmp:A_fixed_ident}
$\eta^\star$ is the unique minimizer of $K_\tau$ over $\bar\Theta$, lies in the interior of $\bar\Theta$, and is well separated: for every $\rho>0$, $\inf\{K_\tau(\eta)-K_\tau(\eta^\star):\eta\in\bar\Theta,\ \|\eta-\eta^\star\|\ge\rho\}>0$.
\item \label{subasmp:A_fixed_chart}
$\theta^\star$ lies in the interior of $\mathbb S^{q-1}_+$ and $\theta(\cdot):U\to V\subset\mathbb S^{q-1}_+$ is a $C^2$ diffeomorphism onto a neighborhood $V$ of $\theta^\star$ with Jacobian of full rank $q-1$ on $U$, $\theta(\vartheta^\star)=\theta^\star$.
\item \label{subasmp:A_fixed_curv}
$I_\tau(\tilde\eta):=-\E_{P_0}\big[\nabla^2\tilde\ell_\tau(\tilde\eta;O)\big]$ exists, is continuous, and satisfies $\lambda_{\min}(I_\tau(\tilde\eta))\ge c_I>0$ on $B_\delta(\tilde\eta^\star)$ for some $\delta>0$.
\item \label{subasmp:A_fixed_prior}
The prior is supported on $\bar\Theta$; in local coordinates it has a density $\pi$ with respect to Lebesgue measure on a neighborhood of $\tilde\eta^\star$ that is continuous and positive at $\tilde\eta^\star$.
\item \label{subasmp:A_fixed_tails}
The design is sub-Gaussian, $\|u^\top(W^\top,X^\top,Z^\top)^\top\|_{\psi_2}\le\kappa\|u\|_2$ for all $u$, and the error is conditionally sub-Gaussian, $\E[e^{t\varepsilon}\mid W,X,Z]\le e^{\kappa^2t^2/2}$ for all $t\in\R$ almost surely.
\end{enumerate}
\end{assumption}

\subsection{Envelopes and moment bounds}

\begin{lemma}[Global envelopes on $\bar\Theta$]
\label{lem:subexp_fixed_tau_new}
Let Assumption~\ref{asmp:A_fixed_tau}\ref{subasmp:A_fixed_tails} hold. There exist random variables $F(O)$, $G(O)$, polynomial in $(|Y|,\|W\|,\|X\|,\|Z\|)$ of degree at most two, such that almost surely
\[
\sup_{\eta\in\bar\Theta}|r_\eta(O)|\le F(O),\qquad
|r_\eta(O)-r_{\eta'}(O)|\le\|\eta-\eta'\|\,G(O)\quad\text{for all }\eta,\eta'\in\bar\Theta,
\]
and $\|F\|_{\psi_1}+\|G\|_{\psi_1}<\infty$. In particular $\sup_{\eta\in\bar\Theta}\Var_{P_0}(r_\eta)<\infty$ and, for all $\eta,\eta'\in\bar\Theta$,
$\Var_{P_0}(r_\eta-r_{\eta'})\le\E_{P_0}[G^2]\,\|\eta-\eta'\|^2$.
\end{lemma}

\begin{proof}
Write $p_{\eta,\tau}=(1-\pi)f_0+\pi f_1$ with $\pi=\Phi(Z^\top\theta/\tau)$, $f_j=\varphi_\sigma(Y-m_j)$, $m_0=W^\top\beta$, $m_1=m_0+X^\top\gamma$. Since $\max\{(1-\pi)f_0,\pi f_1\}\le p_{\eta,\tau}\le\max\{f_0,f_1\}$,
\begin{equation}
\log p_{\eta,\tau}\ \ge\ \max\big\{\log(1-\pi)+\log f_0,\ \log\pi+\log f_1\big\},
\qquad
\log p_{\eta,\tau}\ \le\ \max\{\log f_0,\log f_1\}.
\label{eq:mixture_sandwich}
\end{equation}
For the lower bound, note that for each $z$ either $z^\top\theta\ge0$, in which case $\pi\ge1/2$ and $\log\pi\ge-\log2$, or $z^\top\theta<0$, in which case $1-\pi\ge1/2$; thus
$\log p_{\eta,\tau}\ge-\log2+\min\{\log f_0,\log f_1\}$. On $\mathcal K_a$, $|\log f_j|\le C(1+Y^2+\|W\|^2+\|X\|^2)$ with a constant depending only on $(\mathcal K_a)$; the gate weight never enters the envelope because the better-classified component always has weight at least $1/2$. Hence $\sup_{\bar\Theta}|\log p_{\eta,\tau}|\le C(1+Y^2+\|W\|^2+\|X\|^2)=:\tfrac12F(O)$ and the bound for $r_\eta$ follows by the triangle inequality.

For the Lipschitz modulus, $\eta\mapsto\log p_{\eta,\tau}(Y\mid\cdot)$ is continuously differentiable on the compact convex set $\bar{\mathcal K}=\mathcal K_a\times\{\theta:\|\theta\|\le1,\theta_1\ge0\}$ (extending the gate to $\|\theta\|\le1$, which is smooth since $\tau>0$), so the mean value inequality gives the claim with
$G(O)=\sup_{\xi\in\bar{\mathcal K}}\|\nabla_\eta\log p_{\xi,\tau}\|$. Differentiating,
\begin{align*}
\nabla_\beta\log p&=\sigma^{-2}\{Y-m_0-\omega X^\top\gamma\}W,\qquad
\nabla_\gamma\log p=\sigma^{-2}\omega\,(Y-m_1)X,\\
\partial_{\sigma^2}\log p&=-\tfrac{1}{2\sigma^2}+\tfrac{(1-\omega)(Y-m_0)^2+\omega(Y-m_1)^2}{2\sigma^4},\qquad
\nabla_\theta\log p=\tfrac{Z}{\tau}\,\phi(t)\tfrac{f_1-f_0}{p},
\end{align*}
where $\omega=\pi f_1/p\in[0,1]$ and $t=Z^\top\theta/\tau$. The first three are bounded by quadratic polynomials on $\mathcal K_a$. For the gate term, $p\ge\pi f_1$ and $p\ge(1-\pi)f_0$ give $|f_1-f_0|/p\le\pi^{-1}+(1-\pi)^{-1}$, and the Mills bound $\phi(t)\{\Phi(t)^{-1}+(1-\Phi(t))^{-1}\}\le C(1+|t|)$ yields
$\|\nabla_\theta\log p\|\le C\tau^{-1}(1+\|Z\|/\tau)\|Z\|\le C_\tau(\|Z\|+\|Z\|^2)$.
Thus $G$ is a quadratic polynomial in the data. Products of two sub-Gaussian variables are sub-exponential, so $\|F\|_{\psi_1},\|G\|_{\psi_1}<\infty$ under \ref{subasmp:A_fixed_tails} (note $Y$ is sub-Gaussian since $\beta_0,\gamma_0$ are fixed and $\varepsilon$ is conditionally sub-Gaussian). The variance bounds are immediate.
\end{proof}

\begin{lemma}[Local derivative envelopes]
\label{lem:score_hess_envelopes_fixed_tau_new}
Let Assumption~\ref{asmp:A_fixed_tau}\ref{subasmp:A_fixed_chart}, \ref{subasmp:A_fixed_tails} hold. There exist $\delta>0$ and random variables $F_1,F_2,F_3$, polynomial in the data, with $\E_{P_0}[F_1^2+F_2^2+F_3]<\infty$, such that almost surely
\[
\sup_{\|\tilde\eta-\tilde\eta^\star\|\le\delta}\|\nabla^k\tilde\ell_\tau(\tilde\eta;O)\|\le F_k(O),\qquad k=1,2,
\]
and $\|\nabla^2\tilde\ell_\tau(\tilde\eta;O)-\nabla^2\tilde\ell_\tau(\tilde\eta';O)\|\le F_3(O)\|\tilde\eta-\tilde\eta'\|$ for all $\tilde\eta,\tilde\eta'$ in that ball. Moreover differentiation may be exchanged with $\E_{P_0}$ on the ball, so that
$\nabla K_\tau=-\E_{P_0}\nabla\tilde\ell_\tau$ and $I_\tau=-\E_{P_0}\nabla^2\tilde\ell_\tau$ there, and $\E_{P_0}[\nabla\tilde\ell_\tau(\tilde\eta^\star;O)]=0$.
\end{lemma}

\begin{proof}
Choose $\delta$ so that the ball maps into the chart neighborhood $V$ and into the interior of $\mathcal K_a$. Each derivative of $\tilde\ell_\tau$ up to third order is, by the chain rule and the formulas in the previous proof, a finite sum of terms of the form (rational function of $\sigma^2$, bounded on the ball) $\times$ (responsibility factors $\omega,1-\omega\in[0,1]$ and their derivatives) $\times$ (polynomials in $Y,W,X$) $\times$ (derivatives of $\log\Phi$-type gate functions of $t$, bounded by $C(1+|t|)^3\le C_\tau(1+\|Z\|)^3$) $\times$ (chart Jacobian factors, bounded with bounded derivatives by \ref{subasmp:A_fixed_chart}). Hence each $F_k$ may be taken to be a polynomial in $(|Y|,\|W\|,\|X\|,\|Z\|)$ of fixed degree, and all polynomial moments are finite under \ref{subasmp:A_fixed_tails}. The interchange of derivative and expectation follows from the dominated convergence theorem applied to difference quotients, dominated by $F_1$ (respectively $F_2$, $F_3$), which are integrable. The first-order condition at the interior minimizer $\tilde\eta^\star$ then gives $\E_{P_0}\nabla\tilde\ell_\tau(\tilde\eta^\star;O)=-\nabla K_\tau(\tilde\eta^\star)=0$.
\end{proof}

\subsection{Separation and zonal posterior bounds}

\begin{lemma}[Quadratic separation near $\tilde\eta^\star$]
\label{lem:derived_separation_fixed_tau_new}
Let Assumption~\ref{asmp:A_fixed_tau} hold. There exist $\delta_1\in(0,\delta]$ and $c_K>0$ such that
\[
\Delta K(\eta(\tilde\eta))\ \ge\ c_K\|\tilde\eta-\tilde\eta^\star\|^2
\qquad\text{for all }\tilde\eta\in B_{\delta_1}(\tilde\eta^\star),
\]
and there exist $0<c_\ast\le C_\ast<\infty$ with
$c_\ast\|\tilde\eta-\tilde\eta^\star\|\le\|\eta(\tilde\eta)-\eta^\star\|\le C_\ast\|\tilde\eta-\tilde\eta^\star\|$ on the same ball.
\end{lemma}

\begin{proof}
By Lemma~\ref{lem:score_hess_envelopes_fixed_tau_new} the map $\tilde\eta\mapsto K_\tau(\eta(\tilde\eta))$ is twice continuously differentiable on $B_\delta(\tilde\eta^\star)$ with gradient vanishing at $\tilde\eta^\star$ and Hessian $I_\tau(\tilde\eta)$. Taylor's theorem with integral remainder gives, for $\tilde\eta\in B_\delta(\tilde\eta^\star)$ and $h=\tilde\eta-\tilde\eta^\star$,
\[
\Delta K(\eta(\tilde\eta))
=h^\top\Big[\int_0^1(1-t)\,I_\tau(\tilde\eta^\star+th)\,dt\Big]h
\ \ge\ \tfrac{c_I}{2}\|h\|^2,
\]
using Assumption~\ref{asmp:A_fixed_tau}\ref{subasmp:A_fixed_curv}; no mean-value point is invoked. The norm equivalence follows because $\theta(\cdot)$ is a $C^2$ diffeomorphism with full-rank Jacobian (Assumption~\ref{asmp:A_fixed_tau}\ref{subasmp:A_fixed_chart}): on a compact ball its differential and that of its inverse are bounded, and the regression-block coordinates are shared.
\end{proof}

\begin{lemma}[Fixed-distance zone]
\label{lem:tests_fixed_tau_new}
Let Assumption~\ref{asmp:A_fixed_tau} hold and fix $\varepsilon>0$. Let
$b(\varepsilon)=\tfrac14\inf\{\Delta K(\eta):\eta\in\bar\Theta,\ \|\eta-\eta^\star\|\ge\varepsilon\}>0$.
Then there are constants $c,C>0$ (depending on $\varepsilon$) such that
\[
P_0\Big(\sup_{\eta\in\bar\Theta:\ \|\eta-\eta^\star\|\ge\varepsilon}R_n(\eta)\ \ge\ -2nb(\varepsilon)\Big)\ \le\ Ce^{-cn}.
\]
\end{lemma}

\begin{proof}
Positivity of $b(\varepsilon)$ is Assumption~\ref{asmp:A_fixed_tau}\ref{subasmp:A_fixed_ident}. Let $\mathcal S=\{\eta\in\bar\Theta:\|\eta-\eta^\star\|\ge\varepsilon\}$, a compact set. Fix a $\rho$-net $\eta_1,\dots,\eta_N$ of $\mathcal S$ with $\rho=b(\varepsilon)/\{4\,\E_{P_0}G\}$; since $\bar\Theta$ is a fixed compact set, $N=N(\rho)<\infty$ does not depend on $n$. By Lemma~\ref{lem:subexp_fixed_tau_new}, for any $\eta\in\mathcal S$ with $\|\eta-\eta_j\|\le\rho$,
$R_n(\eta)\le R_n(\eta_j)+\rho\sum_iG(O_i)$.
On the event $E_n=\{\mathbb P_nG\le2\E G\}$, whose complement has probability at most $2e^{-cn}$ by Bernstein's inequality for sub-exponential variables, $\rho\sum_iG(O_i)\le nb(\varepsilon)/2$. For each fixed $j$, by \eqref{eq:Rn_decomp} and $\Delta K(\eta_j)\ge4b(\varepsilon)$,
\[
P_0\big(R_n(\eta_j)\ge-3nb(\varepsilon)\big)
\le P_0\big(\sqrt n\,\mathbb G_nr_{\eta_j}\ge nb(\varepsilon)\big)\le2e^{-cn},
\]
again by Bernstein's inequality, since $\|r_{\eta_j}-\E r_{\eta_j}\|_{\psi_1}\le2\|F\|_{\psi_1}$. A union bound over the fixed finite net and intersection with $E_n$ complete the proof.
\end{proof}

\begin{lemma}[Denominator lower bound]
\label{lem:denom_fixed_tau_new}
Let Assumption~\ref{asmp:A_fixed_tau} hold. There exist constants $c_\pi>0$ and $C<\infty$ such that, with probability tending to one,
\[
\int_{\bar\Theta}e^{R_n(\eta)}\,\Pi(d\eta)\ \ge\ c_\pi\,n^{-d/2}\,e^{-C\log n}.
\]
More precisely, for every $\epsilon>0$ there is $C_\epsilon$ with
$P_0\big(\int e^{R_n}d\Pi\le c_\pi n^{-d/2}e^{-C_\epsilon}\big)\le\epsilon$ for all large $n$.
\end{lemma}

\begin{proof}
Let $B_n=\{\tilde\eta:\|\tilde\eta-\tilde\eta^\star\|\le n^{-1/2}\}$. By Lemma~\ref{lem:score_hess_envelopes_fixed_tau_new}, for $\tilde\eta\in B_n$ a second-order Taylor expansion of $\tilde\ell_\tau(\cdot;O_i)$ at $\tilde\eta^\star$ gives
\[
R_n(\eta(\tilde\eta))\ \ge\ -\,\|\tilde\eta-\tilde\eta^\star\|\Big\|\sum_i\nabla\tilde\ell_\tau(\tilde\eta^\star;O_i)\Big\|-\tfrac12\|\tilde\eta-\tilde\eta^\star\|^2\sum_iF_2(O_i)
\ \ge\ -\Big\|\tfrac{1}{\sqrt n}\sum_i\nabla\tilde\ell_\tau(\tilde\eta^\star;O_i)\Big\|-\tfrac{1}{2}\,\mathbb P_nF_2 .
\]
The first norm is $O_{P_0}(1)$ by the central limit theorem ($\E\nabla\tilde\ell_\tau(\tilde\eta^\star)=0$, finite second moment by Lemma~\ref{lem:score_hess_envelopes_fixed_tau_new}); the second average converges to $\E F_2<\infty$. Hence $\inf_{B_n}R_n\ge-O_{P_0}(1)$, and
$\int e^{R_n}d\Pi\ge\Pi(B_n)\,e^{-O_{P_0}(1)}$. By Assumption~\ref{asmp:A_fixed_tau}\ref{subasmp:A_fixed_prior}, $\Pi(B_n)\ge c_\pi n^{-d/2}$ for a constant $c_\pi>0$ and all large $n$.
\end{proof}

\begin{lemma}[Peeling to radius $\sqrt{\log n/n}$]
\label{lem:peeling_fixed_tau_new}
Let Assumption~\ref{asmp:A_fixed_tau} hold. There exist $M_0<\infty$ and $\varepsilon_0>0$ such that
\[
\Pi^{(\tau)}\Big(M_0\sqrt{\tfrac{\log n}{n}}\le\|\tilde\eta-\tilde\eta^\star\|\le\varepsilon_0\ \Big|\ \mathcal D_n\Big)\xrightarrow[n\to\infty]{P_0}0 .
\]
\end{lemma}

\begin{proof}
Take $\varepsilon_0\le\delta_1\wedge(\delta_1/C_\ast)$ from Lemma~\ref{lem:derived_separation_fixed_tau_new} so that quadratic separation and the norm equivalence apply on the zone, and small enough that the chart covers it. Throughout this proof we pass between local and ambient coordinates via the norm equivalence of Lemma~\ref{lem:derived_separation_fixed_tau_new}, adjusting constants by factors of $c_\ast,C_\ast$ without further comment; in particular nets are constructed in the ambient metric in which the Lipschitz modulus $G$ of Lemma~\ref{lem:subexp_fixed_tau_new} operates. For $j\ge0$ let
$r_j=2^jM_0\sqrt{\log n/n}$ and consider the annuli
$\mathcal A_j=\{\tilde\eta:r_j\le\|\tilde\eta-\tilde\eta^\star\|<r_{j+1}\}$, $0\le j\le J_n$, where $J_n$ is the smallest index with $r_{J_n}\ge\varepsilon_0$.

Deviation bound on one annulus. Fix $j$ and abbreviate $r=r_j$. By Lemma~\ref{lem:derived_separation_fixed_tau_new}, $\Delta K\ge c_Kr^2$ on $\mathcal A_j$. We bound
$P_0\big(\sup_{\mathcal A_j}R_n\ge-\tfrac12nc_Kr^2\big)$.
Cover $\mathcal A_j$ by a $\rho_j$-net with $\rho_j=c_Kr^2/(16\,\E G)$; since $\mathcal A_j$ has diameter $\le 4r$ and lives in $\R^d$, we may choose the net with
$\log N_j\le d\log(Cr/\rho_j)\le d\log(C'/r)\le C''d\log n$
for all $r\ge M_0\sqrt{\log n/n}\ge n^{-1}$. On the event $E_n=\{\mathbb P_nG\le2\E G\}$ (probability $\ge1-2e^{-cn}$), the net approximation costs at most $n\rho_j\cdot2\E G\le\tfrac18nc_Kr^2$. At a net point $\tilde\eta_l\in\mathcal A_j$, by \eqref{eq:Rn_decomp} it suffices to bound
$P_0\big(\sqrt n\,\mathbb G_nr_{\eta_l}\ge\tfrac14nc_Kr^2\big)$.
By Lemma~\ref{lem:subexp_fixed_tau_new} with $\eta'=\eta^\star$ (note $r_{\eta^\star}\equiv0$),
$\Var(r_{\eta_l})\le\E[G^2]\,C_\ast^2(2r)^2=:V_1r^2$ and $\|r_{\eta_l}\|_{\psi_1}\le\|G\|_{\psi_1}C_\ast\,2r=:V_2r$.
Bernstein's inequality for sub-exponential variables \citep[Theorem 2.8.1]{Vershynin_2018} gives
\[
P_0\Big(\sqrt n\,\mathbb G_nr_{\eta_l}\ge\tfrac14nc_Kr^2\Big)
\le2\exp\Big\{-cn\min\Big(\tfrac{c_K^2r^4}{V_1r^2},\ \tfrac{c_Kr^2}{V_2r}\Big)\Big\}
=2\exp\{-c'n\min(r^2,r)\}=2e^{-c'nr^2}
\]
for $r\le1$. Hence, with a union bound over the net,
\[
P_0\Big(\sup_{\mathcal A_j}R_n\ge-\tfrac12nc_Kr_j^2,\ E_n\Big)
\le2\exp\{C''d\log n-c'nr_j^2\}
\le2\exp\{-\tfrac{c'}{2}nr_j^2\},
\]
provided $nr_j^2\ge nM_0^2\log n/n=M_0^2\log n\ge(2C''d/c')\log n$, i.e.\ for $M_0$ large.

Assembling the zones. On the event that all annulus bounds hold together with $E_n$ and the denominator bound of Lemma~\ref{lem:denom_fixed_tau_new} at level $\epsilon$, the posterior mass of $\bigcup_j\mathcal A_j$ is at most
\[
\sum_{j=0}^{J_n}\frac{e^{-\frac12nc_Kr_j^2}\,\Pi(\mathcal A_j)}{c_\pi n^{-d/2}e^{-C_\epsilon}}
\ \le\ C_\epsilon'\,n^{d/2}\sum_{j\ge0}\exp\{-\tfrac{c_K}{2}M_0^24^j\log n\}
\ \le\ C_\epsilon'\,n^{d/2}\cdot2\,n^{-c_KM_0^2/2}\ \longrightarrow\ 0
\]
for $M_0^2>d/c_K$, where we used $\Pi(\mathcal A_j)\le1$ and that the exponents are geometrically dominated by the $j=0$ term. The exceptional events have probability at most
$\epsilon+2e^{-cn}+\sum_j2e^{-c'nr_j^2/2}\le\epsilon+o(1)$, and $\epsilon$ was arbitrary.
\end{proof}

\subsection{Local asymptotic normality and the Bernstein--von Mises theorem}

\begin{lemma}[Uniform LAN on $\sqrt{\log n/n}$ balls]
\label{lem:lan_fixed_tau_new}
Let Assumption~\ref{asmp:A_fixed_tau} hold, set
$\Delta_n=n^{-1/2}\sum_i\nabla\tilde\ell_\tau(\tilde\eta^\star;O_i)$ and $I^\star=I_\tau(\tilde\eta^\star)$, and let $\mathcal H_n=\{h\in\R^d:\|h\|\le M_0\sqrt{\log n}\}$. Then
\[
\sup_{h\in\mathcal H_n}\bigg|
\sum_{i=1}^n\Big\{\tilde\ell_\tau\big(\tilde\eta^\star+\tfrac{h}{\sqrt n};O_i\big)-\tilde\ell_\tau(\tilde\eta^\star;O_i)\Big\}
-h^\top\Delta_n+\tfrac12h^\top I^\star h\bigg|=o_{P_0}(1),
\]
and $\Delta_n\Rightarrow\mathcal N(0,J_\tau)$ with $J_\tau=\E_{P_0}[\nabla\tilde\ell_\tau(\tilde\eta^\star;O)\nabla\tilde\ell_\tau(\tilde\eta^\star;O)^\top]$.
\end{lemma}

\begin{proof}
The central limit theorem applies by Lemma~\ref{lem:score_hess_envelopes_fixed_tau_new} (zero mean, finite second moment). For the expansion, a second-order Taylor formula with integral remainder gives, for $h\in\mathcal H_n$,
\[
\sum_i\{\cdots\}-h^\top\Delta_n+\tfrac12h^\top I^\star h
=\tfrac12h^\top\Big[\mathbb P_n\nabla^2\tilde\ell_\tau(\tilde\eta^\star;\cdot)+I^\star\Big]h
+\tfrac12h^\top\Big[\int_0^12(1-t)\big\{\mathbb A_n(t,h)-\mathbb P_n\nabla^2\tilde\ell_\tau(\tilde\eta^\star;\cdot)\big\}dt\Big]h,
\]
where $\mathbb A_n(t,h)=\mathbb P_n\nabla^2\tilde\ell_\tau(\tilde\eta^\star+th/\sqrt n;\cdot)$. For the first term,
$\|\mathbb P_n\nabla^2\tilde\ell_\tau(\tilde\eta^\star;\cdot)+I^\star\|=O_{P_0}(n^{-1/2})$ by Chebyshev's inequality and $\E F_2^2<\infty$, so its contribution is $O_{P_0}(\|h\|^2n^{-1/2})=O_{P_0}(\log n\cdot n^{-1/2})=o_{P_0}(1)$ uniformly on $\mathcal H_n$. For the second term, the Hessian Lipschitz envelope of Lemma~\ref{lem:score_hess_envelopes_fixed_tau_new} gives
$\|\mathbb A_n(t,h)-\mathbb P_n\nabla^2\tilde\ell_\tau(\tilde\eta^\star;\cdot)\|\le\mathbb P_nF_3\cdot\|h\|/\sqrt n$,
so its contribution is at most $\tfrac12\|h\|^3n^{-1/2}\,\mathbb P_nF_3=O_{P_0}((\log n)^{3/2}n^{-1/2})=o_{P_0}(1)$ uniformly on $\mathcal H_n$.
\end{proof}

\begin{proof}[Proof of Theorem~\ref{thm:mis_contraction_fixed_tau_new}]
Step 1 (localization). By Lemmas~\ref{lem:tests_fixed_tau_new} (applied with $\varepsilon=\varepsilon_0$) and \ref{lem:denom_fixed_tau_new}, the posterior mass of $\{\|\eta-\eta^\star\|\ge\varepsilon_0\}$ is bounded by
$e^{-2nb(\varepsilon_0)}/(c_\pi n^{-d/2}e^{-C_\epsilon})\to0$
on events of probability tending to $1-\epsilon$; together with the norm equivalence of Lemma~\ref{lem:derived_separation_fixed_tau_new} and Lemma~\ref{lem:peeling_fixed_tau_new}, this proves part \ref{thm1:contraction} with the stated $M_0$ and, in particular,
$\Pi^{(\tau)}(\tilde\eta\notin\tilde B_n\mid\mathcal D_n)\xrightarrow{P_0}0$ for
$\tilde B_n=\{\|\tilde\eta-\tilde\eta^\star\|\le M_0\sqrt{\log n/n}\}$.

Step 2 (Laplace approximation on $\tilde B_n$). Let $\Pi_{n,\mathrm{loc}}$ denote the posterior conditioned on $\tilde B_n$, viewed as a law for $h=\sqrt n(\tilde\eta-\tilde\eta^\star)$, with Lebesgue density proportional to
$\pi(\tilde\eta^\star+h/\sqrt n)\exp\{L_n(h)\}$ on $\mathcal H_n$, where $L_n(h)$ is the log-likelihood-ratio in $h$. Let $Q_n$ be $\mathcal N(I^{\star-1}\Delta_n,I^{\star-1})$ conditioned on $\mathcal H_n$. By Lemma~\ref{lem:lan_fixed_tau_new} and the continuity and positivity of $\pi$ at $\tilde\eta^\star$ (Assumption~\ref{asmp:A_fixed_tau}\ref{subasmp:A_fixed_prior}),
\[
\sup_{h\in\mathcal H_n}\Big|\log\frac{d\Pi_{n,\mathrm{loc}}}{dQ_n}(h)-c_n\Big|=o_{P_0}(1)
\]
for a normalizing constant $c_n$; a standard argument (integrate the multiplicative error against either normalized measure) converts this into
$\|\Pi_{n,\mathrm{loc}}-Q_n\|_{\mathrm{TV}}=o_{P_0}(1)$. Since $\Delta_n=O_{P_0}(1)$ and $\lambda_{\min}(I^\star)\ge c_I$, the unconditioned Gaussian places mass $1-o_{P_0}(1)$ on $\mathcal H_n$, so conditioning on $\mathcal H_n$ changes the Gaussian by $o_{P_0}(1)$ in total variation; by Step 1 the same holds for the posterior. The triangle inequality yields part \ref{thm1:bvm}.

Step 3 (rate upgrade). For any $M_n\to\infty$, the limiting Gaussian mass of $\{\|h\|\ge M_n\}$ is $o_{P_0}(1)$ because $\|I^{\star-1}\Delta_n\|=O_{P_0}(1)$; part \ref{thm1:bvm} transfers this to the posterior, proving part \ref{thm1:rate}.
\end{proof}

We restate the lemma.

\begin{lemma}[Existence of the pseudo-true parameter]
\label{lem:exist_pseudotrue_fixed_tau_new}
Let Assumption~\ref{asmp:A_fixed_tau}\ref{subasmp:A_fixed_tails} hold and fix $\tau>0$. Then $K_\tau$ is finite and continuous on $\Theta$, and $\argmin_{\eta\in\bar\Theta}K_\tau(\eta)$ is nonempty and compact. The restriction to $\bar\Theta$ is not removable. Along sequences with $\|\gamma\|\to\infty$ at fixed $(\beta,\sigma^2,\theta)$ the working risk remains bounded, since inflating one mixture component's mean merely empties that component at the bounded cost of its gate weight, so $K_\tau$ is not coercive on $\Theta$ and an unconstrained minimizer need not exist.
\end{lemma}

\begin{proof}[Proof of Lemma~\ref{lem:exist_pseudotrue_fixed_tau_new}]\phantomsection\label{proof:lem:exist_pseudotrue_fixed_tau_new}
Finiteness and continuity. By \eqref{eq:mixture_sandwich} and the argument following it (with $\mathcal K_a$ replaced by an arbitrary compact subset of $\mathcal A$), for every compact $\mathcal C\subset\Theta$ there is an integrable envelope dominating $\sup_{\eta\in\mathcal C}|\log p_{\eta,\tau}|$; finiteness and, by dominated convergence along sequences, continuity of $K_\tau$ on $\Theta$ follow.

Existence on $\bar\Theta$. The set $\bar\Theta=\mathcal K_a\times\mathbb S^{q-1}_+$ is compact and $K_\tau$ is continuous on it, so the minimum is attained and the argmin set, being a closed subset of a compact set, is compact.

Failure of coercivity on $\Theta$. Fix $(\beta,\sigma^2,\theta)$ and let $\|\gamma\|\to\infty$. Since $p_{\eta,\tau}\ge(1-\pi)f_0$ pointwise,
\[
K_\tau(\eta)\ \le\ \E_{P_0}\big[-\log\{1-\pi_{\theta,\tau}(Z)\}\big]
+\frac{\E_{P_0}[(Y-W^\top\beta)^2]}{2\sigma^2}+\tfrac12\log(2\pi\sigma^2),
\]
which is finite under Assumption~\ref{asmp:A_fixed_tau}\ref{subasmp:A_fixed_tails}, using $-\log(1-\Phi(t))\le C(1+t^2)$ and sub-Gaussianity of $Z$ and $Y$, and free of $\gamma$. Hence $K_\tau$ remains bounded along $\|\gamma\|\to\infty$: inflating one mixture component's mean merely empties that component, at the bounded population cost of its gate weight. (When the heterogeneity term is absorbable the same mechanism caps the risk along compensating $(\beta,\gamma)$ directions.) Minimizers over the noncompact $\Theta$ therefore need not exist, which is why \eqref{eq:pseudo_true} is defined over $\bar\Theta$.
\end{proof}

\subsection{Posterior shape and frequentist coverage at fixed smoothing}

\begin{remark}[Posterior shape versus frequentist coverage at fixed $\tau$]
\label{rem:sandwich_note}
Theorem~\ref{thm:mis_contraction_fixed_tau_new} describes the asymptotic posterior shape. Under misspecification the posterior mean obeys the sandwich law $\sqrt n(\hat{\tilde\eta}-\tilde\eta^\star)\Rightarrow\mathcal N(0,I_\tau^{\star-1}V_\tau I_\tau^{\star-1})$ with $V_\tau=\Var_{P_0}[\nabla\tilde\ell_\tau(\tilde\eta^\star;O)]$, so credible sets for the $\tau$-regularized rule are exactly calibrated only under the information equality $V_\tau=I_\tau^\star$ \citep{White1982,Kleijn2012,Mueller2013}, with sandwich rescaling or bagged posteriors \citep{Huggins2021,Syring2019} available otherwise. This caveat is softened here, because misspecification lives entirely in an $O(\tau)$ neighborhood of the boundary and Theorem~\ref{thm:effect_bvm} restores the information equality for the regression block as $\tau_n\downarrow0$, so raw credible sets for $\gamma$ are asymptotically exact on the window \eqref{eq:window}, and the reporting protocol of Section~\ref{sec:reporting_protocol} decides on posterior probabilities of interval hypotheses about $\gamma$, which inherit this calibration.
\end{remark}
\section{Proofs for Section~\ref{sec:asymptotic_tau_new}: Vanishing Smoothing Scale}
\label{sec:proofs_tau_n_new}

Throughout this section Assumption~\ref{asmp:V} is in force. Constants $c,C$ may change from line to line and may depend on $(p,r,q,\sigma_0^2,C_B,c_f,C_f,t_0,\mathcal A_\delta,N_0)$ and on the chart, but never on $(\tau,n)$ or the parameter point unless indicated. We use the abbreviations
\[
U=U_\theta=Z^\top\theta,\quad U_0=Z^\top\theta_0,\quad t=U/\tau,\quad
\pi=\Phi(t),\quad d_0=\mathbbm 1\{U_0\ge0\},
\]
and, for a parameter $\eta=(a,\theta)$ with $a=(\beta,\gamma,\sigma^2)$,
\[
m_0=W^\top\beta,\quad m_1=m_0+X^\top\gamma,\quad
f_j(y)=\varphi_\sigma(y-m_j),\quad
p_t(y)=(1-\Phi(t))f_0(y)+\Phi(t)f_1(y),\quad
q_t(y)=\frac{f_1(y)-f_0(y)}{p_t(y)} .
\]
Under Assumption~\ref{asmp:V}\ref{subasmp:V_gauss}, conditionally on $(W,X,Z)$ the true law of $Y$ is the Gaussian density $f_{d_0}^0$, meaning $f_{d_0}$ evaluated at the true parameter $a_0$. By Assumption~\ref{asmp:V}\ref{subasmp:V_bounded} and compactness of $\mathcal A_\delta$, the quantities $|m_0|,|m_1|,|X^\top\gamma|,\sigma^2$ are uniformly bounded and $\sigma^2$ is bounded away from zero on the relevant parameter range; we use this repeatedly without comment.

For reference we restate here the vanishing-$\tau$ boundary regularity conditions of Section~\ref{sec:asymptotic_tau_new}, in force throughout this section.

\begin{assumption}[Vanishing $\tau$: boundary regularity]
\label{asmp:V}
\leavevmode
\begin{enumerate}[label=(\roman*),nosep]
\item \label{subasmp:V_gauss} $\varepsilon\mid(W,X,Z)\sim\mathcal N(0,\sigma_0^2)$.
\item \label{subasmp:V_bounded} $\|W\|\vee\|X\|\vee\|Z\|\le C_B<\infty$ almost surely, and $\E[(W^\top,X^\top D)(W^\top,X^\top D)^\top]$ is nonsingular for $D\in\{\mathbbm 1\{U_{\theta_0}\ge0\},\,1-\mathbbm 1\{U_{\theta_0}\ge0\}\}$.
\item \label{subasmp:V_margin} There is $t_0>0$ such that for every $\theta\in N_0$ the score $U_\theta$ has a density $f_\theta$ on $[-t_0,t_0]$ with
$0<c_f\le f_\theta(u)\le C_f<\infty$, $(u,\theta)\mapsto f_\theta(u)$ continuous on $[-t_0,t_0]\times N_0$; moreover the conditional moments
$u\mapsto\E[g(W,X,Z)\mid U_\theta=u]$ are continuous at $u=0$ uniformly in $\theta\in N_0$ for $g$ any polynomial of degree at most $4$ in $(W,X,Z)$.
\item \label{subasmp:V_boundary_info} The matrix $\Sigma_Z(\theta_0)=\E\big[\chi^2_{\mathrm B}(W,X)\,\Pi_{\theta_0}ZZ^\top\Pi_{\theta_0}\mid U_{\theta_0}=0\big]$ is positive definite on the tangent space of $\mathbb S^{q-1}$ at $\theta_0$, and $\E[XX^\top\mathbbm 1\{X^\top\gamma_0\neq0\}\mid U_{\theta_0}=0]$ is positive definite, where $\Pi_{\theta_0}=I_q-\theta_0\theta_0^\top$ and $\chi^2_{\mathrm B}(w,x)>0$ is the boundary discrimination weight defined in \eqref{eq:chisq_weight} in the main text.
\item \label{subasmp:V_wedge} There are constants $0<c_w\le C_w$ with
$c_w\|\theta-\theta_0\|\le P_0\{\mathcal E(\theta,\theta_0)\}\le C_w\|\theta-\theta_0\|$ for all $\theta\in N_0$.
\item \label{subasmp:V_anchor} For all sufficiently small $\tau$, every minimizer of $K_\tau$ over $\bar\Theta$ lies in $\mathcal A_\delta\times N_0$.
\item \label{subasmp:V_prior} The prior is supported on $\bar\Theta$ and satisfies Assumption~\ref{asmp:A_fixed_tau}\ref{subasmp:A_fixed_prior} at $\tilde\eta_0$; $\eta_0$ lies in the interior of $\bar\Theta$.
\item \label{subasmp:V_seam} (Seam identifiability.) For every pair of constants $0<C_1<c_0<\infty$ there exist $c_s>0$ and $\tau_s>0$ such that, for all $\tau\in(0,\tau_s]$,
\[
K_\tau(\eta)\ \ge\ K_\tau(\eta_0)+c_s\,\tau
\qquad\text{for all }\eta\in\mathcal A_\delta\times N_0\text{ with }C_1\tau\le\|\theta-\theta_0\|\le c_0\tau .
\]
\end{enumerate}
\end{assumption}

The boundary discrimination weight appearing in \ref{subasmp:V_boundary_info} is
\begin{equation}
\chi^2_{\mathrm B}(w,x)=\int_{-\infty}^{\infty}\phi(s)^2\,
\E_{Y\sim f_{d(s)}}\Bigg[\bigg(\frac{f_1(Y)-f_0(Y)}{p_s(Y)}\bigg)^{\!2}\Bigg]\,ds,
\qquad
p_s=(1-\Phi(s))f_0+\Phi(s)f_1,
\label{eq:chisq_weight}
\end{equation}
where $d(s)=\mathbbm 1\{s\ge0\}$ and $f_0,f_1$ are the $\mathcal N(w^\top\beta_0,\sigma_0^2)$ and $\mathcal N(w^\top\beta_0+x^\top\gamma_0,\sigma_0^2)$ densities. The inner expectation is the second moment, under the true component at signed boundary distance $s$, of the normalized discrepancy between the two mixture components at gate level $\Phi(s)$, strictly positive exactly when $x^\top\gamma_0\neq0$.

\paragraph{Further remarks.}

\begin{remark}[Comparison with hard-threshold and smoothed-frequentist rates]
\label{rem:rate_comparison}
For the hard-threshold objective the boundary is estimable at the $n$ rate, with a multivariate compound Poisson argmin limit established by \citet{Kang2025}. Smoothed frequentist objectives tolerate smoothing scales as small as order $n^{-1}$ up to logarithmic factors, at which point the smoothed rate nearly recovers the $n$ rate \citep[e.g.,][]{SeoLinton2007,Mukherjee2023}. Our window has the same lower edge, $n^{-1}\log n$ up to the power of the logarithm, and as the schedule exponent $\varrho$ approaches one the posterior boundary accuracy $\tau_n$ approaches the $n$ rate. What the smoothing buys at this edge is a jointly Gaussian, correctly centered regression-block posterior together with a boundary posterior of Gaussian shape at its own scale, in place of nonstandard limits that are difficult to simulate and to summarize. What it costs is the smoothing shift of Remark~\ref{rem:boundary_shift}, which is the object we report.
\end{remark}

\begin{remark}[Beyond generic boundary geometry]
\label{rem:nonregular_rates}
Assumption~\ref{asmp:V}\ref{subasmp:V_margin} is a margin condition with exponent $\alpha=1$ \citep{Tsybakov2005}. Under vanishing-density geometries $P_0(|U_{\theta_0}|\le t)\lesssim t^\alpha$ with $\alpha>1$, the smoothing bias improves while the boundary information degrades ($I_{\tau,\vartheta\vartheta}\asymp\tau^{\alpha-2}$), and for $\alpha>2$ it no longer diverges, changing the regime qualitatively. A full theory across margin classes, with adaptation to unknown $\alpha$, is open. We do not pursue it here, since the generic case covers the designs of practical interest and purely discrete boundary scores, as in Section~\ref{sec:analysis}, are better treated as a separate set-identified problem.
\end{remark}

\subsection{The boundary-layer toolkit}

\begin{lemma}[Gate-factor moments]
\label{lem:gate_factor_moments}
For $j\ge1$, $k\ge0$ and $l\in\{1,2\}$ there are constants $C_{jkl},c_*>0$ such that for all $s\in\R$, all $a,a'\in\mathcal A_\delta$, and both $e\in\{0,1\}$:
\begin{enumerate}[label=(\roman*),nosep]
\item\label{gfm:matched} if $e=d(s):=\mathbbm 1\{s\ge0\}$ (matched component), then
$\E_{Y\sim f^{\,\prime}_{e}}\big[|\phi(s)^jq_s(Y)^j|^{\,l}\,\big]\le C_{jkl}\,e^{-c_*s^2}$,
where $f'_e$ denotes the $e$-component Gaussian density at any parameter $a'\in\mathcal A_\delta$ and $q_s$ is evaluated at $a$;
\item\label{gfm:mismatched} for arbitrary $e\in\{0,1\}$ and $|s|\le\bar s$,
$\E_{Y\sim f'_{e}}\big[|q_s(Y)|^{jl}\big]\le C(\bar s)<\infty$.
\end{enumerate}
The same bounds hold with $q_s$ replaced by $\partial_a^mq_s$ for $|m|\le3$, and with additional polynomial factors $(1+|s|)^k$ absorbed into the constants after replacing $c_*$ by $c_*/2$.
\end{lemma}

\begin{proof}
Since $p_t\ge\Phi(t)f_1$ and $p_t\ge(1-\Phi(t))f_0$, we have $p_t^2\ge\Phi(t)(1-\Phi(t))f_0f_1$ and hence
\[
|q_s(y)|\le\frac{f_1(y)+f_0(y)}{p_s(y)}\le
\big\{\Phi(s)(1-\Phi(s))\big\}^{-1/2}\Big(\sqrt{f_1/f_0}(y)+\sqrt{f_0/f_1}(y)\Big).
\]
For Gaussian components with means and variances ranging in compact sets ($\sigma^2$ bounded below), the ratio $f_1/f_0(y)=\exp\{(2y-m_0-m_1)(m_1-m_0)/(2\sigma^2)\}$ has, under any Gaussian law $f'_e$ with bounded mean and variance, moments of every fixed order bounded by a constant $C$ (Gaussian moment generating function with bounded coefficients); likewise for $f_0/f_1$ and for the $a$-derivatives of $q_s$, which by the quotient rule are finite sums of products of $q_s$-type ratios with polynomials in $(Y,W,X)$, themselves bounded by $C_B$ and Gaussian moments. Hence
\[
\E_{f'_e}\big[|q_s|^{jl}\big]\le C\,\{\Phi(s)(1-\Phi(s))\}^{-jl/2}.
\]
For \ref{gfm:mismatched}, $\{\Phi(s)(1-\Phi(s))\}^{-1}\le C(\bar s)$ on $|s|\le\bar s$. For \ref{gfm:matched}, multiply by $\phi(s)^{jl}$ and use the standard bound $\phi(s)\{\Phi(s)(1-\Phi(s))\}^{-1/2}\le C e^{-s^2/4}$ for all $s$: for $|s|\le1$ both factors are bounded, while for $|s|>1$ the Mills lower bound $\Phi(-|s|)\ge\phi(s)\,|s|/(1+s^2)$ gives $\phi(s)^2/\{\Phi(s)(1-\Phi(s))\}\le2\phi(s)(1+s^2)/|s|\le Ce^{-s^2/4}$, which yields the exponential factor with $c_*$ depending on $(j,l)$; polynomial factors $(1+|s|)^k$ are absorbed by halving the exponent.
\end{proof}

\begin{lemma}[Boundary-layer integration]
\label{lem:boundary_layer_integration}
Let $h(o,s)$ be measurable with $|h(o,s)|\le C_h\,e^{-c_*s^2}$ for all $o$ in the (bounded) support of $O$ and all $s\in\R$. Then for all $\theta\in N_0$ and $\tau\in(0,\tau_1]$:
\begin{enumerate}[label=(\roman*),nosep]
\item $\big|\E_{P_0}\big[h(O,U_\theta/\tau)\big]\big|\le C\,C_h\,\tau$;
\item if in addition $s\mapsto h(\cdot,s)$ and the conditional expectations $u\mapsto\E[h(O,s)\mid U_\theta=u]$ are continuous at $u=0$ uniformly in $s$ on compacts (in the sense of Assumption~\ref{asmp:V}\ref{subasmp:V_margin}), then
\[
\tau^{-1}\,\E_{P_0}\big[h(O,U_\theta/\tau)\big]\ \longrightarrow\ f_\theta(0)\int_{-\infty}^{\infty}\E\big[h(O,s)\mid U_\theta=0\big]\,ds
\qquad(\tau\downarrow0),
\]
uniformly over $\theta\in N_0$.
\end{enumerate}
\end{lemma}

\begin{proof}
Split at $|U_\theta|\le t_0$. On the inner region, condition on $U_\theta=u$ and substitute $u=\tau s$:
\[
\E\big[h\,\mathbbm 1\{|U_\theta|\le t_0\}\big]
=\int_{-t_0}^{t_0}\E[h(O,u/\tau)\mid U_\theta=u]f_\theta(u)\,du
=\tau\int_{-t_0/\tau}^{t_0/\tau}\E[h(O,s)\mid U_\theta=\tau s]f_\theta(\tau s)\,ds .
\]
The integrand is bounded by $C_fC_he^{-c_*s^2}$, integrable uniformly; this gives (i) for the inner region, and (ii) follows by dominated convergence using the continuity assumptions. On the outer region $|U_\theta|>t_0$, boundedness of the design gives $|U_\theta/\tau|>t_0/\tau$, so $|h|\le C_he^{-c_*t_0^2/\tau^2}=o(\tau^k)$ for every $k$.
\end{proof}

\begin{lemma}[Score identities and the $\chi^2$ representation]
\label{lem:conditional_score_gate_mismatch_new}
Fix $(W,X,Z)$ and a parameter $a\in\mathcal A_\delta$ with $x^\top\gamma\ne0$, and let
$\chi^2_s=\int\{f_1(y)-f_0(y)\}^2/p_s(y)\,dy$. Then for $e\in\{0,1\}$,
\begin{equation}
\E_{Y\sim f_e}\big[q_s(Y)\big]=\{e-\Phi(s)\}\,\chi^2_s .
\label{eq:chisq_identity_app}
\end{equation}
Consequently, writing $\omega_s(y)=\Phi(s)f_1(y)/p_s(y)$ for the responsibility,
\[
\E_{f_e}[\omega_s(Y)]-\Phi(s)=\Phi(s)(1-\Phi(s))\{e-\Phi(s)\}\,\chi^2_s ,
\qquad
0\le\Phi(s)(1-\Phi(s))\chi^2_s\le1 .
\]
Moreover, with $\partial_t\log p_t=\phi(t)q_t$ and $\partial^2_t\log p_t=-t\phi(t)q_t-\phi(t)^2q_t^2$,
\begin{equation}
\E_{f_{d(s)}}\big[-\partial^2_t\log p_t(Y)\big|_{t=s}\big]
=\phi(s)^2\,\E_{f_{d(s)}}[q_s(Y)^2]\;+\;|s|\,\phi(s)\,\Phi(-|s|)\,\chi^2_s\ \ \ge\ 0 ,
\label{eq:pointwise_info_positivity}
\end{equation}
with strict inequality whenever $f_1\neq f_0$.
\end{lemma}

\begin{proof}
Write $f_e=p_s+\{e-\Phi(s)\}(f_1-f_0)$, which holds for $e\in\{0,1\}$ by inspection of the mixture. Then
\[
\E_{f_e}[q_s]=\int\frac{f_1-f_0}{p_s}\Big[p_s+\{e-\Phi(s)\}(f_1-f_0)\Big]
=\int(f_1-f_0)+\{e-\Phi(s)\}\chi^2_s=\{e-\Phi(s)\}\chi^2_s,
\]
since the components integrate to one. The responsibility identity follows from $\omega_s-\Phi(s)=\Phi(s)(1-\Phi(s))q_s$, which is direct algebra, and the bound $\Phi(1-\Phi)\chi^2\le1$ follows from $\E_{f_1}[q_s]-\E_{f_0}[q_s]=\chi^2_s$ together with $\E_{f_e}[\omega_s]\in[0,1]$. The derivative formulas are direct calculation using $\partial_tp_t=\phi(t)(f_1-f_0)$, $\phi'(t)=-t\phi(t)$ and $\partial_tq_t=-\phi(t)q_t^2$. Taking $\E_{f_{d(s)}}$ in $-\partial^2_t\log p=t\phi q+\phi^2q^2$ and applying \eqref{eq:chisq_identity_app} with $e=d(s)$:
$s\{d(s)-\Phi(s)\}=s(1-\Phi(s))\ge0$ for $s\ge0$ and $s\{-\Phi(s)\}=|s|\Phi(-|s|)\ge0$ for $s<0$; in both cases $s\{d(s)-\Phi(s)\}=|s|\Phi(-|s|)$. Both summands in \eqref{eq:pointwise_info_positivity} are therefore nonnegative, the first strictly positive when $f_1\ne f_0$.
\end{proof}

\subsection{Risk localization}

\begin{lemma}[Risk localization]
\label{lem:risk_localization}
Under Assumption~\ref{asmp:V}\ref{subasmp:V_gauss}--\ref{subasmp:V_wedge} there are constants $c,C,\tau_1>0$ such that for all $\tau\in(0,\tau_1]$ and all $\eta\in\mathcal A_\delta\times N_0$:
\begin{enumerate}[label=(\alph*),nosep]
\item $|K_\tau(\eta)-K_\infty(\eta)|\le C\tau$, where $K_\infty$ is the hard-threshold risk obtained by replacing $\pi_{\theta,\tau}$ with $\mathbbm 1\{U_\theta\ge0\}$;
\item $K_\infty(\eta)-K_\infty(\eta_0)\ge c\,\|a-a_0\|^2+c\,\|\theta-\theta_0\|$.
\end{enumerate}
Consequently every minimizer $\eta^\star(\tau)$ of $K_\tau$ over $\mathcal A_\delta\times N_0$ satisfies
$\|\theta^\star(\tau)-\theta_0\|\le C'\tau$ and $\|a^\star(\tau)-a_0\|\le C'\sqrt\tau$.
\end{lemma}

The cone-shaped lower bound in (b) is the population signature of the hard threshold, the risk growing linearly in the boundary error once the error exceeds the smoothing scale and quadratically below it.

\begin{proof}[Proof of Lemma~\ref{lem:risk_localization}]
(a) Fix $\eta\in\mathcal A_\delta\times N_0$ and let $d=\mathbbm 1\{U_\theta\ge0\}$, so $p_\infty=f_d$ at the parameter $a$. Then
\[
\frac{p_{\eta,\tau}}{p_\infty}=w_d+(1-w_d)\frac{f_{1-d}}{f_d},
\]
where $w_d\in\{\Phi(t),1-\Phi(t)\}$ is the gate weight of the component selected by $d$, so that $w_d=\Phi(|t|)\ge1/2$ and $1-w_d=\Phi(-|t|)$. Using $\log(1+x)\le x$ and $w_d\ge1/2$,
\[
-2\Phi(-|t|)\ \le\ \log w_d\ \le\ \log\frac{p_{\eta,\tau}}{p_\infty}\ \le\ \Phi(-|t|)\,\frac{f_{1-d}}{f_d}(Y).
\]
Taking absolute values, expectations under the true conditional law (a Gaussian with bounded parameters), and using $\E[f_{1-d}/f_d(Y)\mid W,X,Z]\le C$ (Lemma~\ref{lem:gate_factor_moments} argument),
$\E\big[\,|\log(p_{\eta,\tau}/p_\infty)|\;\big|\;W,X,Z\big]\le C\,\Phi(-|U_\theta|/\tau)$.
Lemma~\ref{lem:boundary_layer_integration}(i) with $h=\Phi(-|s|)\le e^{-s^2/2}\vee\ldots$ (indeed $\Phi(-|s|)\le e^{-s^2/2}$) integrates this to $|K_\tau(\eta)-K_\infty(\eta)|\le C\tau$, uniformly on $\mathcal A_\delta\times N_0$.

(b) With Gaussian working components and the true conditional law $\mathcal N(\mu_0(W,X,Z),\sigma_0^2)$, $\mu_0=W^\top\beta_0+X^\top\gamma_0d_0$, a direct computation gives
\[
K_\infty(\eta)-K_\infty(\eta_0)
=\tfrac12\Big[\log\tfrac{\sigma^2}{\sigma_0^2}+\tfrac{\sigma_0^2}{\sigma^2}-1\Big]
+\frac{1}{2\sigma^2}\,\E_{P_0}\big[(\mu_0-m_d)^2\big],
\qquad m_d:=W^\top\beta+X^\top\gamma\,\mathbbm 1\{U_\theta\ge0\}.
\]
The bracket is $\ge c(\sigma^2-\sigma_0^2)^2$ on the compact range. For the mean term write
$\mu_0-m_d=-A-B$ with
$A=W^\top(\beta-\beta_0)+X^\top(\gamma-\gamma_0)d_0$ and
$B=X^\top\gamma\,\{\mathbbm 1\{U_\theta\ge0\}-d_0\}$, so $|B|\le C\,\mathbbm 1_{\mathcal E}$ with $\mathcal E=\mathcal E(\theta,\theta_0)$. First,
\[
\E[(A+B)^2]\ \ge\ \E[A^2\mathbbm 1_{\mathcal E^c}]\ \ge\ \E[A^2]-C\,P_0(\mathcal E)
\ \ge\ c_1\big(\|\beta-\beta_0\|^2+\|\gamma-\gamma_0\|^2\big)-C_2\|\theta-\theta_0\|,
\]
using the nonsingularity of $\E[(W^\top,X^\top d_0)^\top(\cdots)]$ (Assumption~\ref{asmp:V}\ref{subasmp:V_bounded}), boundedness, and the wedge upper bound (Assumption~\ref{asmp:V}\ref{subasmp:V_wedge}). Second, on $\mathcal E$ we have $(A+B)^2\ge\tfrac12B^2-A^2$, so
\[
\E[(A+B)^2]\ \ge\ \tfrac12\E[B^2\mathbbm 1_{\mathcal E}]-\E[A^2\mathbbm 1_{\mathcal E}]
\ \ge\ \tfrac12\,c_3\,P_0(\mathcal E)-C_4\,\|a-a_0\|^2\,P_0(\mathcal E),
\]
where $c_3>0$ bounds $\E[(X^\top\gamma)^2\mid\mathcal E]$ below: for $\theta\in N_0$ the wedge concentrates on $\{|U_{\theta_0}|\le C_B\|\theta-\theta_0\|\}$, and by Assumption~\ref{asmp:V}\ref{subasmp:V_margin}--\ref{subasmp:V_boundary_info} and continuity of the conditional moments, $\E[(X^\top\gamma)^2\mid\mathcal E]\ge\tfrac12\E[(X^\top\gamma_0)^2\mid U_{\theta_0}=0]-o(1)-C\|\gamma-\gamma_0\|\ge c_3$ after shrinking $N_0$ and $\mathcal A_\delta$ if necessary. For $\|a-a_0\|^2\le c_3/(4C_4)$ this gives
$\E[(A+B)^2]\ge\tfrac14c_3\,P_0(\mathcal E)\ge\tfrac14c_3c_w\|\theta-\theta_0\|$
by the wedge lower bound. Combining the two displays with weights $\tfrac12$ each and adjusting constants (using also $\|a-a_0\|^2\ge c_3/(4C_4)$ on the complementary region, where the first display alone gives a fixed positive lower bound dominating $C_2\varepsilon_0$ after shrinking $N_0$): there exist $c,c'>0$ with
\[
K_\infty(\eta)-K_\infty(\eta_0)\ \ge\ c\,\|a-a_0\|^2+c'\,\|\theta-\theta_0\|
\qquad\text{on }\mathcal A_\delta\times N_0 .
\]
Conclusion. If $\eta^\star=\eta^\star(\tau)$ minimizes $K_\tau$ over $\mathcal A_\delta\times N_0$ then
$K_\infty(\eta^\star)\le K_\tau(\eta^\star)+C\tau\le K_\tau(\eta_0)+C\tau\le K_\infty(\eta_0)+2C\tau$,
so by (b), $c\|a^\star-a_0\|^2+c'\|\theta^\star-\theta_0\|\le2C\tau$.
\end{proof}

\subsection{Anisotropic information bounds}

Write $S_\tau(\tilde\eta)=\nabla_{\tilde\eta}K_\tau(\eta(\tilde\eta))$ and $I_\tau(\tilde\eta)=\nabla^2_{\tilde\eta}K_\tau(\eta(\tilde\eta))$, with blockwise subscripts $a$ and $\vartheta$.

\begin{lemma}[Anisotropic information bounds]
\label{lem:aniso_information}
Under Assumption~\ref{asmp:V} there are constants $0<c\le C<\infty$ and $\tau_1>0$ such that for all $\tau\in(0,\tau_1]$ and all $\tilde\eta$ with $\|a-a_0\|\le C_1\sqrt\tau$, $\|\vartheta-\vartheta_0\|\le C_1\tau$:
\[
c\,\preceq\, I_{\tau,aa}(\tilde\eta)\,\preceq\, C,\qquad
c\,\tau^{-1}\preceq I_{\tau,\vartheta\vartheta}(\tilde\eta)\preceq C\,\tau^{-1},\qquad
\|I_{\tau,a\vartheta}(\tilde\eta)\|\le C,
\]
where the first two displays are two-sided bounds in the positive-semidefinite order. Moreover the score at the hard-threshold parameter satisfies
$\|S_{\tau,a}(\tilde\eta_0)\|\le C\tau$ and $\|S_{\tau,\vartheta}(\tilde\eta_0)\|\le C$.
\end{lemma}

\begin{proof}[Proof of Lemma~\ref{lem:aniso_information}]
Throughout, derivatives of $\tilde\ell_\tau$ are computed from the explicit formulas in the proof of Lemma~\ref{lem:subexp_fixed_tau_new} together with $\partial_t\log p=\phi(t)q_t$, $\partial^2_t\log p=-t\phi(t)q_t-\phi(t)^2q_t^2$, $\partial^3_t\log p=(t^2-1)\phi q+3t\phi^2q^2+2\phi^3q^3$, the chain rule $\nabla_\theta=Z\tau^{-1}\partial_t$, and the bounded chart Jacobian. Differentiation under the integral is justified exactly as in Lemma~\ref{lem:score_hess_envelopes_fixed_tau_new}, the dominating envelopes now being those produced below.

Derivative envelopes. Consider first parameters with $\theta=\theta_0$ (no sign mismatch). Every mixed derivative $\nabla^j_a\nabla^k_\vartheta\tilde\ell_\tau$ with $k\ge1$ is a finite sum of terms of the form
$\tau^{-k}\times(\text{polynomial in }(1+|t|))\times\phi(t)^m q_t\text{-products}\times(\text{bounded data polynomials})$
with $m\ge1$ powers of $\phi q$-type factors. By Lemma~\ref{lem:gate_factor_moments}\ref{gfm:matched}, the conditional expectation (under the true law, given $(W,X,Z)$) of the absolute value (or its square) of each such term is bounded by $C\tau^{-k}e^{-c_*t^2}$ (respectively $C\tau^{-2k}e^{-c_*t^2}$); Lemma~\ref{lem:boundary_layer_integration}(i) then gives
\begin{equation}
\E\big\|\nabla^j_a\nabla^k_\vartheta\tilde\ell_\tau\big\|\le C\tau^{1-k},
\qquad
\E\big\|\nabla^j_a\nabla^k_\vartheta\tilde\ell_\tau\big\|^2\le C\tau^{1-2k},
\qquad k\ge1,\ j+k\le3 .
\label{eq:derivative_envelopes}
\end{equation}
Derivatives with $k=0$ have $\tau$-free bounded-moment envelopes as in Lemma~\ref{lem:score_hess_envelopes_fixed_tau_new}. For parameters with $0<\|\theta-\theta_0\|\le c_1\tau$, the only modification occurs on the mismatch event
$\mathcal M=\{\mathbbm 1\{U_\theta\ge0\}\neq d_0\}\subseteq\{|U_{\theta_0}|\le c_1C_B\tau\}\cap\{|t|\le2c_1C_B\}$:
there the matched-component bound of Lemma~\ref{lem:gate_factor_moments}\ref{gfm:matched} is unavailable, but $|t|$ is bounded, so Lemma~\ref{lem:gate_factor_moments}\ref{gfm:mismatched} bounds the conditional moments by constants, and $P_0(\mathcal M)\le C\tau$ by the density bound; the contribution of $\mathcal M$ to each expectation is therefore of the same order as \eqref{eq:derivative_envelopes} with constants depending on $c_1$. The same argument applies uniformly for $a$ ranging over $\mathcal A_\delta$.

Block orders of $I_\tau$. The upper bounds
$\|I_{\tau,aa}\|\le C$, $\|I_{\tau,a\vartheta}\|\le C$, $\|I_{\tau,\vartheta\vartheta}\|\le C\tau^{-1}$
on the stated neighborhood follow from \eqref{eq:derivative_envelopes} with $(j,k)=(2,0),(1,1),(0,2)$.

Lower bound, $\vartheta\vartheta$ block at $\tilde\eta_0$. With $J_0=\nabla_\vartheta\theta(\vartheta_0)$ (full rank $q-1$) and $t=U_0/\tau$,
\[
I_{\tau,\vartheta\vartheta}(\tilde\eta_0)
=\tau^{-2}\,J_0^\top\,\E\Big[ZZ^\top\,\E_{f^0_{d_0}}\big[-\partial^2_t\log p_t(Y)\big]\Big]J_0
=\tau^{-2}\,J_0^\top\,\E\big[ZZ^\top\rho(U_0/\tau;W,X)\big]J_0,
\]
where, by \eqref{eq:pointwise_info_positivity},
$\rho(s;w,x)=\phi(s)^2\E_{f_{d(s)}}[q_s^2]+|s|\phi(s)\Phi(-|s|)\chi^2_s\ \ge\ \phi(s)^2\E_{f_{d(s)}}[q_s^2]\ge0$,
all quantities evaluated at $a_0$. By Lemma~\ref{lem:gate_factor_moments}, $0\le\rho(s;w,x)\le Ce^{-c_*s^2}$, and Lemma~\ref{lem:boundary_layer_integration}(ii) gives, uniformly in unit vectors $v\in\R^{q-1}$,
\[
\tau\,v^\top I_{\tau,\vartheta\vartheta}(\tilde\eta_0)v
\ \longrightarrow\
f_{\theta_0}(0)\,(J_0v)^\top\E\Big[ZZ^\top\!\!\int\rho(s;W,X)ds\;\Big|\;U_{\theta_0}=0\Big](J_0v).
\]
Discarding the second (nonnegative) summand of $\rho$, the limit matrix dominates
$f_{\theta_0}(0)\,J_0^\top\E[\chi^2_{\mathrm B}(W,X)ZZ^\top\mid U_{\theta_0}=0]J_0$ in the positive semidefinite order, with $\chi^2_{\mathrm B}$ as in \eqref{eq:chisq_weight}; since $J_0v$ lies in the tangent space at $\theta_0$ and is bounded away from zero in norm, Assumption~\ref{asmp:V}\ref{subasmp:V_boundary_info} makes the limit positive definite, uniformly over unit $v$. Hence $I_{\tau,\vartheta\vartheta}(\tilde\eta_0)\succeq c\,\tau^{-1}$ for small $\tau$.

Lower bound, $aa$ block at $\tilde\eta_0$. Direct differentiation of $\log p$ twice in $a$ and the collapse of the responsibilities away from the boundary give
$I_{\tau,aa}(\tilde\eta_0)=\mathcal I_0+E_\tau$ with $\mathcal I_0$ the oracle information of Theorem~\ref{thm:effect_bvm} and $\|E_\tau\|\le C\tau$: indeed the difference of the integrands is supported, in conditional expectation, on gate-discrepancy factors bounded by $C\Phi(-|t|)\le Ce^{-t^2/2}$ (the responsibilities satisfy $\E_{f^0_{d_0}}|\omega_t-d_0|\le C\Phi(-|t|)$ by the responsibility identity in Lemma~\ref{lem:conditional_score_gate_mismatch_new} and $\Phi(1-\Phi)\chi^2\le1$), and Lemma~\ref{lem:boundary_layer_integration}(i) applies. $\mathcal I_0$ is positive definite by Assumption~\ref{asmp:V}\ref{subasmp:V_bounded} and $\sigma_0^2>0$.

Extension to the anisotropic tube. For $\|a-a_0\|\le\delta_1$ and $\|\vartheta-\vartheta_0\|\le C_1\tau$, the mixed third-derivative envelopes \eqref{eq:derivative_envelopes} give
\[
\|I_{\tau,\vartheta\vartheta}(\tilde\eta)-I_{\tau,\vartheta\vartheta}(\tilde\eta_0)\|
\le C\tau^{-2}\,\|\vartheta-\vartheta_0\|+C\tau^{-1}\,\|a-a_0\|
\le C(C_1+\delta_1)\,\tau^{-1},
\]
which is dominated by $\tfrac12c\tau^{-1}$ once $\delta_1$ and $C_1$ are fixed small enough relative to $c$; the $aa$ and $a\vartheta$ blocks have perturbations $O(\delta_1+C_1\tau)$ by the $(3,0)$ and $(2,1)$ envelopes, which are $O(1)$. This establishes the two-sided block bounds on the tube $\{\|a-a_0\|\le\delta_1,\ \|\vartheta-\vartheta_0\|\le C_1\tau\}$. (The constants $\delta_1,C_1$ are fixed by this argument and referenced by all later proofs.)

Score bounds at $\tilde\eta_0$. For the regression block, the conditional score identities give, exactly as in the displays above,
$\E[\nabla_a\tilde\ell_\tau(\tilde\eta_0;O)\mid W,X,Z]$ equal to gate-discrepancy terms bounded by $C\Phi(-|t|)$ in conditional absolute value (the $\varepsilon$-linear terms vanish by $\E[\varepsilon\mid\cdot]=0$ and the Gaussian change-of-measure factors are handled by Lemma~\ref{lem:gate_factor_moments}); Lemma~\ref{lem:boundary_layer_integration}(i) yields $\|S_{\tau,a}(\tilde\eta_0)\|\le C\tau$. For the boundary block, by \eqref{eq:chisq_identity_app},
\[
\E[\nabla_\vartheta\tilde\ell_\tau(\tilde\eta_0;O)\mid W,X,Z]
=J_0^\top Z\,\tau^{-1}\phi(t)\{d_0-\Phi(t)\}\chi^2_t ,
\]
and $|\phi(s)\{d(s)-\Phi(s)\}\chi^2_s|=\phi(s)\Phi(-|s|)\chi^2_s\le Ce^{-c_*s^2}$ by Lemma~\ref{lem:gate_factor_moments} (using $\Phi(1-\Phi)\chi^2\le1$ and the matched-moment bound), so Lemma~\ref{lem:boundary_layer_integration}(i) gives $\|S_{\tau,\vartheta}(\tilde\eta_0)\|\le C\tau^{-1}\cdot C\tau=C$. We emphasize that the boundary-block score does not vanish as $\tau\downarrow0$, only its product with the inverse boundary information does.
\end{proof}

\subsection{The pseudo-true path}

\begin{proof}[Proof of Proposition~\ref{prop:eta_star_bias_rate_new}]
By Assumption~\ref{asmp:V}\ref{subasmp:V_anchor} and Lemma~\ref{lem:risk_localization}, any minimizer satisfies $\|a^\star(\tau)-a_0\|\le C\sqrt\tau$ and $\|\theta^\star(\tau)-\theta_0\|\le C_{\mathrm{loc}}\tau$. Let $(\delta_1,C_1)$ be the tube constants fixed in the proof of Lemma~\ref{lem:aniso_information}. If $C_{\mathrm{loc}}>C_1$, apply Assumption~\ref{asmp:V}\ref{subasmp:V_seam} with the pair $(C_1,C_{\mathrm{loc}})$: since $K_\tau(\eta^\star)\le K_\tau(\eta_0)$, the minimizer cannot lie in the region $C_1\tau\le\|\theta^\star-\theta_0\|\le C_{\mathrm{loc}}\tau$ for $\tau\le\tau_s$, whence $\|\theta^\star(\tau)-\theta_0\|<C_1\tau$; and $\|a^\star(\tau)-a_0\|\le C\sqrt\tau\le\delta_1$ for small $\tau$. Thus the minimizer lies in the interior of the anisotropic tube of Lemma~\ref{lem:aniso_information}, where $K_\tau$ is twice continuously differentiable; hence the first-order condition $S_\tau(\tilde\eta^\star)=0$ holds. Taylor's theorem with integral remainder along the segment from $\tilde\eta_0$ to $\tilde\eta^\star$ (which stays in the neighborhood by convexity of the bounds) gives
\[
0=S_\tau(\tilde\eta^\star)=S_\tau(\tilde\eta_0)+\bar I\,(\tilde\eta^\star-\tilde\eta_0),
\qquad
\bar I:=\int_0^1I_\tau\big(\tilde\eta_0+t(\tilde\eta^\star-\tilde\eta_0)\big)dt ,
\]
and $\bar I$ inherits the block bounds of Lemma~\ref{lem:aniso_information}:
$c\preceq\bar I_{aa}\preceq C$, $c\tau^{-1}\preceq\bar I_{\vartheta\vartheta}\preceq C\tau^{-1}$, $\|\bar I_{a\vartheta}\|\le C$.
Write $x=a^\star-a_0$, $v=\vartheta^\star-\vartheta_0$ and solve the block system
$S_a+\bar I_{aa}x+\bar I_{a\vartheta}v=0$, $S_\vartheta+\bar I_{\vartheta a}x+\bar I_{\vartheta\vartheta}v=0$:
from the second equation, $\|v\|\le\|\bar I_{\vartheta\vartheta}^{-1}\|(\|S_\vartheta\|+C\|x\|)\le C\tau(1+\|x\|)\le C'\tau$;
substituting into the first, $\|x\|\le\|\bar I_{aa}^{-1}\|(\|S_a\|+C\|v\|)\le C(\tau+\tau)=C'\tau$.
(The crude bound $\|x\|\le C\sqrt\tau$ was used only to place the segment inside the neighborhood.) This proves the bias bounds, and the chart norm-equivalence transfers them to $\theta$.

Eventual uniqueness and continuity. On the anisotropic neighborhood the Hessian of $K_\tau$ satisfies, for any direction $(u_a,u_\vartheta)$,
\[
(u_a,u_\vartheta)^\top I_\tau(\tilde\eta)(u_a,u_\vartheta)
\ \ge\ c\|u_a\|^2+c\tau^{-1}\|u_\vartheta\|^2-2C\|u_a\|\|u_\vartheta\|
\ \ge\ \tfrac{c}{2}\big(\|u_a\|^2+\tau^{-1}\|u_\vartheta\|^2\big)
\]
for $\tau\le c^2/(8C^2)$, by Young's inequality. Hence $K_\tau$ is strictly convex there, so it has at most one stationary point; existence holds since the global minimizer over $\bar\Theta$ lies inside by Assumption~\ref{asmp:V}\ref{subasmp:V_anchor} and the localization above. Continuity of $\tau\mapsto\tilde\eta^\star(\tau)$ on $(0,\tau_1]$ follows from the implicit function theorem applied to $(\tau,\tilde\eta)\mapsto S_\tau(\tilde\eta)$, which is continuously differentiable with invertible $\tilde\eta$-derivative on the neighborhood.
\end{proof}
\subsection{Variance and deviation bounds for the triangular array}

In the remainder of this section $\tau=\tau_n=n^{-\varrho}$, $\varrho\in(0,1)$ (the proofs use only $n\tau_n/(\log n)^6\to\infty$, automatic for polynomial schedules, and $\sqrt n\tau_n\to0$ where indicated), $\eta_n^\star=\eta^\star(\tau_n)$ as in Proposition~\ref{prop:eta_star_bias_rate_new}, and
$r_\eta(O)=\log\{p_{\eta,\tau_n}/p_{\eta_n^\star,\tau_n}\}(Y\mid W,X,Z)$, $R_n(\eta)=\sum_ir_\eta(O_i)$, so that \eqref{eq:Rn_decomp} holds with $\Delta K=K_n-K_n(\eta_n^\star)$, $K_n=K_{\tau_n}$.

\begin{lemma}[$\tau$-free envelope and variance bounds]
\label{lem:variance_bounds_tau_n}
There are constants $C,K_0$ such that for all $n$ and all $\eta,\eta'\in\bar\Theta$:
\begin{enumerate}[label=(\roman*),nosep]
\item\label{vb:envelope} $|r_\eta(O)|\le C(1+\varepsilon^2)$ almost surely; in particular $\|r_\eta\|_{\psi_1}\le K_0$;
\item\label{vb:zone3} if $\eta,\eta'\in\mathcal A_\delta\times N_0$ with $\|\theta-\theta_0\|\vee\|\theta'-\theta_0\|\le C_1\tau_n$, then
\[
\E_{P_0}\big[(r_\eta-r_{\eta'})^2\big]\ \le\ C\big(\|a-a'\|^2+\tau_n^{-1}\|\theta-\theta'\|^2\big);
\]
\item\label{vb:zone2} if $\eta=(a,\theta),\eta'=(a',\theta')\in\mathcal A_\delta\times N_0$, then
\[
\E_{P_0}\big[(r_\eta-r_{\eta'})^2\big]\ \le\ C\big(\|a-a'\|^2+\|\theta-\theta'\|+\tau_n\sqrt{\log(1/\tau_n)}\big).
\]
\end{enumerate}
\end{lemma}

\begin{proof}
\ref{vb:envelope} As in the proof of Lemma~\ref{lem:subexp_fixed_tau_new}, the better-classified mixture component always has gate weight at least $1/2$, so
$|\log p_{\eta,\tau}|\le\log2+\max_j|\log f_j|\le C(1+(Y-m_j)^2)\le C'(1+\varepsilon^2)$
using bounded design and bounded parameter ranges; $\varepsilon$ is Gaussian, so $\varepsilon^2$ is sub-exponential.

\ref{vb:zone3} Write $r_\eta-r_{\eta'}=\log p_{\eta,\tau}-\log p_{\eta',\tau}=\int_0^1\partial_s\log p_{\eta_s,\tau}\,ds$ along the segment $\eta_s=(1-s)\eta'+s\eta$ (with $\theta$-leg traversed in the chart, all points remaining in the anisotropic tube after enlarging $C_1$ by a fixed factor). By Cauchy--Schwarz and Fubini,
\[
\E[(r_\eta-r_{\eta'})^2]\le\int_0^1\E\big[\{(\eta-\eta')^\top\nabla\log p_{\eta_s,\tau}\}^2\big]ds
\le2\|a-a'\|^2\sup_s\E\|\nabla_a\ell_{\tau}(\eta_s)\|^2
+2\|\vartheta-\vartheta'\|^2\sup_s\E\|\nabla_\vartheta\ell_\tau(\eta_s)\|^2,
\]
and the envelopes \eqref{eq:derivative_envelopes} (with $(j,k,l)=(0,1,2)$, valid uniformly on the tube) bound the second supremum by $C\tau^{-1}$ and the first by $C$.

\ref{vb:zone2} Split $r_\eta-r_{\eta'}$ into an $a$-leg at fixed $\theta$ and a $\theta$-leg at fixed $a$. The $a$-leg is handled as above with the $\tau$-free bound $\E\|\nabla_a\ell\|^2\le C$. For the $\theta$-leg set $D=\log p_{(a,\theta),\tau}-\log p_{(a,\theta'),\tau}$ and let $K=\sqrt{2\log(1/\tau)}$. Partition the sample space into
$\mathcal R_1=\{\min(|U_\theta|,|U_{\theta'}|)\ge K\tau,\ \mathrm{sign}(U_\theta)=\mathrm{sign}(U_{\theta'})\}$,
$\mathcal R_2=\{\mathrm{sign}(U_\theta)\neq\mathrm{sign}(U_{\theta'})\}$, and the remainder $\mathcal R_3$.
On $\mathcal R_1$, both gates are within $\Phi(-K)\le\tau$ of the common indicator $d$, and writing
$\log p_{\pi}-\log f_d=\log\{w_d+(1-w_d)f_{1-d}/f_d\}$ with $1-w_d\le\Phi(-K)$, the sandwich used in the proof of Lemma~\ref{lem:risk_localization}(a) gives $|D|\le|{\log p_\pi-\log f_d}|+|{\log p_{\pi'}-\log f_d}|\le\Phi(-K)\{2+ (f_{1-d}/f_d)(Y)+(f'_{1-d}/f'_d)(Y)\}$, whose square has expectation $\le C\Phi(-K)^2\le C\tau^2$. On $\mathcal R_2$, $|D|\le C(1+\varepsilon^2)$ by \ref{vb:envelope}; since $\mathcal R_2$ is measurable with respect to $(W,X,Z)$ (it is defined by the two scores) and $\E[(1+\varepsilon^2)^2\mid W,X,Z]\le C$ under Gaussian errors with bounded conditional variance, conditioning gives $\E[D^2\mathbbm 1_{\mathcal R_2}]\le C\,P_0(\mathcal R_2)\le C'\|\theta-\theta'\|$, where $P_0(\mathcal R_2)\le P_0(|U_\theta|\le C_B\|\theta-\theta'\|)\le C\|\theta-\theta'\|$ by the density bound on $[-t_0,t_0]$ and the fact that sign disagreement with $|U_\theta|>C_B\|\theta-\theta'\|$ is impossible. On $\mathcal R_3$, which is likewise $(W,X,Z)$-measurable, one of the scores lies within $K\tau$ of zero, so $P_0(\mathcal R_3)\le CK\tau$, and the same conditioning gives a contribution $\le CK\tau=C\tau\sqrt{2\log(1/\tau)}$.
\end{proof}

\begin{lemma}[Truncated Bernstein bound]
\label{lem:truncated_bernstein}
Let $\xi,\xi_1,\dots,\xi_n$ be i.i.d.\ with $\E\xi=0$, $\E\xi^2\le v$ and $\|\xi\|_{\psi_1}\le K_0$. Then for every $t\ge n^{-1}$ and $R\ge CK_0\log n$ with $C$ a sufficiently large absolute constant,
\[
P\Big(\sum_{i=1}^n\xi_i\ge t\Big)\ \le\ \exp\Big(-\frac{t^2}{8nv+2Rt}\Big)+n\,e^{-R/(2K_0)} .
\]
\end{lemma}

\begin{proof}
Write $\xi_i=\xi_i\mathbbm 1\{|\xi_i|\le R\}+\xi_i\mathbbm 1\{|\xi_i|>R\}$. The event that any $|\xi_i|>R$ has probability at most $n\,2e^{-R/K_0}\cdot\tfrac12\le ne^{-R/(2K_0)}$ (sub-exponential tail). On its complement, the sum equals the sum of the truncated variables; these have mean $|\E\xi\mathbbm 1\{|\xi|\le R\}|=|\E\xi\mathbbm 1\{|\xi|>R\}|\le\E|\xi|\mathbbm 1\{|\xi|>R\}\le CK_0e^{-R/(2K_0)}\le t/(2n)$ for $R\ge CK_0\log n$ and the relevant $t\ge n^{-1}$, variance at most $v$, and absolute bound $2R$; Bernstein's inequality for bounded centered variables, applied at level $t/2$, gives $\exp\{-(t/2)^2/(2nv+2R(t/2)/3)\}\le\exp\{-t^2/(8nv+2Rt)\}$, the first term.
\end{proof}

\subsection{Proof of Theorem~\ref{thm:tau_contraction_new}}

\begin{proof}[Proof of Theorem~\ref{thm:tau_contraction_new}]
Fix $\epsilon>0$. Throughout, ``the denominator event'' refers to the conclusion of the following analogue of Lemma~\ref{lem:denom_fixed_tau_new}, proved identically with the anisotropic prior ball
$B_n^\star=\{\|a-a_n^\star\|\le n^{-1/2},\ \|\vartheta-\vartheta_n^\star\|\le(\tau_n/n)^{1/2}\}$:
expanding to second order with the envelopes \eqref{eq:derivative_envelopes} and noting that
$\E\|D_n^{(a)}\!\sum_i\nabla_a\ell\|^2\le C$, $\E\|(\tau_n/n)^{1/2}\sum_i\nabla_\vartheta\ell(\tilde\eta_n^\star)\|^2\le C$ (score variance $O(\tau_n^{-1})$ per observation; the means vanish at $\tilde\eta_n^\star$) and that the quadratic term is $\tfrac12(\tau_n/n)\sum_i\sup_{B_n^\star}\|\nabla^2_{\vartheta\vartheta}\ell(\cdot;O_i)\|=O_{P_0}(\tau_n\cdot\tau_n^{-1})=O_{P_0}(1)$ by Markov's inequality (the per-observation envelope may be taken uniform over $B_n^\star$ because the $\vartheta$-radius $(\tau_n/n)^{1/2}$ shifts $t=U/\tau_n$ by at most $C_B(n\tau_n)^{-1/2}=o(1)$, changing the boundary-layer factors only by bounded multiples), and we obtain $\inf_{B_n^\star}R_n\ge-O_{P_0}(1)$, whence with probability at least $1-\epsilon$,
\begin{equation}
\int e^{R_n}d\Pi\ \ge\ \Pi(B_n^\star)\,e^{-C_\epsilon}\ \ge\ c_\pi\,n^{-(p+r+1)/2}(\tau_n/n)^{(q-1)/2}e^{-C_\epsilon}
\ \ge\ e^{-C_\epsilon'\log n},
\label{eq:denom_tau_n}
\end{equation}
using prior positivity at $\tilde\eta_0$ (Assumption~\ref{asmp:V}\ref{subasmp:V_prior}; the ball center $\tilde\eta_n^\star$ converges to $\tilde\eta_0$) and $\tau_n\ge n^{-1}$.

We remove posterior mass from three zones; in each, we bound the $P_0$-probability that $\sup_{\text{zone}}R_n$ exceeds minus twice the zone's separation, then conclude via \eqref{eq:denom_tau_n} as in the proof of Theorem~\ref{thm:mis_contraction_fixed_tau_new}.

Zone I ($\|\eta-\eta_n^\star\|\ge\varepsilon_0$ (within $\bar\Theta$)). We first claim the uniform separation
$\inf_{n}\inf\{\Delta K(\eta):\|\eta-\eta_n^\star\|\ge\varepsilon_0\}\ge b_0>0$.
Indeed $K_\tau\to K_\infty$ uniformly on $\bar\Theta$ as $\tau\downarrow0$: the bound in the proof of Lemma~\ref{lem:risk_localization}(a) holds for every $\theta\in\mathbb S^{q-1}_+$ with the integrating factor $\E[\Phi(-|U_\theta|/\tau)]$, which converges to $P_0(U_\theta=0)=0$ pointwise (Assumption~\ref{asmp:design}\ref{asmp:design_nohyper}) and monotonically in $\tau$, hence uniformly over the compact sphere by Dini's theorem combined with continuity of $\theta\mapsto\E[\Phi(-|U_\theta|/\tau)]$; the regression-block constants are uniform on $\mathcal K_a$. Moreover $K_\infty$ is continuous on $\bar\Theta$ with unique minimizer $\eta_0$ on the nondegenerate-gate class (identification, Proposition~\ref{prop:identification}, plus the explicit form of $K_\infty$), so its sublevel separation is positive on $\{\|\eta-\eta_0\|\ge\varepsilon_0/2\}$, and $\eta_n^\star\to\eta_0$. The claim follows. Deviations over Zone I are controlled exactly as in Lemma~\ref{lem:tests_fixed_tau_new}, with two changes: the envelope is the $\tau$-free bound of Lemma~\ref{lem:variance_bounds_tau_n}\ref{vb:envelope}, and the Lipschitz modulus in $\theta$ now carries a factor $\tau_n^{-2}$ (crude pointwise bound), so the covering resolution is $\rho_n\asymp b_0\tau_n^{2}$ and $\log N_n\le Cd\log(1/\tau_n)\le C'\log n$, still negligible against the Bernstein exponent $cn$.

Zone II ($C_1\tau_n\le\|\theta-\theta_0\|\le2\varepsilon_0$, any $a\in\mathcal A_\delta$). Here $C_1$ is the tube constant of Lemma~\ref{lem:aniso_information} and we set $c_0=4\bar C/c'$, where $\bar C$ is the constant of Lemma~\ref{lem:risk_localization}(a) and $c,c'$ those of part (b). For $\|\theta-\theta_0\|\ge c_0\tau_n$ (the cone subzone), Lemma~\ref{lem:risk_localization} together with $K_n(\eta_n^\star)\le K_n(\eta_0)\le K_\infty(\eta_0)+\bar C\tau_n$ gives
$\Delta K(\eta)\ge c\|a-a_0\|^2+c'\|\theta-\theta_0\|-2\bar C\tau_n\ge\tfrac{c'}{4}\|\theta-\theta_0\|+\tfrac14(c\|a-a_0\|^2-4\bar C\tau_n)_+$. For $C_1\tau_n\le\|\theta-\theta_0\|<c_0\tau_n$ (the seam subzone), Assumption~\ref{asmp:V}\ref{subasmp:V_seam} applied with the pair $(C_1,c_0)$ gives $K_n(\eta)\ge K_n(\eta_0)+c_s\tau_n\ge K_n(\eta_n^\star)+c_s\tau_n$, and combining with the first-order bound as above, $\Delta K(\eta)\ge\max\{c_s\tau_n,\ (c\|a-a_0\|^2-4\bar C\tau_n)_+\}\ge\tfrac12c_s\tau_n+\tfrac14(c\|a-a_0\|^2-4\bar C\tau_n)_+$. In both subzones,
\[
\Delta K(\eta)\ \ge\ g(\eta):=\kappa_0\big\{\tau_n+\|\theta-\theta_0\|+(\|a-a_0\|^2-4\bar C\tau_n/c)_+\big\},
\qquad\kappa_0:=\tfrac14\min\{c_s,\ c'\!,\ c\}\wedge\tfrac{c_sC_1}{4c_0},
\]
using $\tau_n\le\|\theta-\theta_0\|/C_1$ on the cone subzone and $\|\theta-\theta_0\|\le c_0\tau_n$ on the seam subzone to trade the two terms against each other at the cost of the stated constants.
Partition Zone II into the sets $\mathcal G_l=\{\eta\in\text{Zone II}:g_l\le g(\eta)<2g_l\}$, $g_l=2^lg_{\min}$, $g_{\min}=\kappa_0\tau_n$, $0\le l\le L_n=O(\log n)$. On $\mathcal G_l$, Lemma~\ref{lem:variance_bounds_tau_n}\ref{vb:zone2} (with $\eta'=\eta_n^\star$, absorbing the $O(\tau_n)$ distance of $\eta_n^\star$ from $\eta_0$, and noting that each of $\|a-a_0\|^2$, $\|\theta-\theta_0\|$ and $\tau_n\sqrt{\log(1/\tau_n)}$ is bounded by $C\sqrt{\log n}\,g_l/\kappa_0$ on $\mathcal G_l$) gives
$\Var(r_\eta)\le C\log n\cdot g_l$.
Apply Lemma~\ref{lem:truncated_bernstein} with $t=ng_l/2$, $v=C\log n\,g_l$, $R=CK_0\log n$: the exponent is
\[
\frac{(ng_l/2)^2}{4nC\log n\,g_l+2CK_0\log n\cdot ng_l/2}\ \ge\ \frac{c\,n\,g_l}{\log n}
\ \ge\ \frac{c\,n\,g_{\min}}{\log n}\ =\ c\,\kappa_0\,\frac{n\tau_n}{\log n}\,,
\]
which diverges (it is of order at least $(\log n)^5$) under $n\tau_n/(\log n)^6\to\infty$, while the truncation cost $ne^{-R/(2K_0)}\le n^{-2}$ for $C$ large. Covering: on $\mathcal G_l$ take a net at resolution $\rho_l=g_l\tau_n^{2}/C'$ in $\eta$; by the crude pointwise Lipschitz bound ($G_{\tau_n}\le C\tau_n^{-2}$ on bounded designs) the net approximation error in $R_n/n$ is at most $g_l/4$ on the event $\{\mathbb P_nG_{\tau_n}\le2\E G_{\tau_n}\}$, and $\log N_l\le Cd\log(1/(\rho_l))\le C\log n$, negligible against $cng_l/\log n$. Summing the resulting bounds over $l\le L_n$ and pairing with \eqref{eq:denom_tau_n}: the posterior mass of Zone II is at most
$\sum_le^{-cng_l/(2\log n)}e^{C_\epsilon'\log n}\to0$,
since already $cng_{\min}/(2\log n)\ge(\log n)^4\gg C_\epsilon'\log n$ and the exponents grow geometrically in $l$.

Zone III-a ($\|\theta-\theta_0\|<C_1\tau_n$ and $\delta_1/2\le\|a-a_0\|$). Here Lemma~\ref{lem:risk_localization} gives the fixed separation $\Delta K(\eta)\ge c\delta_1^2/4-2\bar C\tau_n\ge c\delta_1^2/8$ for small $\tau_n$, and the Zone-I argument (Bernstein at a fixed gap with $O(\log n)$ covering cost) removes the posterior mass.

Zone III-b ($\|\theta-\theta_0\|<C_1\tau_n$, $\|a-a_0\|<\delta_1/2$, anisotropic annuli). On this tube the two-sided block bounds of Lemma~\ref{lem:aniso_information} hold (any segment between such points and $\tilde\eta_n^\star$ stays in the $(\delta_1,C_1\tau_n)$-tube, since $\|a_n^\star-a_0\|\le C\tau_n$ and $\|\vartheta_n^\star-\vartheta_0\|<C_1\tau_n$ by Proposition~\ref{prop:eta_star_bias_rate_new}), and as in the proof of that proposition the risk is strictly convex with
\[
\Delta K(\eta)\ \ge\ \tfrac{c}{4}\,\rho_n(\tilde\eta)^2,
\qquad
\rho_n(\tilde\eta)^2:=\|a-a_n^\star\|^2+\tau_n^{-1}\|\vartheta-\vartheta_n^\star\|^2,
\]
by the integral form of Taylor's theorem around $\tilde\eta_n^\star$ (where the gradient vanishes) and Young's inequality on the cross terms. Peel the tube minus the target ball into annuli
$\mathcal A_j=\{2^j\rho_{\min}\le\rho_n(\tilde\eta)<2^{j+1}\rho_{\min}\}$ with $\rho_{\min}=M_0\log n/\sqrt n$. On $\mathcal A_j$, Lemma~\ref{lem:variance_bounds_tau_n}\ref{vb:zone3} gives $\Var(r_\eta)\le C\rho_j^2$ with $\rho_j=2^j\rho_{\min}$, and Lemma~\ref{lem:truncated_bernstein} with $t=nc\rho_j^2/8$, $v=C\rho_j^2$, $R=CK_0\log n$ yields the exponent
\[
\frac{(nc\rho_j^2/8)^2}{4nC\rho_j^2+2CK_0\log n\cdot nc\rho_j^2/8}\ \ge\ \frac{c'\,n\rho_j^2}{\log n}
\ \ge\ c'M_0^2\,4^j\log n .
\]
Covering of $\mathcal A_j$ at resolution $\rho_j^2\tau_n^2/C$ again costs $\log N_j\le C\log n$, dominated once $M_0$ is large; the annuli are intersected with Zone III-b, and the peeling terminates at $\rho_j\le\delta_1$ since beyond that radius every point of the tube lies in Zone III-a. Summing over $j$ and pairing with \eqref{eq:denom_tau_n} as before gives posterior mass $\to0$ outside $\{\rho_n(\tilde\eta)<\rho_{\min}\}$ for $M_0^2>2(C_\epsilon'+d)/c'$. Coverage check (Zones I, II, III-a, III-b). A point outside the target box with $\|\theta-\theta_0\|\ge C_1\tau_n$ lies in Zone II if $\|\theta-\theta_0\|\le2\varepsilon_0$ and otherwise satisfies $\|\eta-\eta_n^\star\|\ge2\varepsilon_0-C\tau_n\ge\varepsilon_0$, hence lies in Zone I; a point with $\|\theta-\theta_0\|<C_1\tau_n$ lies in Zone III-a, Zone III-b, or the target box. The four zones therefore cover the complement of the target box.

Since $\{\rho_n<\rho_{\min}\}\subseteq\{\|a-a_n^\star\|<M_0\log n/\sqrt n\}\cap\{\|\vartheta-\vartheta_n^\star\|<M_0\sqrt{\tau_n}\log n/\sqrt n\}$ and the chart distorts norms by bounded factors, the first display of the theorem follows; the second follows by the triangle inequality with Proposition~\ref{prop:eta_star_bias_rate_new}.
\end{proof}

\subsection{Proof of Theorem~\ref{thm:effect_bvm}}

\begin{proof}[Proof of Theorem~\ref{thm:effect_bvm}]
Recall the anisotropic norming $D_n=\mathrm{diag}(n^{-1/2}I_{p+r+1},(\tau_n/n)^{1/2}I_{q-1})$ and write the local parameter $h=(h_a,h_\vartheta)=D_n^{-1}(\tilde\eta-\tilde\eta_n^\star)$, so that by Theorem~\ref{thm:tau_contraction_new} the posterior of $h$ concentrates on $\mathcal H_n=\{\|h\|\le M_0\log n\}$.

Step 1 (anisotropic LAN). Define
\[
\tilde\Delta_n=D_n\sum_{i=1}^n\nabla\tilde\ell_n(\tilde\eta_n^\star;O_i),
\qquad
\tilde I_n=n\,D_n\,I_{\tau_n}(\tilde\eta_n^\star)\,D_n .
\]
By Lemma~\ref{lem:aniso_information}, $\tilde I_n$ has $aa$-block $I_{\tau_n,aa}(\tilde\eta_n^\star)$, $\vartheta\vartheta$-block $\tau_nI_{\tau_n,\vartheta\vartheta}(\tilde\eta_n^\star)$, both with eigenvalues in $[c,C]$, and cross block $\sqrt{\tau_n}\,I_{\tau_n,a\vartheta}=O(\sqrt{\tau_n})$. We claim
\begin{equation}
\sup_{h\in\mathcal H_n}\Big|\sum_i\big\{\tilde\ell_n(\tilde\eta_n^\star+D_nh;O_i)-\tilde\ell_n(\tilde\eta_n^\star;O_i)\big\}-h^\top\tilde\Delta_n+\tfrac12h^\top\tilde I_nh\Big|=o_{P_0}(1).
\label{eq:aniso_lan}
\end{equation}
Taylor expansion to third order around $\tilde\eta_n^\star$ produces two error terms. (i) The Hessian fluctuation
$\tfrac12h^\top D_n\{\sum_i\nabla^2\tilde\ell_n(\tilde\eta_n^\star;O_i)+nI_{\tau_n}(\tilde\eta_n^\star)\}D_nh$:
blockwise, the $aa$-entry has per-observation variance $O(1)$ (envelopes \eqref{eq:derivative_envelopes}, $(j,k,l)=(2,0,2)$), so after norming its standard deviation is $O(n^{-1/2})$; the $a\vartheta$-entry has per-observation second moment $O(\tau_n^{-1})$ and norming $n^{-1/2}(\tau_n/n)^{1/2}$, giving $O(n^{-1/2})$; the $\vartheta\vartheta$-entry has per-observation second moment $O(\tau_n^{-3})$ and norming $\tau_n/n$, giving standard deviation $(\tau_n/n)\sqrt{n\tau_n^{-3}}=(n\tau_n)^{-1/2}$. Multiplying by $\|h\|^2\le M_0^2(\log n)^2$, the fluctuation term is
$O_{P_0}\big((\log n)^2\{n^{-1/2}+(n\tau_n)^{-1/2}\}\big)=o_{P_0}(1)$ under $n\tau_n/(\log n)^6\to\infty$.
(ii) The third-order remainder is bounded by $\tfrac16\sum_{(\iota_1\iota_2\iota_3)}\|h\|^3$ times the normed empirical means of the third-derivative envelopes along the segment; using \eqref{eq:derivative_envelopes} ($L^1$, $k=0,1,2,3$) and Markov's inequality, the worst term is the $\vartheta\vartheta\vartheta$ one:
$n\cdot\tau_n^{-2}\cdot(\tau_n/n)^{3/2}\,(\log n)^3=(n\tau_n)^{-1/2}(\log n)^3=o_{P_0}(1)$;
the mixed terms are smaller ($aa\vartheta$: $n\cdot1\cdot n^{-1}(\tau_n/n)^{1/2}(\log n)^3=o(1)$; $a\vartheta\vartheta$: $n\cdot\tau_n^{-1}\cdot n^{-1/2}(\tau_n/n)(\log n)^3=o(1)$; $aaa$: $O(n^{-1/2}(\log n)^3)$). (The segment stays inside the anisotropic tube where the envelopes hold because $D_n\mathcal H_n$ has $\vartheta$-extent $M_0\sqrt{\tau_n}\log n/\sqrt n\le C_1\tau_n$ eventually, by $n\tau_n/(\log n)^6\to\infty$; its $a$-extent $M_0\log n/\sqrt n$ is eventually below $\delta_1$.)

Step 2 (joint Gaussian approximation and marginalization). Exactly as in Steps~1--2 of the proof of Theorem~\ref{thm:mis_contraction_fixed_tau_new}, \eqref{eq:aniso_lan}, the prior continuity at $\tilde\eta_0$ (note $\sup_{h\in\mathcal H_n}\|D_nh\|\to0$, so the prior density ratio tends to one uniformly), the posterior localization to $\mathcal H_n$, and the uniform nondegeneracy $c\preceq\tilde I_n\preceq C$ yield
\begin{equation}
\big\|\Pi^{(\tau_n)}\big(h\in\cdot\mid\mathcal D_n\big)-\mathcal N\big(\tilde I_n^{-1}\tilde\Delta_n,\ \tilde I_n^{-1}\big)\big\|_{\mathrm{TV}}=o_{P_0}(1),
\label{eq:joint_tv_gaussian}
\end{equation}
where we use that $\tilde\Delta_n=O_{P_0}(1)$: its $a$-block has covariance $\Var_{P_0}(\nabla_a\tilde\ell_n)\le C$ and mean zero (stationarity of $\tilde\eta_n^\star$), and its $\vartheta$-block has covariance $\tau_n\Var(\nabla_\vartheta\tilde\ell_n)\le C$ and mean zero. Total variation dominates marginal total variation, so the marginal posterior of $h_a$ is within $o_{P_0}(1)$ of $\mathcal N\big([\tilde I_n^{-1}\tilde\Delta_n]_a,\ [\tilde I_n^{-1}]_{aa}\big)$. Since the cross block of $\tilde I_n$ is $O(\sqrt{\tau_n})$, the blockwise inversion identities give
$[\tilde I_n^{-1}]_{aa}=\big(\tilde I_{n,aa}\big)^{-1}+O(\tau_n)$ and
$[\tilde I_n^{-1}\tilde\Delta_n]_a=\tilde I_{n,aa}^{-1}\tilde\Delta_{n,a}+O(\sqrt{\tau_n})\,O_{P_0}(1)$,
and the total-variation distance between the two Gaussians with these parameters is $o_{P_0}(1)$ (means differing by $o_{P_0}(1)$, variances by $o(1)$, with variances bounded below). Hence
\begin{equation}
\big\|\Pi_a^{(\tau_n)}\big(h_a\in\cdot\mid\mathcal D_n\big)-\mathcal N\big(\tilde I_{n,aa}^{-1}\tilde\Delta_{n,a},\ \tilde I_{n,aa}^{-1}\big)\big\|_{\mathrm{TV}}=o_{P_0}(1).
\label{eq:marginal_bvm_pseudotrue}
\end{equation}

Step 3 (recentering at $a_0$ and the oracle limit). Write $h_a^0=\sqrt n(a-a_0)=h_a+\sqrt n(a_n^\star-a_0)$. By Proposition~\ref{prop:eta_star_bias_rate_new} the deterministic shift satisfies $\|\sqrt n(a_n^\star-a_0)\|\le C\sqrt n\,\tau_n\to0$ under the window's upper constraint; a vanishing shift changes a Gaussian by $o(1)$ in total variation (the covariance is bounded below), so \eqref{eq:marginal_bvm_pseudotrue} holds for $h_a^0$ with the same parameters. It remains to replace $(\tilde I_{n,aa},\tilde\Delta_{n,a})$ by the oracle quantities. First, $\tilde I_{n,aa}=I_{\tau_n,aa}(\tilde\eta_n^\star)\to\mathcal I_0$: by the $aa$-expansion in the proof of Lemma~\ref{lem:aniso_information}, $I_{\tau,aa}(\tilde\eta_0)=\mathcal I_0+O(\tau)$; and by the $(3,0)$ and $(2,1)$ envelopes of \eqref{eq:derivative_envelopes}, both $O(1)$, the map $\tilde\eta\mapsto I_{\tau,aa}(\tilde\eta)$ is Lipschitz on the tube with a $\tau$-free constant, so
$\|I_{\tau_n,aa}(\tilde\eta_n^\star)-I_{\tau_n,aa}(\tilde\eta_0)\|\le C\big(\|a_n^\star-a_0\|+\|\vartheta_n^\star-\vartheta_0\|\big)=O(\tau_n)$
by Proposition~\ref{prop:eta_star_bias_rate_new}. Hence $\tilde I_{n,aa}=\mathcal I_0+o(1)$. Second, define the oracle score $\dot\ell_a(O)=\nabla_a\log f^0_{d_0}(Y)$ (the score of the known-partition Gaussian regression at $a_0$) and decompose
\[
\tilde\Delta_{n,a}-\Delta_{n,a}^{\mathrm{or}}
=\frac1{\sqrt n}\sum_i\big\{\nabla_a\tilde\ell_n(\tilde\eta_n^\star;O_i)-\dot\ell_a(O_i)\big\},
\qquad \Delta_{n,a}^{\mathrm{or}}=\frac1{\sqrt n}\sum_i\dot\ell_a(O_i).
\]
Both summand families have exactly zero mean ($\E\nabla_a\tilde\ell_n(\tilde\eta_n^\star)=0$ by stationarity in the interior; $\E\dot\ell_a=0$ since the oracle model is correctly specified), so the difference has variance equal to $\Var_{P_0}\{\nabla_a\tilde\ell_n(\tilde\eta_n^\star)-\dot\ell_a\}$. This variance tends to zero: the two scores differ (i) through the gate--indicator discrepancy at the same parameter, whose conditional second moment is bounded by $C\,\E[\Phi(-|U_0|/\tau_n)\,(1+\varepsilon^2+\cdots)]=O(\tau_n)$ by the responsibility identity, Lemma~\ref{lem:gate_factor_moments} and Lemma~\ref{lem:boundary_layer_integration}; and (ii) through the parameter displacement $\tilde\eta_n^\star-\tilde\eta_0$, contributing at most the squared displacement times the $L^2$ envelopes of the $a$-score derivatives, i.e.\ $O(\tau_n^2)\cdot O(1)+O(\tau_n^2)\cdot O(\tau_n^{-1})=o(1)$. Hence $\tilde\Delta_{n,a}=\Delta_{n,a}^{\mathrm{or}}+o_{P_0}(1)$, and another vanishing perturbation of the Gaussian parameters in \eqref{eq:marginal_bvm_pseudotrue} delivers the theorem's display. Finally $\Delta_{n,a}^{\mathrm{or}}\Rightarrow\mathcal N(0,\mathcal I_0)$ by the central limit theorem for the correctly specified oracle score (information equality holds for it exactly), which gives both the stated weak limit and, by the standard argument (the credible interval endpoints converge to the Wald endpoints of the oracle model), the asymptotic coverage claim for credible intervals of continuous linear functionals.
\end{proof}

\begin{corollary}[Joint anisotropic Gaussian approximation at the pseudo-true center]
\label{cor:joint_bvm}
Let the assumptions and the smoothing window of Theorem~\ref{thm:effect_bvm} be in force. With $h=D_n^{-1}(\tilde\eta-\tilde\eta_n^\star)$ for $D_n=\mathrm{diag}\big(n^{-1/2}I_{p+r+1},(\tau_n/n)^{1/2}I_{q-1}\big)$, and with $\tilde\Delta_n=D_n\sum_{i=1}^n\nabla\tilde\ell_n(\tilde\eta_n^\star;O_i)$ and $\tilde I_n=n\,D_nI_{\tau_n}(\tilde\eta_n^\star)D_n$,
\[
\big\|\Pi^{(\tau_n)}\big(h\in\cdot\mid\mathcal D_n\big)-\mathcal N\big(\tilde I_n^{-1}\tilde\Delta_n,\ \tilde I_n^{-1}\big)\big\|_{\mathrm{TV}}\xrightarrow[n\to\infty]{P_0}0 .
\]
The centering is the pseudo-true point $\tilde\eta_n^\star$ and the covariance is the inverse working information rather than a sandwich, so no frequentist calibration at $\theta_0$ is implied, and by Remark~\ref{rem:boundary_shift} no joint analogue centered at $\theta_0$ exists on this window.
\end{corollary}

\begin{proof}
The display is established as \eqref{eq:joint_tv_gaussian} in Step~2 of the proof of Theorem~\ref{thm:effect_bvm}, from the LAN expansion \eqref{eq:aniso_lan} of Step~1, the localization of Theorem~\ref{thm:tau_contraction_new}, prior continuity at $\tilde\eta_0$, and the two-sided bounds $c\preceq\tilde I_n\preceq C$ of Lemma~\ref{lem:aniso_information}.
\end{proof}
\section{Additional simulation results}
\label{sec:additional_simulation}

This section collects the supporting material for the simulation suite of Section~\ref{sec:simulations}. It gives the full interior-schedule results for the effect and boundary blocks, the redesigned near-null study, the inactive-coordinate summary for the high-dimensional setting, the MCMC diagnostic summaries, the absorbable specification with a global treatment effect, and the stress design under which the smoothing window of Theorem~\ref{thm:effect_bvm} binds.

\subsection{The comparison study at the representative sample size}

The representative comparison is now reported in main-text Table~\ref{tab:sim_comparisons}, where it directly supports Section~\ref{sec:sim_comparisons}. That table uses a replicate set independent of the full coverage grid, so small differences such as coverage $0.944$ there against $0.952$ in Supplementary Table~\ref{tab:sim_gamma_coverage} reflect Monte Carlo error rather than a methodological discrepancy. The diagnostic below explains the optimization failure of the smoothed M estimator in that comparison.

\subsection{Optimization diagnostic for the smoothed M estimator}
\label{sec:mle_diagnostic}

The smoothed M estimator of Table~\ref{tab:sim_comparisons} maximizes the identical collapsed working likelihood the posterior is built on, over $(\beta,\gamma,\log\sigma^2)$ and a local chart of the hemisphere, by quasi-Newton search with central-difference gradients from five random hemisphere starts, keeping the best objective value and recording nonconverged runs as estimates, which mirrors a standard applied workflow. A dedicated diagnostic at $n=1000$ on the interior schedule reruns all $500$ replicates under the published protocol and again with the true direction added as one further start. The published protocol reproduces the reported attenuation (bias $-0.183$ against $-0.185$ in Table~\ref{tab:sim_comparisons}, the difference reflecting independent replicate sets). In $60$ percent of replicates the five-start search terminates in a basin whose objective is worse than the one found from the oracle start, by a median of $72$ log-likelihood units, and the attenuation is confined to these replicates. With the oracle start available the bias is $+0.002$, matching the posterior mean bias on the same data, and in only $7.4$ percent of replicates does any spurious basin beat the oracle-started fit, so the working likelihood itself has no material pathology at this design. The sign convention that reports $(-\theta,-\gamma)$ when the fitted direction lands in the wrong hemisphere never fired in these replicates. The gap is thus an optimization failure of multistart quasi-Newton search on a nonconvex objective whose information about $\theta$ concentrates in a boundary layer of width of order $\tau$, while the great-circle slice sampler traverses the sphere globally within the same likelihood. Warm starts or a much larger number of restarts repair the point estimator, and the undercoverage of the sandwich intervals in Table~\ref{tab:sim_comparisons} is the interval consequence of the same entrapped point estimates. The replication script and per-replicate objective values are included with the simulation code.

\subsection{Coverage of the regression block across the grid}

Table~\ref{tab:sim_gamma_coverage} reports empirical coverage of the central $95\%$ credible interval for $\gamma$ across the full smoothing-schedule grid, supporting Section~\ref{sec:sim_effect_block}.

\begin{table}[!htbp]
  \centering
  \caption{Empirical coverage of the central $95\%$ credible interval for the treatment-effect contrast $\gamma$, by smoothing-schedule exponent $\rho$ (rows, $\tau_n=n^{-\rho}$) and sample size $n$ (columns), under the canonical design \eqref{eq:sim_dgp} with $500$ replicates per cell and the collapsed update throughout. Interior schedules ($\rho\in\{2/3,5/6\}$) lie inside the window \eqref{eq:window}, $\rho\le1/2$ above it, $\rho=1$ below it. Nominal coverage is $0.95$. Monte Carlo standard errors are at most $0.013$.}
  \label{tab:sim_gamma_coverage}
  \begin{adjustbox}{max width=\textwidth}
  \begin{tabular}{lccccc}
    \toprule
    & \multicolumn{5}{c}{Sample size $n$}\\
    \cmidrule(lr){2-6}
    Schedule $\rho$ & $250$ & $500$ & $1000$ & $2000$ & $4000$\\
    \midrule
    $1/3$ & $0.954$ & $0.966$ & $0.956$ & $0.952$ & $0.950$\\
    $1/2$ & $0.958$ & $0.936$ & $0.944$ & $0.958$ & $0.936$\\
    $2/3$ & $0.948$ & $0.944$ & $0.952$ & $0.958$ & $0.946$\\
    $5/6$ & $0.966$ & $0.956$ & $0.936$ & $0.948$ & $0.954$\\
    $1$   & $0.962$ & $0.936$ & $0.954$ & $0.956$ & $0.950$\\
    \bottomrule
  \end{tabular}
  \end{adjustbox}
\end{table}

\subsection{Boundary behavior in detail}

This subsection gives the detailed boundary study summarized in Section~\ref{sec:sim_boundary}. This study probes Theorem~\ref{thm:tau_contraction_new} and Proposition~\ref{prop:eta_star_bias_rate_new}. At the interior schedule the boundary is recovered rapidly and identically under the symmetric and skewed designs (Table~\ref{tab:sim_boundary}). The mean angular error to $\theta_0$ falls from $0.103$ radians at $n=250$ to $0.0063$ at $n=4000$, an empirical scaling of almost exactly $n^{-1}$, while subgroup misclassification falls from $2.7$ percent to $0.16$ percent and the stability metric $\lambda_{\max}(M)$ rises from $0.986$ to $0.99994$. The error sits within a factor of roughly six to ten of the nominal localization scale $\sqrt{\tau_n/n}$, consistent with the logarithmic factors in Theorem~\ref{thm:tau_contraction_new}. The angular error to the pseudo-true direction $\theta^\star(\tau_n)$ agrees with the error to $\theta_0$ to four decimals in both designs, because the quadrature-verified smoothing shift (at most $8\times10^{-4}$ radians for the skewed design over this range of $\tau_n$) is one percent of the total error. At sample sizes and smoothing scales accessible to simulation, total boundary error is therefore dominated by posterior spread rather than by the smoothing shift, in the skewed design as well as the symmetric one. The shift itself is a deterministic population quantity, and it is quantified by the quadrature protocol of Supplementary Section~\ref{sec:pseudo_true_discussion}, where the skewed design exhibits it with a fitted log-log slope of $0.86$ in $\tau$. Monte Carlo and quadrature thus resolve complementary components of the boundary error, the spread and the shift respectively.

\subsection{Full numerical results}
\label{sec:numerical_results_sim}

Table~\ref{tab:sim_effect_full} reports the full interior-schedule results for both design variants. The per-coordinate regression metrics reduce to the treatment contrast $\gamma$, since the baseline coefficients are nuisance parameters whose draws are not stored by the suite. Boundary uncertainty is summarized by the angular errors, the misclassification rate, and $\lambda_{\max}(M)$, and the angular error is reported against both the hard-threshold direction $\theta_0$ and the pseudo-true direction $\theta^\star(\tau_n)$.

\begin{table*}
  \centering
  \caption{Interior-schedule results ($\tau_n=n^{-2/3}$, $500$ replicates per cell) for the symmetric (Gaussian) and skewed designs.}
  \begin{adjustbox}{max width=\textwidth}
  \begin{tabular}{l c ccccc cc cc}
    \toprule
    Variant & $n$ & Coverage & Bias & RMSE & CI length & sd ratio & Ang.\ to $\theta_0$ & Ang.\ to $\theta^\star$ & Misclass. & $\lambda_{\max}(M)$ \\
    \midrule
    \multirow{5}{*}{Symmetric}
      & $250$  & $0.948$ & $0.014$  & $0.151$ & $0.587$ & $1.021$ & $0.103$  & $0.103$  & $0.027$  & $0.986$ \\
      & $500$  & $0.944$ & $0.010$  & $0.108$ & $0.407$ & $1.008$ & $0.049$  & $0.049$  & $0.012$  & $0.997$ \\
      & $1000$ & $0.952$ & $-0.002$ & $0.068$ & $0.286$ & $1.004$ & $0.025$  & $0.025$  & $0.006$  & $0.999$ \\
      & $2000$ & $0.958$ & $-0.000$ & $0.051$ & $0.201$ & $1.002$ & $0.012$  & $0.012$  & $0.003$  & $0.9998$ \\
      & $4000$ & $0.946$ & $-0.001$ & $0.036$ & $0.142$ & $1.003$ & $0.0063$ & $0.0064$ & $0.0016$ & $0.9999$ \\
    \midrule
    \multirow{5}{*}{Skewed}
      & $250$  & $0.946$ & $0.011$  & $0.146$ & $0.574$ & $1.019$ & $0.107$  & $0.107$  & $0.026$  & $0.983$ \\
      & $500$  & $0.952$ & $0.009$  & $0.104$ & $0.399$ & $1.007$ & $0.056$  & $0.056$  & $0.013$  & $0.996$ \\
      & $1000$ & $0.954$ & $0.001$  & $0.071$ & $0.280$ & $1.003$ & $0.027$  & $0.027$  & $0.006$  & $0.999$ \\
      & $2000$ & $0.956$ & $-0.000$ & $0.050$ & $0.197$ & $1.002$ & $0.013$  & $0.013$  & $0.003$  & $0.9998$ \\
      & $4000$ & $0.974$ & $0.004$  & $0.034$ & $0.139$ & $1.002$ & $0.0066$ & $0.0066$ & $0.0015$ & $0.9999$ \\
    \bottomrule
  \end{tabular}
  \end{adjustbox}
  \label{tab:sim_effect_full}
\end{table*}

The interior-schedule boundary summary discussed in Section~\ref{sec:sim_boundary} is collected in Table~\ref{tab:sim_boundary}.

\begin{table}[!htbp]
  \centering
  \caption{Boundary behavior at the interior schedule $\tau_n=n^{-2/3}$ under the canonical design \eqref{eq:sim_dgp}. Rows report the mean angular error to $\theta_0$, the mean angular error to the pseudo-true direction $\theta^\star(\tau_n)$, the predicted localization scale $\sqrt{\tau_n/n}$, the subgroup misclassification probability, and the stability metric $\lambda_{\max}(M)$. Symmetric design; the skewed design is indistinguishable (Table~\ref{tab:sim_effect_full}).}
  \label{tab:sim_boundary}
  \begin{adjustbox}{max width=\textwidth}
  \begin{tabular}{lccccc}
    \toprule
    & \multicolumn{5}{c}{Sample size $n$}\\
    \cmidrule(lr){2-6}
    Metric & $250$ & $500$ & $1000$ & $2000$ & $4000$\\
    \midrule
    Angular error to $\theta_0$ & $0.103$ & $0.049$ & $0.025$ & $0.012$ & $0.0063$\\
    Angular error to $\theta^\star(\tau_n)$ & $0.103$ & $0.049$ & $0.025$ & $0.012$ & $0.0064$\\
    Localization scale $\sqrt{\tau_n/n}$ & $0.0100$ & $0.0056$ & $0.0032$ & $0.0018$ & $0.0010$\\
    Misclassification & $0.027$ & $0.012$ & $0.006$ & $0.003$ & $0.0016$\\
    $\lambda_{\max}(M)$ & $0.986$ & $0.997$ & $0.999$ & $0.9998$ & $0.9999$\\
    \bottomrule
  \end{tabular}
  \end{adjustbox}
\end{table}

\subsection{Operating characteristics of the protocol in detail}

This subsection reports the detailed protocol operating characteristics summarized in Section~\ref{sec:sim_protocol}. This study evaluates the reporting protocol of Section~\ref{sec:reporting_protocol} as a decision rule. We vary the true contrast over $\gamma_0\in\{0,0.25,0.5,1,1.5,2\}$ at the clinical threshold $\delta=1$ and $p_{\mathrm{report}}=0.9$, with $n\in\{500,1000,2000\}$ at the interior schedule and $500$ replicates per cell (Table~\ref{tab:sim_protocol}). The protocol produced no false reports in the $500$ replicates. At every subclinical contrast $\gamma_0\in\{0,0.25,0.5\}$ and every $n$ the report probability is $0.000$, bounding the false-report probability near zero at these designs rather than establishing that it is zero, while the boundary posterior remains diffuse under the null ($\lambda_{\max}(M)\approx0.29$ against the dispersion floor $0.2$) and concentrates as the contrast grows, reaching $0.9998$ at $\gamma_0=2$. At the knife edge $\gamma_0=\delta$ the report probability is $0.106$, $0.078$, and $0.094$ at the three sample sizes, against the limit $1-p_{\mathrm{report}}=0.10$ predicted by Proposition~\ref{prop:protocol_consistency}, and the mean posterior tail probability is $0.49$ to $0.50$ with a Brier score of $0.33$, matching the value $1/3$ implied by the uniform limit. At $\gamma_0\in\{1.5,2\}$ the report probability is $0.996$ to $1.000$. The membership probabilities $q(z)$ are well calibrated, with reliability bins tracking the diagonal to within $0.03$ at the knife edge. The comparison with the test-then-report workflow of \citet{Fan2017} isolates the value of the joint posterior. The test itself is well behaved, with size $0.034$ to $0.064$ under the null and power above $0.93$ at $\gamma_0\ge0.5$. The plug-in confidence interval reported after the test is not. Under the null it covers the true contrast $34$ percent of the time, and under strong signal its coverage degrades with $n$, reaching $0.078$ at $\gamma_0=1.5$ and $n=2000$. The protocol and the test agree about when heterogeneity exists. They differ in that the posterior quantifies the effect honestly after the decision, and the plug-in workflow does not.

\subsection{Protocol diagnostics at the representative sample size}

Main-text Table~\ref{tab:sim_protocol} compares the reporting and interval operating characteristics at $n=1000$. Table~\ref{tab:protocol_ope_char} records the accompanying posterior diagnostics. The knife-edge Brier score $0.327$ matches the value $1/3$ implied by the uniform limit of Proposition~\ref{prop:protocol_consistency}, while $\lambda_{\max}(M)$ traces the transition from a diffuse boundary under the null to a concentrated boundary under strong heterogeneity.

\begin{table}[!htbp]
  \centering
  \caption{Posterior diagnostics for the reporting experiment of main-text Table~\ref{tab:sim_protocol} at $n=1000$. For each true contrast $\gamma_0$ we report the Brier score of $\Pi(H_\delta\mid\mathcal D_n)$ and the mean boundary stability $\lambda_{\max}(M)$.}
  \label{tab:protocol_ope_char}
  \begin{adjustbox}{max width=0.65\textwidth}
  \begin{tabular}{lcc}
    \toprule
    $\gamma_0$ & Brier of $\Pi(H_\delta)$ & $\lambda_{\max}(M)$\\
    \midrule
    $0$    & $0.000$ & $0.287$\\
    $0.25$ & $0.000$ & $0.493$\\
    $0.5$  & $0.000$ & $0.843$\\
    $1$    & $0.327$ & $0.989$\\
    $1.5$  & $0.000$ & $0.998$\\
    $2$    & $0.000$ & $0.999$\\
    \bottomrule
  \end{tabular}
  \end{adjustbox}
\end{table}

\subsection{Near-null treatment effect}
\label{sec:simulation_under_null}

As discussed in Section~\ref{sec:reporting_protocol}, the boundary direction $\theta$ is weakly identified when the heterogeneity contrast is near zero. The redesigned near-null study sets $\gamma_0=0$ in the canonical design \eqref{eq:sim_dgp} with the parametric linear baseline and evaluates it across the sample-size grid at the interior schedule. The anticipated pattern is confirmed. The posterior of $\gamma$ has small bias ($0.001$ to $0.003$) and near-nominal coverage ($0.956$ to $0.960$), while the boundary posterior stays diffuse because the likelihood carries essentially no directional information, with $\lambda_{\max}(M)\approx0.29$ against the dispersion floor $1/q=0.2$ (Table~\ref{tab:sim_near_null}). The reporting protocol suppresses subgroup claims throughout, and no replicate at any sample size reported a subgroup, so the Bayes action is $\mathsf a_0$ everywhere and no boundary summary is issued.

\begin{table*}
  \centering
  \caption{Redesigned near-null study, $\gamma_0=0$ in the canonical design \eqref{eq:sim_dgp} with the parametric linear baseline at the interior schedule. For the effect contrast $\gamma$ we report bias, RMSE, and CP of the central $95\%$ credible interval. We also report the protocol report probability $P(\text{report})$ and the mean stability metric $\lambda_{\max}(M)$, which sits near its lower limit under weak identification.}
  \begin{adjustbox}{max width=\textwidth}
  \begin{tabular}{c rrr rr}
    \toprule
    & \multicolumn{3}{c}{Effect contrast $\gamma$} & \multicolumn{2}{c}{Protocol} \\
    \cmidrule(lr){2-4} \cmidrule(lr){5-6}
    \textbf{Sample size $n$} & Bias & RMSE & CP & $P(\text{report})$ & $\lambda_{\max}(M)$ \\
    \midrule
    $500$  & $0.003$ & $0.130$ & $0.958$ & $0.000$ & $0.295$ \\
    $1000$ & $0.002$ & $0.088$ & $0.960$ & $0.000$ & $0.287$ \\
    $2000$ & $0.001$ & $0.066$ & $0.956$ & $0.000$ & $0.292$ \\
    \bottomrule
  \end{tabular}
  \end{adjustbox}
  \label{tab:sim_near_null}
\end{table*}

\subsection{Inactive change-plane coordinates}
\label{sec:sim_inactive}

The high-dimensional study of Section~\ref{sec:simulation_variable_selection} augments the boundary index with $50$ pure-noise covariates whose true loadings are zero, and fits the normalized horseshoe prior on $\theta$. For the resulting $q=55$ boundary coordinates the suite stores boundary summaries rather than coordinatewise credible intervals, so we report the accuracy of the estimated principal direction on the active coordinates together with the subgroup misclassification and the boundary stability. Table~\ref{tab:sim_inactive} shows that the noise coordinates are absorbed cleanly. The angular error on the active coordinates falls from $0.351$ at $n=250$ to $0.022$ at $n=4000$, misclassification falls from $10.7$ to $0.9$ percent, and $\lambda_{\max}(M)$ rises toward one, so the horseshoe recovers the active boundary while shrinking the irrelevant modifiers toward zero.

\begin{table*}
  \centering
  \caption{High-dimensional setting ($50$ pure-noise boundary covariates, normalized horseshoe prior, $200$ replicates). Angular error of the posterior principal direction on the active coordinates, subgroup misclassification, and boundary stability.}
  \begin{adjustbox}{max width=\textwidth}
  \begin{tabular}{c rrr}
    \toprule
    \textbf{Sample size $n$} & Active angular error & Misclassification & $\lambda_{\max}(M)$ \\
    \midrule
    $250$  & $0.351$ & $0.107$ & $0.788$ \\
    $500$  & $0.095$ & $0.033$ & $0.948$ \\
    $1000$ & $0.044$ & $0.016$ & $0.982$ \\
    $2000$ & $0.026$ & $0.010$ & $0.991$ \\
    $4000$ & $0.022$ & $0.009$ & $0.987$ \\
    \bottomrule
  \end{tabular}
  \end{adjustbox}
  \label{tab:sim_inactive}
\end{table*}

\subsection{MCMC diagnostics}
\label{sec:sim_diagnostics}

The uniformity just described required the collapsed update. An earlier run of the wide-gate cells with the augmented sampler failed convergence diagnostics, with rank-normalized $\widehat R$ for $\gamma$ between $1.38$ and $1.84$, bulk boundary-functional effective sample sizes as low as $9$, and spurious attenuation of the posterior mean (Table~\ref{tab:sim_diagnostics} retains those fits). The collapsed rerun gives $\widehat R=1.000$ with effective sample sizes near twenty thousand and the nominal operating characteristics reported above. The wide-gate regime is thus where computation rather than the estimand degrades, the failure mode the diagnostics workflow of Section~\ref{sec:inference} is designed to catch.

Each fit uses four dispersed chains, and convergence is assessed by rank-normalized $\widehat R$ and effective sample sizes (ESS) computed on sign-invariant functionals of the boundary direction, namely $(z^\top\theta)^2$, the entries of $\theta\theta^\top$, and the membership probabilities $q(z)$ on a covariate grid. Table~\ref{tab:sim_diagnostics} contrasts the collapsed update at $\tau=0.035$ with the augmented update at the wide gate $\tau=0.1$. The collapsed update is uniformly reliable, with $\widehat R$ for $\gamma$ at $1.000$, bulk ESS above $18000$, and boundary diagnostics well inside their targets at every sample size. The augmented update fails at the wide gate, with $\widehat R$ for $\gamma$ between $1.38$ and $1.44$, boundary $\widehat R$ between $1.64$ and $1.73$, and bulk boundary ESS as low as $9$. The failure is the computational finding reported in Section~\ref{sec:sim_effect_block}.

\begin{table*}
  \centering
  \caption{Convergence diagnostics per fit (four dispersed chains, pooled over replicates). The collapsed update is uniformly reliable, including at the sharpest schedules (not shown, $\widehat R$ at $1.000$ down to $\tau_n=0.00025$). The augmented update fails at wide gates, which is the computational finding reported in Section~\ref{sec:sim_effect_block}.}
  \begin{adjustbox}{max width=\textwidth}
  \begin{tabular}{c cccc cccc}
    \toprule
    & \multicolumn{4}{c}{Collapsed update, $\tau=0.035$} & \multicolumn{4}{c}{Augmented update, $\tau=0.1$} \\
    \cmidrule(lr){2-5} \cmidrule(lr){6-9}
    $n$ & $\widehat R_\gamma$ & bulk ESS$_\gamma$ & $\max\widehat R_\theta$ & $\min$ bulk ESS$_\theta$ & $\widehat R_\gamma$ & bulk ESS$_\gamma$ & $\max\widehat R_\theta$ & $\min$ bulk ESS$_\theta$ \\
    \midrule
    $250$  & $1.000$ & $18252$ & $1.003$ & $2463$ & $1.384$ & $450$  & $1.732$ & $9$  \\
    $500$  & $1.000$ & $21839$ & $1.002$ & $3544$ & $1.396$ & $711$  & $1.666$ & $11$ \\
    $1000$ & $1.000$ & $22828$ & $1.001$ & $4499$ & $1.403$ & $1145$ & $1.642$ & $15$ \\
    $2000$ & $1.000$ & $23006$ & $1.001$ & $5234$ & $1.435$ & $1523$ & $1.725$ & $16$ \\
    $4000$ & $1.000$ & $23118$ & $1.001$ & $5557$ & $1.419$ & $1548$ & $1.715$ & $18$ \\
    \bottomrule
  \end{tabular}
  \end{adjustbox}
  \label{tab:sim_diagnostics}
\end{table*}

The monitored functionals are $\widehat R$ and bulk and tail ESS for $\gamma$ and for the sign-invariant boundary functionals. The collapsed rerun of the wide-gate cells gives $\widehat R=1.000$ with bulk effective sample sizes between $19{,}000$ and $23{,}000$, and the same wide-gate failure mechanism affected the oversmoothed application fits of Section~\ref{sec:analysis} before their collapsed rerun.

\subsection{Robustness and timing tables}
\label{sec:sim_robust_timing_tables}

The full grids underlying the main-text summary are reported here. Three departures from the canonical design test robustness and scale, the first replacing the linear baseline by the nonlinear mean
\[
\mu_0(W_i)=1.0-0.3\,W_{i2}W_{i3}-0.4\exp(-|W_{i4}|)-0.65\log(W_{i5}^2)+0.4\sin(\pi W_{i2}),
\]
which the parametric baseline cannot capture but the semiparametric BART variant of Section~\ref{sec:semiparametric} absorbs, the second augmenting the boundary index with $50$ pure-noise covariates under the normalized horseshoe prior on $\theta$ (a computational extension outside the formal theory, Remark~\ref{rem:scope_theory}), and the third profiling per-sweep cost with and without the collapsed update. Under the nonlinear baseline the parametric specification degrades as misspecification theory predicts, with coverage for $\gamma$ falling from $0.845$ at $n=250$ to $0.155$ at $n=4000$, while the BART variant holds $0.92$ to $0.96$ at every $n$. In the high-dimensional setting the normalized horseshoe delivers coverage $0.93$ to $0.96$ for $\gamma$, with the angular error on the active coordinates falling from $0.35$ to $0.022$ and subgroup misclassification from $10.7$ to $0.9$ percent. The collapsed update costs $2.3$ to $2.6$ times the augmented sweep, both linear in $n$ ($2.7$ versus $1.0$ seconds per thousand sweeps at $n=4000$ on one core).

Table~\ref{tab:sim_robust_timing} collects the robustness and scalability summary discussed in Section~\ref{sec:simulation_variable_selection}, covering the nonlinear-baseline and high-dimensional departures from the canonical design together with the per-sweep timing of the sampler.

\begin{table}[!htbp]
  \centering
  \caption{Robustness and scalability summary. The nonlinear-baseline rows report the coverage of the $95\%$ credible interval for $\gamma$ under the parametric and BART baselines when $\mu_0$ is nonlinear. The high-dimensional rows report the coverage for $\gamma$ and the mean angular error on the active coordinates under the normalized horseshoe with $50$ noise covariates. The timing row reports seconds per $1000$ sweeps with the collapsed update.}
  \label{tab:sim_robust_timing}
  \begin{adjustbox}{max width=\textwidth}
  \begin{tabular}{llccccc}
    \toprule
    & & \multicolumn{5}{c}{Sample size $n$}\\
    \cmidrule(lr){3-7}
    Setting & Quantity & $250$ & $500$ & $1000$ & $2000$ & $4000$\\
    \midrule
    Nonlinear baseline (Parametric) & $\gamma$ coverage & $0.845$ & $0.775$ & $0.565$ & $0.323$ & $0.155$\\
    Nonlinear baseline (BART) & $\gamma$ coverage & $0.920$ & $0.945$ & $0.960$ & $0.939$ & $0.940$\\
    High dimension (horseshoe) & $\gamma$ coverage & $0.930$ & $0.960$ & $0.955$ & $0.940$ & $0.955$\\
    High dimension (horseshoe) & active angular error & $0.351$ & $0.095$ & $0.044$ & $0.026$ & $0.022$\\
    Scalability & sec / $1000$ sweeps & $0.181$ & $0.364$ & $0.680$ & $1.287$ & $2.672$\\
    \bottomrule
  \end{tabular}
  \end{adjustbox}
\end{table}

\subsection{The absorbable specification with a global treatment effect}
\label{sec:absorbable_results}

This subsection reports the fourth departure of Section~\ref{sec:simulation_variable_selection}, which fits the recommended absorbable specification of Section~\ref{sec:model_identification} on data carrying a genuine global treatment effect. The data generating process is the canonical design \eqref{eq:sim_dgp} augmented by an ungated term $\delta_0 X_i$ with $\delta_0=1$, so that $Y_i=W_i^\top\beta_0+\delta_0 X_i+\gamma_0 X_i\mathbbm 1\{Z_i^\top\theta_0\ge0\}+\varepsilon_i$ under the canonical covariates, with $\gamma_0=2$ in the absorbable variant and $\gamma_0=-2$ in the opposite-signed variant. The fit places the treatment as the last column of the design $W$, so the absorbable specification carries the ungated effect $\delta$ in the baseline block, $\gamma$ is the between-subgroup contrast in the treatment effect, and $\delta+\gamma$ is the effect inside the subgroup. The hemisphere convention is imposed by the canonicalized full-sphere sampler with profile-scan initialization of Supplementary Section~\ref{sec:algorithm_details}. Each cell uses $500$ replicates at the fixed gate $\tau=0.1$ and at the interior schedule $\tau_n=n^{-2/3}$, and the results are collected in Table~\ref{tab:sim_absorbable}.

The interior-schedule cells reproduce the oracle-grade behavior of the canonical effect study. For $n\ge500$ the posterior mean bias of the contrast stays below $0.005$ in absolute value and the ratio of its posterior standard deviation to the known-subgroup oracle standard deviation lies between $1.00$ and $1.02$, with coverage of the $95\%$ credible interval from $0.936$ to $0.950$ for $\gamma$ and from $0.944$ to $0.958$ for $\delta$. The fixed $\tau=0.1$ cells show the expected small transported bias in the contrast, ranging from $+0.011$ to $+0.026$ across the grid, of the same sign and magnitude as the oversmoothing shift documented for the canonical design in Section~\ref{sec:sim_effect_block}. The opposite-signed variant is symmetric, its contrast and outside-effect biases reversing sign while coverage stays nominal. At $n=250$ the posterior for the contrast is wider than the oracle, with an sd ratio of $1.38$ at the fixed gate and $2.06$ at the interior schedule, because with weak boundary information the canonicalized posterior retains orbit-smearing mass around the two mirror modes, a conservative rather than an anticonservative deviation, with coverage there of $0.958$ to $0.968$.

\begin{table}[!htbp]
  \centering
  \caption{Absorbable specification with a global treatment effect, the canonical design \eqref{eq:sim_dgp} augmented by $\delta_0 X_i$ with $\delta_0=1$, $500$ replicates per cell. The variant \textup{absorbable} has $\gamma_0=2$ and \textup{absorbable\_opp} has $\gamma_0=-2$. Columns report the smoothing scale $\tau$, empirical coverage and posterior-mean bias of the central $95\%$ credible interval for the between-subgroup contrast $\gamma$, the ratio of the posterior standard deviation of $\gamma$ to the known-subgroup oracle standard deviation, coverage and bias for the ungated effect $\delta$, and the mean angular error to $\theta_0$. Schedule rows use the interior schedule $\tau_n=n^{-2/3}$. Cell-averaged rank-normalized $\widehat R$ for $\gamma$ is at most $1.003$, with per-replicate values above $1.05$ in three of six thousand replicates, all at $n=250$ and coverage Monte Carlo standard errors are about $0.008$ to $0.012$.}
  \label{tab:sim_absorbable}
  \begin{adjustbox}{max width=\textwidth}
  \begin{tabular}{l c c cc c cc c}
    \toprule
    Variant & $n$ & $\tau$ & $\gamma$ cover & $\gamma$ bias & sd ratio & $\delta$ cover & $\delta$ bias & Ang.\ to $\theta_0$\\
    \midrule
    \multirow{10}{*}{Absorbable}
      & $250$  & $0.100$ & $0.968$ & $0.011$  & $1.379$ & $0.964$ & $0.000$  & $0.112$\\
      & $250$  & $0.025$ & $0.958$ & $-0.001$ & $2.065$ & $0.954$ & $0.000$  & $0.108$\\
      & $500$  & $0.100$ & $0.956$ & $0.025$  & $1.026$ & $0.936$ & $-0.020$ & $0.064$\\
      & $500$  & $0.016$ & $0.946$ & $-0.000$ & $1.016$ & $0.950$ & $-0.005$ & $0.051$\\
      & $1000$ & $0.100$ & $0.936$ & $0.026$  & $1.020$ & $0.938$ & $-0.015$ & $0.040$\\
      & $1000$ & $0.010$ & $0.942$ & $-0.004$ & $1.008$ & $0.952$ & $0.000$  & $0.025$\\
      & $2000$ & $0.100$ & $0.922$ & $0.021$  & $1.017$ & $0.950$ & $-0.012$ & $0.025$\\
      & $2000$ & $0.006$ & $0.950$ & $0.004$  & $1.003$ & $0.950$ & $-0.002$ & $0.012$\\
      & $4000$ & $0.100$ & $0.950$ & $0.021$  & $1.016$ & $0.940$ & $-0.014$ & $0.017$\\
      & $4000$ & $0.004$ & $0.948$ & $-0.001$ & $1.001$ & $0.944$ & $-0.000$ & $0.006$\\
    \midrule
    \multirow{2}{*}{Absorbable\_opp}
      & $1000$ & $0.100$ & $0.950$ & $-0.025$ & $1.019$ & $0.954$ & $0.015$  & $0.038$\\
      & $1000$ & $0.010$ & $0.936$ & $0.001$  & $1.009$ & $0.958$ & $-0.001$ & $0.025$\\
    \bottomrule
  \end{tabular}
  \end{adjustbox}
\end{table}

\subsection{The stress design and the smoothing window}
\label{sec:stress_design}

The simulations of Section~\ref{sec:sim_effect_block} find the window of Theorem~\ref{thm:effect_bvm} conservative for the canonical design \eqref{eq:sim_dgp}, whose transported smoothing bias is intrinsically small. A systematic quadrature search over the canonical model class explains why and locates the exception. The mixture working likelihood classifies boundary points by their residuals once the covariate distribution separates the two sides of the boundary, so the transported bias is largest at moderate signal-to-noise ratio with covariate mass concentrated asymmetrically inside the smoothing gray zone. The stress design realizes that configuration. It keeps the canonical structure of \eqref{eq:sim_dgp}, with boundary covariates $Z_i=W_i$, Gaussian coordinates $W_{i2},W_{i3},W_{i4}$, treatment $X_i\sim\mathrm{Bernoulli}(0.5)$, and Gaussian error, but replaces the fifth boundary covariate by the two-component mixture $S_i=0.37+0.05\,N_i$ with probability $0.3$ and $S_i=0.28+1.1\,N_i$ with probability $0.7$, where $N_i\sim\mathcal N(0,1)$, and sets $\theta_0\propto(0.30,0.10,-0.06,0.08,-1.0)$ normalized to the unit sphere, $\gamma_0=2$, and $\sigma_0^2=2.25$. The boundary then sits at $S$ near $0.30$ with the narrow mixture component parked just inside the smoothing gray zone on one side, giving a subgroup prevalence near $0.45$ and a boundary-score standard deviation near $0.89$. The pseudo-true path $\tau\mapsto\gamma^\star(\tau)$ was computed for this design by the quadrature protocol of Supplementary Section~\ref{sec:pseudo_true_discussion} before any Monte Carlo was run, so the comparison of empirical bias with predicted shift in Table~\ref{tab:sim_stress} is a genuine out-of-sample prediction.

On the interior schedule $\tau_n=n^{-2/3}$ coverage of the central $95\%$ credible interval for $\gamma$ is nominal, from $0.934$ to $0.956$, while under deliberate oversmoothing it degrades in step with the ratio of the transported bias to the shrinking interval, to $0.902$ at $(n,\tau)=(4000,0.1)$, $0.866$ at $(4000,0.15)$, and $0.850$ at $(8000,0.1)$. The empirical posterior-mean bias agrees with the predicted pseudo-true shift to within $0.008$ at every cell, within about one Monte Carlo standard error at the $\tau\ge0.1$ cells that drive the coverage loss and within three at the narrower gates, confirming that the window binds where the precomputable pseudo-true path locates it.

\begin{table}[!htbp]
  \centering
  \caption{Stress design of Supplementary Section~\ref{sec:stress_design}, $500$ replicates per cell with the collapsed update throughout. Columns report empirical coverage of the central $95\%$ credible interval for $\gamma$, the empirical posterior-mean bias, the quadrature pseudo-true shift $\gamma^\star(\tau)-\gamma_0$ computed before the Monte Carlo, the ratio of the posterior standard deviation of $\gamma$ to the oracle standard deviation, and the mean angular error to $\theta_0$. The smallest $\tau$ at each $n$ is the interior schedule $\tau_n=n^{-2/3}$. Nominal coverage is $0.95$. Coverage Monte Carlo standard errors are about $0.009$ to $0.016$ and bias Monte Carlo standard errors about $0.002$ to $0.006$, and the predicted shift for the interior schedule at $n=8000$ was not computed.}
  \label{tab:sim_stress}
  \begin{adjustbox}{max width=\textwidth}
  \begin{tabular}{c c ccccc}
    \toprule
    $n$ & $\tau$ & Coverage & Bias & Predicted shift & sd ratio & Angular error\\
    \midrule
    \multirow{4}{*}{$1000$}
      & $0.010$ & $0.934$ & $0.0081$ & $0.0063$ & $1.022$ & $0.0274$\\
      & $0.035$ & $0.938$ & $0.0156$ & $0.0210$ & $1.035$ & $0.0323$\\
      & $0.100$ & $0.950$ & $0.0428$ & $0.0469$ & $1.071$ & $0.0697$\\
      & $0.150$ & $0.934$ & $0.0597$ & $0.0598$ & $1.089$ & $0.0948$\\
    \midrule
    \multirow{4}{*}{$4000$}
      & $0.004$ & $0.954$ & $0.0042$ & $0.0025$ & $1.006$ & $0.0064$\\
      & $0.035$ & $0.942$ & $0.0132$ & $0.0210$ & $1.025$ & $0.0140$\\
      & $0.100$ & $0.902$ & $0.0451$ & $0.0469$ & $1.059$ & $0.0406$\\
      & $0.150$ & $0.866$ & $0.0548$ & $0.0598$ & $1.079$ & $0.0668$\\
    \midrule
    \multirow{2}{*}{$8000$}
      & $0.0025$ & $0.956$ & $0.0039$ & ---      & $1.003$ & $0.0030$\\
      & $0.100$  & $0.850$ & $0.0460$ & $0.0469$ & $1.058$ & $0.0377$\\
    \bottomrule
  \end{tabular}
  \end{adjustbox}
\end{table}

A companion false-report check confirms that the protocol does not manufacture a reportable subgroup from a genuinely homogeneous benefit. With outside effect $\delta_0=1$ and zero heterogeneity contrast $\gamma_0=0$, fit under the absorbable specification of Section~\ref{sec:model_identification} with the reporting inputs of Section~\ref{sec:sim_protocol}, we ran $500$ replicates at each of $n\in\{500,1000,2000\}$ at the interior schedule, imposing the hemisphere convention by the canonicalized full-sphere update with profile-scan initialization of Supplementary Section~\ref{sec:algorithm_details}. The protocol produced at most one false report across the $1500$ replicates, with report probability $0.000$, $0.000$, and $0.002$ at the three sample sizes (Table~\ref{tab:sim_protocol_null}). The honest null posterior is nonetheless active rather than inert. Its tail probability $\Pi(|\gamma|\ge1)$ sits between $0.651$ and $0.673$ and its direction posterior partially concentrates on noise-carved boundaries, with $\lambda_{\max}(M)$ rising from $0.688$ to $0.727$ across the grid, because under the absorbable specification the model retains the freedom to carve noise subgroups when none exist. The tail probability alone is therefore not a safe evidence summary in this regime, and the demanding bar $p_{\mathrm{report}}=0.9$ is what holds the false-report rate at $0.000$ to $0.002$. An earlier hemisphere-truncated run of this study reported a much smaller null tail near $0.18$ with $\lambda_{\max}(M)$ between $0.42$ and $0.45$, an artifact of chains trapped in a near-homogeneous basin, and the corrected sampler documented in Supplementary Section~\ref{sec:algorithm_details}, which has maximum rank-normalized $\widehat R$ for $\gamma$ of $1.045$ across the $1500$ replicates, while the sign-invariant direction functionals show elevated $\widehat R$ (up to $1.13$, in $7.3$ percent of replicates), reflecting slow exploration of the unidentified direction under the null rather than a failure on the reported functionals, is the one reported here.

\begin{table}[!htbp]
  \centering
  \caption{Protocol false-report check under a homogeneous treatment effect ($\delta_0=1$, $\gamma_0=0$) in the absorbable specification, $500$ replicates per cell at the interior schedule. Columns report the report probability $P(\text{report})$, the mean posterior tail probability $\Pi(H_\delta\mid\mathcal D_n)=\Pi(|\gamma|\ge1\mid\mathcal D_n)$, and the mean boundary stability $\lambda_{\max}(M)$. Fits impose the hemisphere convention by the canonicalized full-sphere update with profile-scan initialization of Supplementary Section~\ref{sec:algorithm_details}, with maximum rank-normalized $\widehat R$ for $\gamma$ of $1.045$ across the $1500$ replicates.}
  \label{tab:sim_protocol_null}
  \begin{adjustbox}{max width=\textwidth}
  \begin{tabular}{cccc}
    \toprule
    $n$ & $P(\text{report})$ & $\Pi(H_\delta\mid\mathcal D_n)$ & $\lambda_{\max}(M)$\\
    \midrule
    $500$  & $0.000$ & $0.651$ & $0.688$\\
    $1000$ & $0.000$ & $0.652$ & $0.707$\\
    $2000$ & $0.002$ & $0.673$ & $0.727$\\
    \bottomrule
  \end{tabular}
  \end{adjustbox}
\end{table}
\section{Additional Analyses of the PREMIER Trial}
\label{sec:additional_analysis_results}

This section collects the supporting material for the two-part PREMIER analysis of Section~\ref{sec:analysis}. The first subsections cover the secondary binary-profile analysis, its decision summaries, profiles, posterior principal direction, and smoothing sensitivity. The later subsections report the primary continuous-covariate analysis, its arm-specific sensitivity, boundary loadings, membership distributions, and convergence diagnostics.

\subsection{Posterior principal direction for the secondary analysis}

Table~\ref{tab:theta_hat_hs_methods} reports the posterior principal direction estimates $\hat{\theta}$ under the horseshoe prior for the secondary binary analysis. Under both baselines the largest loading falls on an interaction term rather than a main effect, and the two baselines agree on it, placing the dominant weight on \texttt{female\_african\_american} ($-0.924$ under BART and $-0.825$ under the parametric baseline) with a secondary parametric loading on \texttt{african\_american\_baseline\_hypertension} ($0.306$). This agreement of the point directions should not be read as a concentrated boundary. The second moment of the boundary direction is diffuse at this smoothing value, with $\lambda_{\max}(M)=0.21$ under BART and $0.18$ under the parametric baseline against the dispersion floor $1/11\approx0.09$, and the Bayes action of Table~\ref{tab:premier_decision_binary} withholds a report, so under the purely discrete $Z$ of this analysis the direction remains set-identified and the loadings are descriptive. Because $\hat{\theta}$ is a direction, its global sign is arbitrary and individual coordinates should not be read as marginal causal effects. The table records which combinations of baseline features would contribute most strongly to any residual heterogeneity.

\begin{table*}
  \centering
  \caption{Posterior principal direction estimates of the effect-modifier direction coefficients $\hat{\theta}$ under the horseshoe prior, by baseline, for the secondary binary analysis.}
  \begin{adjustbox}{width=0.8\textwidth}
  \begin{tabular}{l r r}
    \toprule
    \textbf{Covariate}  & \textbf{BART (HS)} & \textbf{Parametric (HS)} \\
    \midrule
    \texttt{intercept} & 0.123 & 0.269 \\
    \texttt{age\_50plus}  & 0.222 & 0.142 \\
    \texttt{female} & -0.015 & 0.134 \\
    \texttt{african\_american}  & -0.086 & -0.063 \\
    \texttt{baseline\_hypertension} & 0.124 & 0.238 \\
    \texttt{age\_50plus\_female}  & 0.060 & 0.072 \\
    \texttt{age\_50plus\_african\_american} & -0.053 & -0.115 \\
    \texttt{age\_50plus\_baseline\_hypertension} & 0.092 & 0.090 \\
    \texttt{female\_african\_american} & -0.924 & -0.825 \\
    \texttt{female\_baseline\_hypertension} & 0.086 & 0.170 \\
    \texttt{african\_american\_baseline\_hypertension} & 0.192 & 0.306 \\
    \bottomrule
  \end{tabular}
  \end{adjustbox}
  \label{tab:theta_hat_hs_methods}
\end{table*}

\subsection{The secondary binary-profile analysis}
\label{sec:analysis_secondary}
This analysis is retained from the original study design and uses the purely binary profile covariates. Following \citet{Svetkey2005} we include self-reported race, sex, hypertension status, and age above fifty as the baseline covariates $W_i$ and $Z_i$, together with their two-way interactions. Because every component of $Z$ is a binary indicator or an interaction of indicators, the score $Z^\top\theta$ takes at most $16$ distinct values, the boundary direction is only set-identified, and the analysis is inference on the induced partition under $\tau$-regularization rather than on a unique direction (Remark~\ref{rem:nonregular_rates}). We report it at $\tau=0.035$, the upper endpoint of the nominal window for $n=810$, using the same event $H_\delta=\{|\gamma|\ge3\}$ and threshold $p_{\mathrm{report}}=0.9$ as the primary analysis. As in the primary analysis, the treatment indicator is a column of $W_i$, so the baseline carries the treatment effect $\delta$ outside the subgroup and $\gamma$ is the difference in treatment effect between the subgroup and its complement.

At $\tau=0.035$ the posterior probability of clinically meaningful heterogeneity is $0.52$ under the parametric baseline and $0.48$ under BART, so the Bayes action is $\mathsf a_0$ under both, with posterior means of $\gamma$ of $-0.70$ and $-0.50$ mmHg and wide credible intervals (Table~\ref{tab:premier_decision_binary}, with the accompanying membership and density figures). Main-text Table~\ref{tab:premier_profile_summary} translates the membership probabilities $q(Z_i)$ into the clinically readable profiles that make this diffuse boundary concrete.

The profile compression documented in the main text is the discrete-covariate signature of a diffuse boundary. Both baselines place the leading loading of the posterior principal direction on the interaction of female with African American (Table~\ref{tab:theta_hat_hs_methods}), yet the boundary second moment stays near its dispersion floor and no partition of the profiles gains posterior support. Because $Z$ is purely discrete the direction is at most set-identified, so even agreement of the point directions identifies a partition only weakly. The secondary analysis therefore provides no evidence for a subgroup boundary in the binary profiles, and we read its tables as descriptive summaries under $\tau$-regularization rather than as a discovery.

\subsection{Decision summaries and profiles for the secondary analysis}

This subsection collects the decision table and the membership and contrast figures for the secondary binary analysis at $\tau=0.035$, discussed in Section~\ref{sec:analysis_secondary}. The profile membership table appears as main-text Table~\ref{tab:premier_profile_summary}.

\begin{table*}
  \centering
  \caption{Secondary binary-profile analysis of PREMIER at $\tau=0.035$. Columns report the posterior probability of the clinically meaningful heterogeneity event $H_\delta=\{|\gamma|\ge 3\}$, the Bayes action under $p_{\mathrm{report}}=0.9$, boundary stability $\lambda_{\max}(M)$, average hard membership $\bar q$, and the heterogeneity contrast $\gamma$ with its credible interval.}
  \begin{adjustbox}{width=10cm}
  \begin{tabular}{c rrrrrr}
    \toprule
      \textbf{Model}
      & $\Pi(H_\delta\mid \mathcal D_n)$ & \textup{Bayes action} & $\lambda_{\max}(M)$
      & $\bar{q}$ & $\gamma$
      &  CI of $\gamma$   \\
    \midrule
    Parametric & 0.52 & $\mathsf{a}_0$ & 0.18 & 0.72 & -0.70 & (-18.08, 15.93)   \\
    BART       & 0.48 & $\mathsf{a}_0$ & 0.21 & 0.69 & -0.50 & (-7.98, 7.07)   \\
    \bottomrule
  \end{tabular}
  \end{adjustbox}
  \label{tab:premier_decision_binary}
\end{table*}

\begin{figure}
    \centering
    \begin{subfigure}[t]{0.48\textwidth}
    \includegraphics[width=\linewidth]{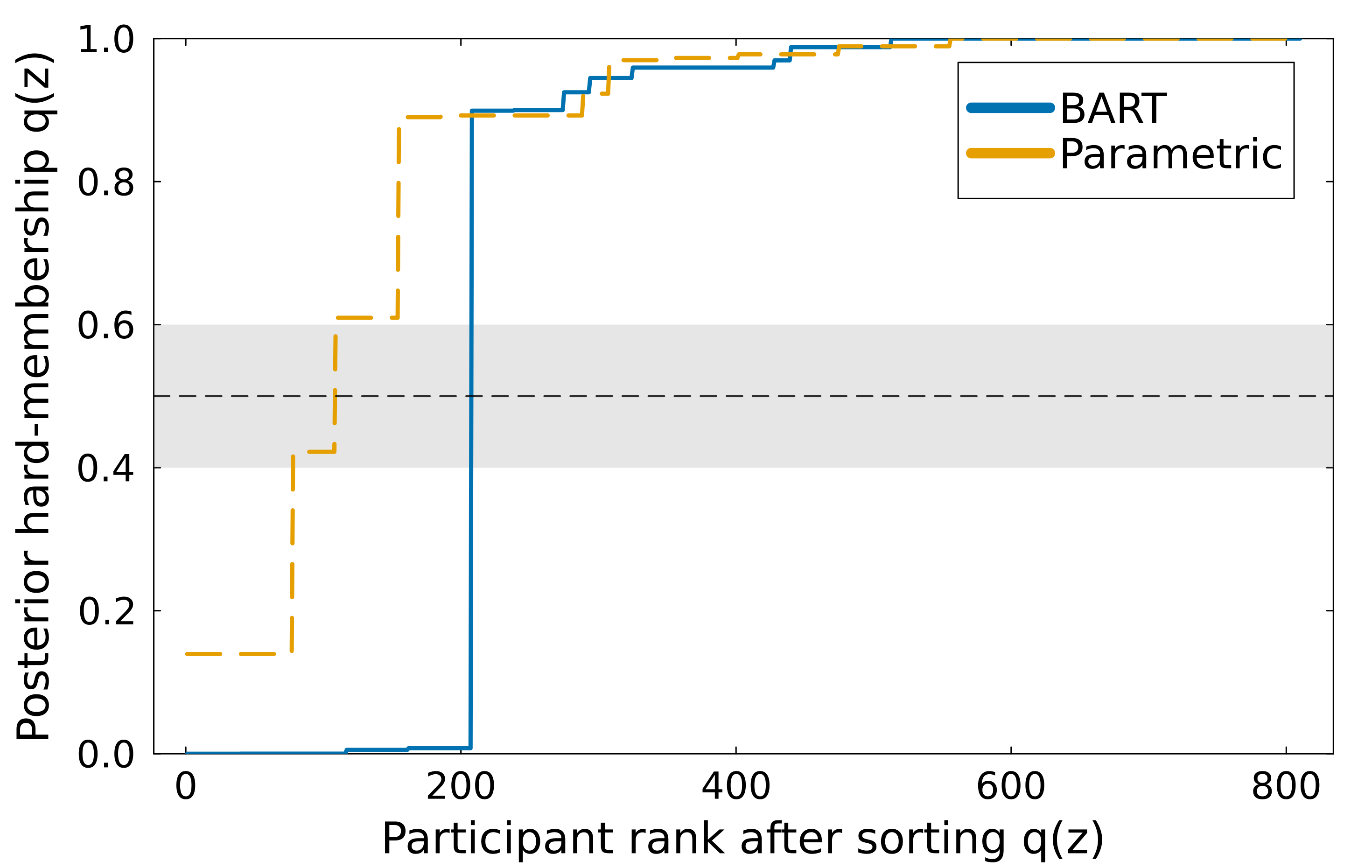}
    \caption{}
    \label{fig:s_sorted_0035}
    \end{subfigure}
    \begin{subfigure}[t]{0.48\textwidth}
    \includegraphics[width=\linewidth]{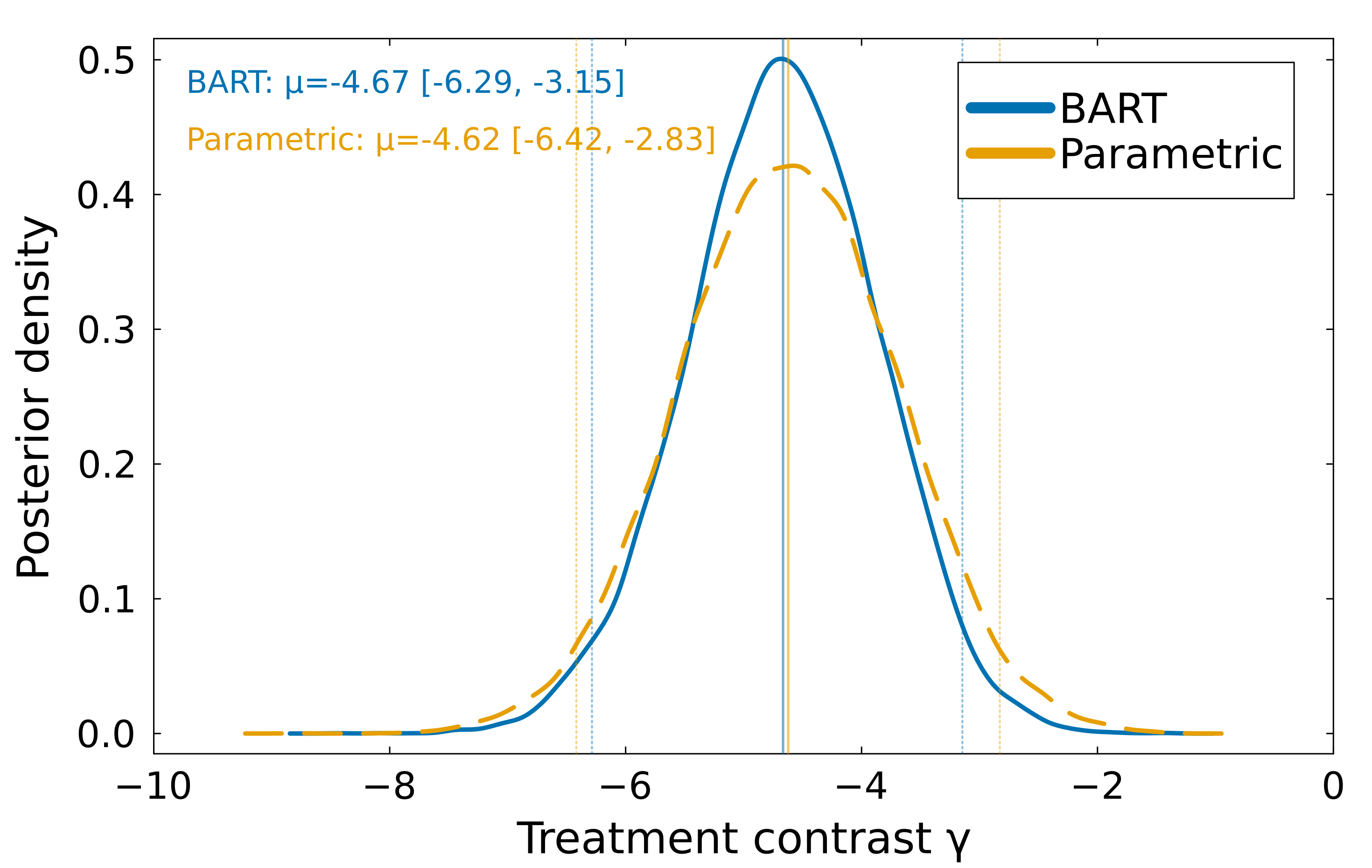}
    \caption{}
    \label{fig:treatment_0035}
    \end{subfigure}
    \caption{Secondary binary analysis at $\tau=0.035$. (a) Sorted posterior hard-membership probabilities $q(Z_i)$. The dashed horizontal line marks $q(Z_i)=0.5$. (b) Posterior densities for the heterogeneity contrast parameter. Dashed vertical lines mark $\pm\delta$ with $\delta=3$ mmHg.}
\end{figure}

\subsection{Sensitivity tables for the secondary analysis}
This subsection collects the figures and tables supporting the sensitivity analyses of the secondary binary analysis, discussed in Section~\ref{sec:sensitivity_analysis_delta_p}. Table~\ref{tab:sensitivity_delta_p} recomputes the reporting decision over the grid of $(\delta,p_{\mathrm{report}})$ preference inputs at $\tau=0.035$. Table~\ref{tab:contrast_decision_rule_result_tau01_001} reports the smoothing-sensitivity decision summaries at $\tau\in\{0.1,0.01\}$, Figures~\ref{fig:s_sorted_01} and \ref{fig:treatment_01} display the membership and contrast posteriors at $\tau=0.1$, and Figures~\ref{fig:s_sorted_001} and \ref{fig:treatment_001} display them at $\tau=0.01$.

For the secondary binary analysis the reporting inputs are preference quantities rather than properties of the data. At the deliberately low threshold $\delta=0.5$, the tail probability is $0.93$ under both baselines, so both report through $p_{\mathrm{report}}=0.90$ but withhold at $0.95$ and $0.99$. At $\delta=1$ both report only at thresholds at or below $0.85$. Neither baseline reports a boundary anywhere in the clinically plausible band. At $\delta=2$ the tail probability is $0.69$ under both baselines and at $\delta=3$ it is $0.52$ and $0.48$, so every displayed reporting threshold yields $\mathsf a_0$ (Table~\ref{tab:sensitivity_delta_p}). The headline pair $(3,0.9)$ therefore lies deep inside the withhold region rather than on its edge, and the decision surface is nearly identical under the two baselines.

Across the smoothing grid the secondary decision is uniformly $\mathsf a_0$. The tail probability rises with the smoothing scale, from $0.57$ and $0.65$ at $\tau=0.01$ to $0.72$ and $0.75$ at the oversmoothed $\tau=0.1$, but never approaches the reporting bar, while the posterior mean of $\gamma$ drifts from near zero at $\tau=0.035$ to between $-5$ and $-10$ mmHg with credible intervals spanning roughly $35$ mmHg (Table~\ref{tab:contrast_decision_rule_result_tau01_001}). Read against the aggregation rule of Section~\ref{sec:choice_tau}, the secondary analysis withholds at every smoothing value, and the drift of the point summaries with $\tau$ is a further caution against attaching scientific meaning to any single fit. One caveat remains. The window and its endpoints are derived under the continuous-margin condition, which the purely discrete $Z$ of this analysis violates, so the window is at most a heuristic guide here. All fits use the collapsed update of Section~\ref{sec:inference}.

\begin{table}[!htbp]
  \centering
  \caption{Sensitivity of the reporting decision to $(\delta, p_{\mathrm{report}})$ at $\tau=0.035$ for the secondary binary analysis, under the linear-HS and BART-HS baselines. For each pair we report the posterior tail probability $\Pi(|\gamma|\ge\delta\mid\mathcal D_n)$ and the resulting Bayes action $\mathsf{a}_1$ or $\mathsf{a}_0$. The headline choice $(\delta,p_{\mathrm{report}})=(3~\mathrm{mmHg},0.9)$ is highlighted in bold.}
  \label{tab:sensitivity_delta_p}
  \small
  \begin{adjustbox}{max width=\textwidth}
  \begin{tabular}{l c c *{6}{c}}
    \toprule
    Baseline & $\delta$ (mmHg) & $\Pi(|\gamma|\ge\delta\mid\mathcal D_n)$ & \multicolumn{6}{c}{Bayes action at $p_{\mathrm{report}}$} \\
    \cmidrule(lr){4-9}
     &  &  & $0.70$ & $0.80$ & $0.85$ & $\boldsymbol{0.90}$ & $0.95$ & $0.99$ \\
    \midrule
    \multirow{6}{*}{Linear-HS}
      & $0.5$ & $0.928$ & $\mathsf{a}_1$ & $\mathsf{a}_1$ & $\mathsf{a}_1$ & $\mathsf{a}_1$ & $\mathsf{a}_0$ & $\mathsf{a}_0$ \\
      & $1$ & $0.853$ & $\mathsf{a}_1$ & $\mathsf{a}_1$ & $\mathsf{a}_1$ & $\mathsf{a}_0$ & $\mathsf{a}_0$ & $\mathsf{a}_0$ \\
      & $2$ & $0.692$ & $\mathsf{a}_0$ & $\mathsf{a}_0$ & $\mathsf{a}_0$ & $\mathsf{a}_0$ & $\mathsf{a}_0$ & $\mathsf{a}_0$ \\
      & $\boldsymbol{3}$ & $0.519$ & $\mathsf{a}_0$ & $\mathsf{a}_0$ & $\mathsf{a}_0$ & $\boldsymbol{\mathsf{a}_0}$ & $\mathsf{a}_0$ & $\mathsf{a}_0$ \\
      & $4$ & $0.375$ & $\mathsf{a}_0$ & $\mathsf{a}_0$ & $\mathsf{a}_0$ & $\mathsf{a}_0$ & $\mathsf{a}_0$ & $\mathsf{a}_0$ \\
      & $5$ & $0.285$ & $\mathsf{a}_0$ & $\mathsf{a}_0$ & $\mathsf{a}_0$ & $\mathsf{a}_0$ & $\mathsf{a}_0$ & $\mathsf{a}_0$ \\
    \midrule
    \multirow{6}{*}{BART-HS}
      & $0.5$ & $0.938$ & $\mathsf{a}_1$ & $\mathsf{a}_1$ & $\mathsf{a}_1$ & $\mathsf{a}_1$ & $\mathsf{a}_0$ & $\mathsf{a}_0$ \\
      & $1$ & $0.864$ & $\mathsf{a}_1$ & $\mathsf{a}_1$ & $\mathsf{a}_1$ & $\mathsf{a}_0$ & $\mathsf{a}_0$ & $\mathsf{a}_0$ \\
      & $2$ & $0.690$ & $\mathsf{a}_0$ & $\mathsf{a}_0$ & $\mathsf{a}_0$ & $\mathsf{a}_0$ & $\mathsf{a}_0$ & $\mathsf{a}_0$ \\
      & $\boldsymbol{3}$ & $0.479$ & $\mathsf{a}_0$ & $\mathsf{a}_0$ & $\mathsf{a}_0$ & $\boldsymbol{\mathsf{a}_0}$ & $\mathsf{a}_0$ & $\mathsf{a}_0$ \\
      & $4$ & $0.289$ & $\mathsf{a}_0$ & $\mathsf{a}_0$ & $\mathsf{a}_0$ & $\mathsf{a}_0$ & $\mathsf{a}_0$ & $\mathsf{a}_0$ \\
      & $5$ & $0.162$ & $\mathsf{a}_0$ & $\mathsf{a}_0$ & $\mathsf{a}_0$ & $\mathsf{a}_0$ & $\mathsf{a}_0$ & $\mathsf{a}_0$ \\
    \bottomrule
  \end{tabular}
  \end{adjustbox}
\end{table}

\begin{table}[!htbp]
  \centering
  \caption{Secondary binary analysis, posterior summaries under $\tau\in\{0.1,0.01\}$ for the clinically meaningful heterogeneity event $H_\delta=\{|\gamma|\ge 3\}$, the Bayes action under $p_{\mathrm{report}}=0.9$, boundary stability, hard subgroup membership, and the heterogeneity contrast $\gamma$.}
  \begin{adjustbox}{width=10cm}
  \begin{tabular}{cl rrrrrr}
    \toprule
      $\tau$ &\textbf{Model}
       & $\Pi(H_\delta\mid \mathcal D_n)$ & \textup{Bayes action} & $\lambda_{\max}(M)$
      & $\bar{q}$ & $\gamma$
      &  CI of $\gamma$   \\
    \midrule
   \multirow{2}{*}{0.1}& Parametric & 0.72 & $\mathsf{a}_0$ & 0.36 & 0.82 & -9.99 & (-28.84, 6.62)   \\
    & BART       & 0.75 & $\mathsf{a}_0$ & 0.35 & 0.81 & -9.32 & (-29.03, 5.13)   \\
    \midrule
    \multirow{2}{*}{0.01}& Parametric & 0.57 & $\mathsf{a}_0$ & 0.19 & 0.78 & -6.74 & (-30.62, 5.13)   \\
    & BART       & 0.65 & $\mathsf{a}_0$ & 0.21 & 0.77 & -5.13 & (-28.24, 6.46)   \\
    \bottomrule
  \end{tabular}
  \end{adjustbox}
  \label{tab:contrast_decision_rule_result_tau01_001}
\end{table}

\begin{figure}
    \centering
    \begin{subfigure}[t]{0.48\textwidth}
    \includegraphics[width=\linewidth]{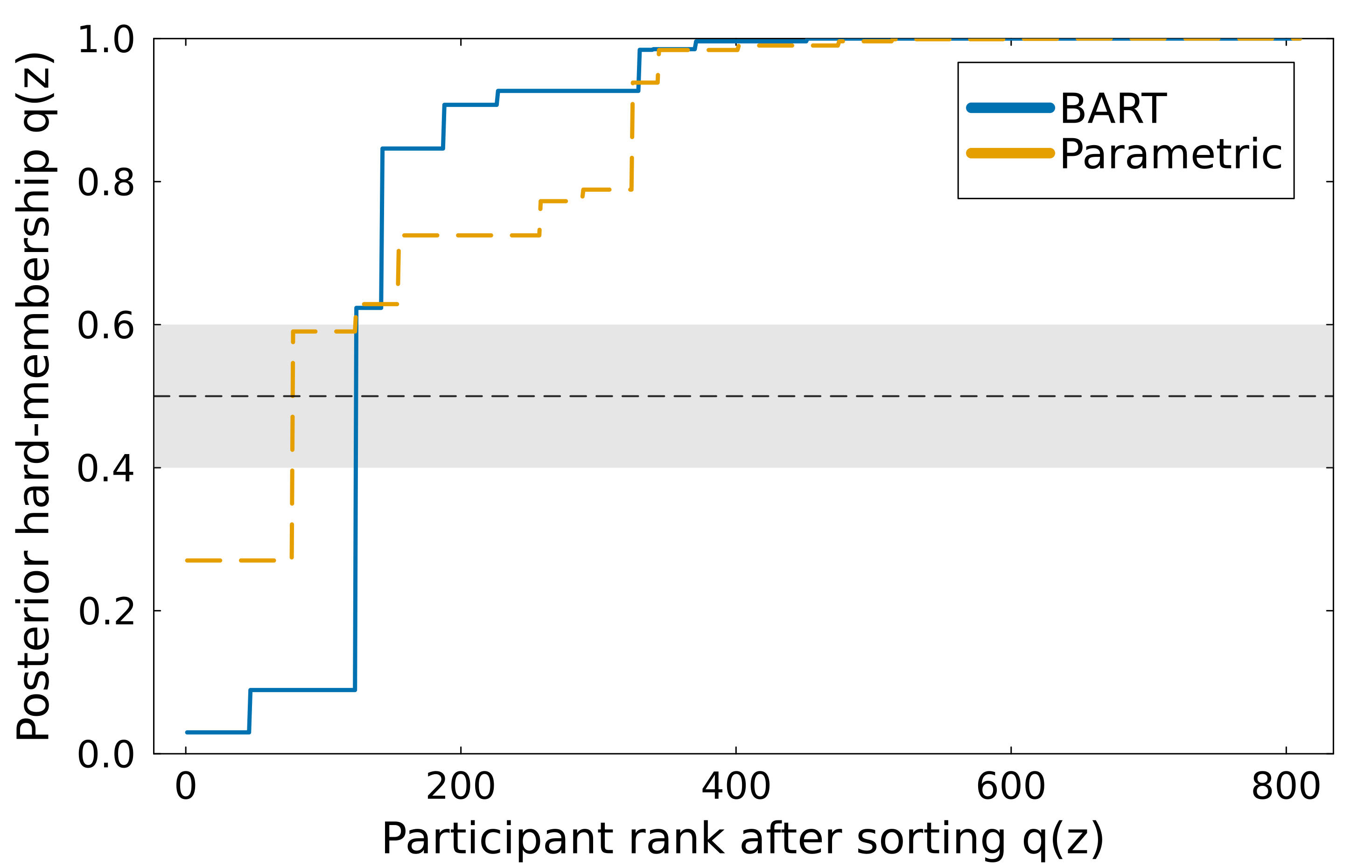}
    \caption{}
    \label{fig:s_sorted_01}
    \end{subfigure}
    \begin{subfigure}[t]{0.48\textwidth}
    \includegraphics[width=\linewidth]{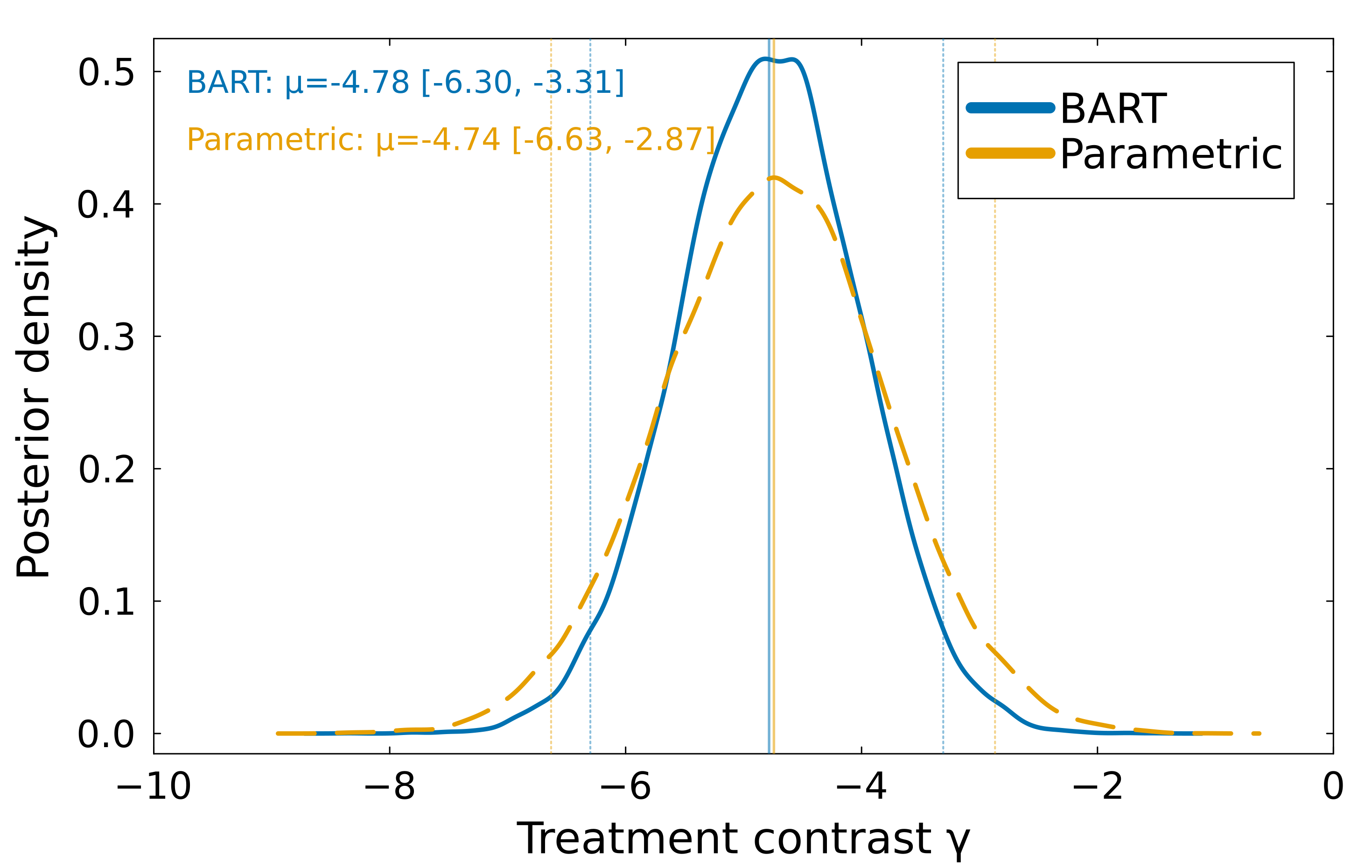}
    \caption{}
    \label{fig:treatment_01}
    \end{subfigure}
    \caption{Secondary binary analysis under $\tau=0.1$. (a) Sorted posterior hard-membership probabilities $q(Z_i)$. The dashed horizontal line marks $q(Z_i)=0.5$. (b) Posterior densities for the heterogeneity contrast parameter.}
\end{figure}

\begin{figure}
    \centering
    \begin{subfigure}[t]{0.48\textwidth}
    \includegraphics[width=\linewidth]{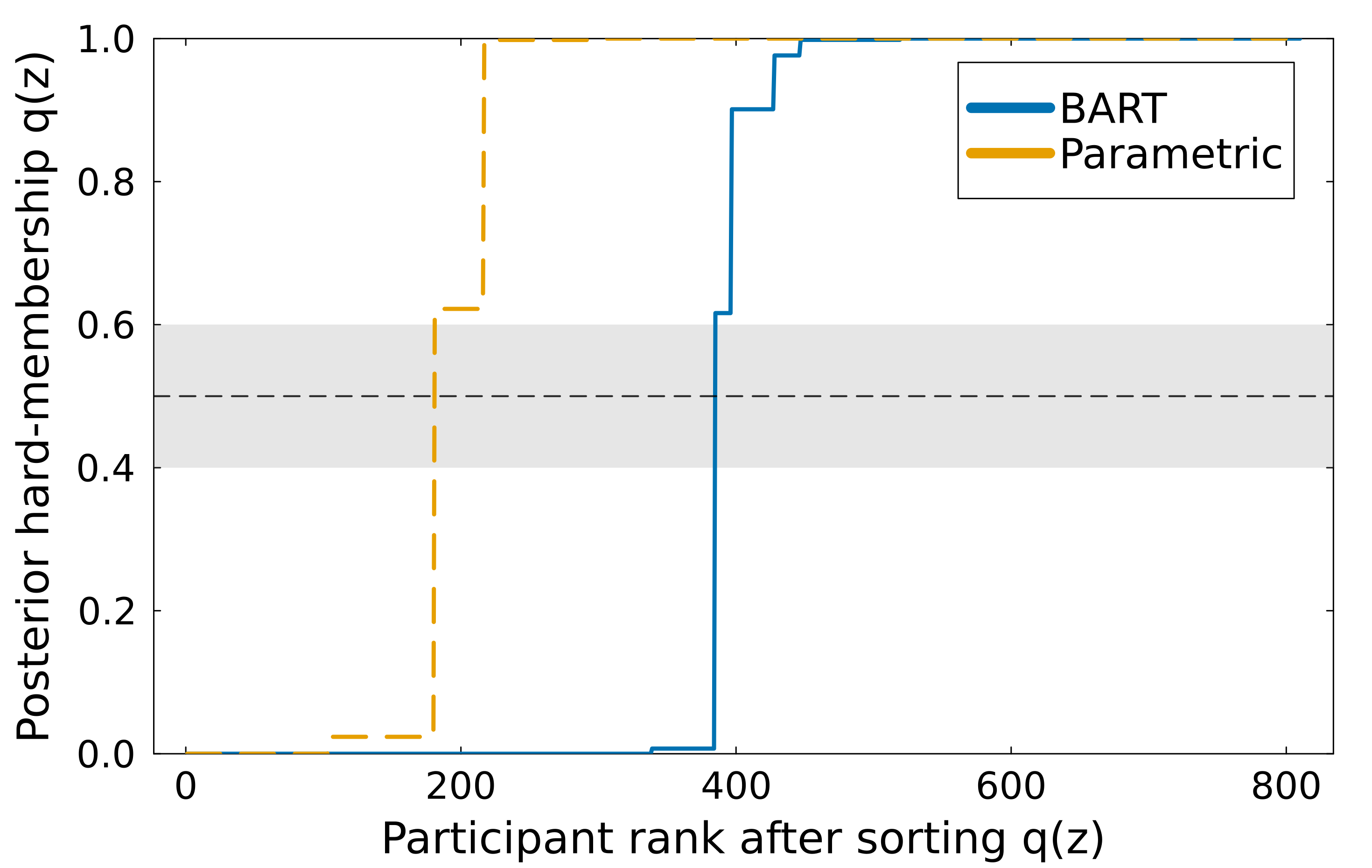}
    \caption{}
    \label{fig:s_sorted_001}
    \end{subfigure}
    \begin{subfigure}[t]{0.48\textwidth}
    \includegraphics[width=\linewidth]{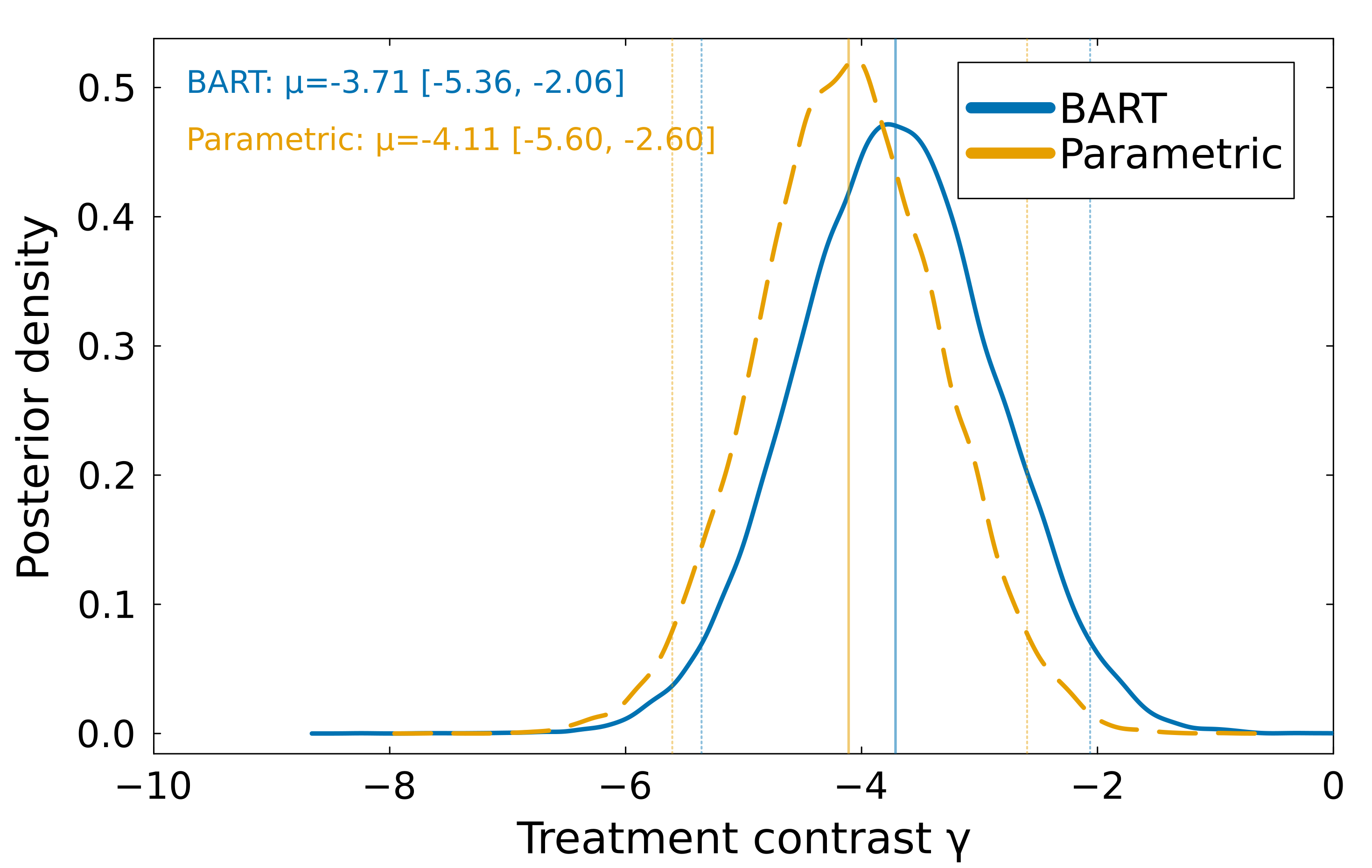}
    \caption{}
    \label{fig:treatment_001}
    \end{subfigure}
    \caption{Secondary binary analysis under $\tau=0.01$. (a) Sorted posterior hard-membership probabilities $q(Z_i)$. The dashed horizontal line marks $q(Z_i)=0.5$. (b) Posterior densities for the heterogeneity contrast parameter.}
\end{figure}

\subsection{Arm-specific sensitivity for the primary analysis}
Table~\ref{tab:premier_arms} reports the prespecified arm-specific sensitivity analysis under the parametric baseline across the full smoothing grid, comparing each active arm with advice only. Both arms withhold a boundary report at every grid value. For the established recommendations arm the posterior probability of the clinically meaningful event stays between $0.53$ and $0.56$ and the boundary is diffuse, with $\lambda_{\max}(M)$ between $0.23$ and $0.24$ against the dispersion floor $1/6\approx0.17$. For the arm augmented with the structured dietary component the probability is higher, between $0.72$ and $0.75$, and the boundary second moment concentrates on a nearly pure age direction, with $\lambda_{\max}(M)$ between $0.61$ and $0.66$ at every grid value. This concentration is suggestive, but the heterogeneity evidence stays short of the reporting bar, so the protocol withholds in the dietary arm as well and no age threshold is interpreted. The treatment effect inside the subgroup, $\delta+\gamma$, is stable over the entire grid in both arms, near $-4.3$ mmHg in the established recommendations arm and between $-4.8$ and $-5.0$ mmHg in the dietary arm, whereas the split into $\delta$ and $\gamma$ carries posterior standard deviations of $2.6$ to $5.1$ mmHg, the weakly identified decomposition discussed in Section~\ref{sec:analysis_result}.

\begin{table}[!htbp]
  \centering
  \caption{Arm-specific sensitivity analysis of PREMIER (parametric baseline, each active arm against advice only) across the smoothing grid. Columns report the posterior mean (standard deviation) of the baseline treatment effect $\delta$ and of the heterogeneity contrast $\gamma$, the posterior mean and 95\% credible interval of the treatment effect inside the subgroup $\delta+\gamma$, the posterior probability of $H_\delta=\{|\gamma|\ge3\}$, the Bayes action under $p_{\mathrm{report}}=0.9$, and the boundary stability $\lambda_{\max}(M)$.}
  \label{tab:premier_arms}
  \begin{adjustbox}{max width=\textwidth}
  \begin{tabular}{clrrrrcc}
    \toprule
    Arm & $\tau$ & $\delta$ (SD) & $\gamma$ (SD) & $\delta+\gamma$ (95\% CI) & $\Pi(H_\delta\mid\mathcal D_n)$ & Bayes action & $\lambda_{\max}(M)$\\
    \midrule
    \multirow{6}{*}{Established recommendations}
      & $0.005$ & $-2.79$ $(3.85)$ & $-1.70$ $(4.63)$ & $-4.49$ $(-10.47, -1.48)$ & $0.551$ & $\mathsf{a}_0$ & $0.239$\\
      & $0.01$  & $-2.29$ $(4.37)$ & $-2.02$ $(4.88)$ & $-4.32$ $(-9.80, -1.47)$ & $0.536$ & $\mathsf{a}_0$ & $0.234$\\
      & $0.02$  & $-3.13$ $(4.00)$ & $-1.18$ $(4.70)$ & $-4.32$ $(-9.84, -1.47)$ & $0.527$ & $\mathsf{a}_0$ & $0.241$\\
      & $0.035$ & $-2.68$ $(3.83)$ & $-1.64$ $(4.47)$ & $-4.32$ $(-9.49, -1.51)$ & $0.526$ & $\mathsf{a}_0$ & $0.235$\\
      & $0.05$  & $-2.80$ $(4.33)$ & $-1.51$ $(5.03)$ & $-4.31$ $(-10.06, -1.37)$ & $0.561$ & $\mathsf{a}_0$ & $0.245$\\
      & $0.1$   & $-2.68$ $(4.45)$ & $-1.55$ $(5.13)$ & $-4.23$ $(-9.34, -1.35)$ & $0.549$ & $\mathsf{a}_0$ & $0.232$\\
    \midrule
    \multirow{6}{*}{Established plus dietary}
      & $0.005$ & $-2.94$ $(2.76)$ & $-1.99$ $(4.17)$ & $-4.93$ $(-8.05, -0.79)$ & $0.739$ & $\mathsf{a}_0$ & $0.656$\\
      & $0.01$  & $-3.03$ $(3.01)$ & $-1.92$ $(4.25)$ & $-4.95$ $(-7.95, -0.92)$ & $0.723$ & $\mathsf{a}_0$ & $0.616$\\
      & $0.02$  & $-3.36$ $(2.57)$ & $-1.43$ $(4.13)$ & $-4.79$ $(-7.99, -0.72)$ & $0.722$ & $\mathsf{a}_0$ & $0.645$\\
      & $0.035$ & $-3.36$ $(2.89)$ & $-1.41$ $(4.34)$ & $-4.77$ $(-8.01, -0.74)$ & $0.738$ & $\mathsf{a}_0$ & $0.631$\\
      & $0.05$  & $-3.17$ $(2.72)$ & $-1.66$ $(4.23)$ & $-4.83$ $(-8.02, -0.68)$ & $0.733$ & $\mathsf{a}_0$ & $0.630$\\
      & $0.1$   & $-3.39$ $(3.47)$ & $-1.43$ $(4.80)$ & $-4.82$ $(-8.10, -0.72)$ & $0.746$ & $\mathsf{a}_0$ & $0.607$\\
    \bottomrule
  \end{tabular}
  \end{adjustbox}
\end{table}

\subsection{Boundary loadings for the primary analysis}
Table~\ref{tab:premier_primary_loadings} reports the posterior principal direction of the pooled primary analysis at the central in-window value $\tau=0.02$ under both baselines. The two directions disagree. The parametric direction loads almost entirely on baseline SBP ($0.95$), whereas the BART direction spreads its weight over African American ($-0.76$), age ($0.45$), and baseline SBP ($0.31$), and both are diffuse ($\lambda_{\max}(M)=0.26$ under either baseline against the dispersion floor $1/6\approx0.17$). Over the whole smoothing grid the pooled $\lambda_{\max}(M)$ stays between $0.26$ and $0.32$ for both baselines, and the leading coordinate is itself unstable, the parametric direction switching to an intercept and hypertension pattern at $\tau=0.1$. No coordinate pattern is stable enough to interpret, which is the quantitative basis for the Step~3(a) decision of Section~\ref{sec:analysis_result}. The contrast with the arm-specific analysis is sharp. Against advice only, the dietary arm concentrates its direction on age in years (loading $0.98$, $\lambda_{\max}(M)$ between $0.61$ and $0.66$) at every grid value, although its heterogeneity evidence also remains below the reporting bar (Table~\ref{tab:premier_arms}).

\begin{table}[!htbp]
  \centering
  \caption{Posterior principal direction of the pooled primary analysis at $\tau=0.02$, by baseline, with the boundary stability $\lambda_{\max}(M)$. Covariates are standardized. The direction is reported up to sign.}
  \label{tab:premier_primary_loadings}
  \begin{adjustbox}{max width=0.75\textwidth}
  \begin{tabular}{lrr}
    \toprule
    Covariate & Parametric & BART\\
    \midrule
    intercept & $0.060$ & $0.255$\\
    age (years) & $0.105$ & $0.453$\\
    baseline SBP & $0.946$ & $0.312$\\
    female & $-0.080$ & $-0.202$\\
    African American & $-0.288$ & $-0.764$\\
    baseline hypertension & $0.022$ & $0.089$\\
    \midrule
    $\lambda_{\max}(M)$ & $0.257$ & $0.260$\\
    \bottomrule
  \end{tabular}
  \end{adjustbox}
\end{table}

\subsection{Membership distributions for the primary analysis}
The hard-membership probabilities $q(Z_i)$ of the pooled primary analysis are compressed in a narrow band above one half. Under the parametric baseline at $\tau=0.02$ the sample quartiles are $0.62$, $0.67$, and $0.71$, with minimum $0.54$ and maximum $0.89$, and the BART baseline gives a nearly identical distribution, so no participant is confidently assigned to either side and no boundary separates the sample. This is the membership-level signature of a diffuse direction, and it is stable across the smoothing grid and across baselines. It contrasts with the dietary-arm analysis, where the age-dominated direction places roughly two fifths of the arm below one half and produces a genuine gradient of membership in age.

\subsection{Per-chain diagnostics for the primary analysis}
Table~\ref{tab:premier_primary_diag} summarizes the convergence diagnostics of the pooled primary analysis, computed from four dispersed chains per fit for the treatment effect inside the subgroup, $\delta+\gamma$, and for the heterogeneity contrast $\gamma$. The identified functional $\delta+\gamma$ passes comfortably everywhere, with worst rank-normalized $\widehat R$ of $1.010$ and smallest bulk effective sample size of $289$ over the pooled grid, and the arm-specific fits behave the same way, with $\widehat R$ at most $1.012$ and bulk ESS at least $231$. The split components $\gamma$ and $\delta$ mix more slowly, with in-window pooled bulk ESS between $110$ and $339$ and $\widehat R$ up to $1.034$, after extending four pooled fits, the parametric at $\tau\in\{0.05,0.1\}$ and BART at $\tau\in\{0.035,0.1\}$, to $60000$ iterations. This reflects the geometry of the posterior rather than a sampler failure. Near a diffuse boundary the posterior correlation between $\gamma$ and $\delta$ is $-0.94$ under the parametric baseline and $-0.95$ under BART, so the data pin the sum while the decomposition is weakly identified and the chain traverses a long ridge. All fits use the collapsed update. An earlier augmented-update run of the two oversmoothed BART fits failed these diagnostics with $\widehat R$ up to $1.40$, the wide-gate pathology documented in Section~\ref{sec:simulations}.

\begin{table}[!htbp]
  \centering
  \caption{Convergence diagnostics of the pooled primary analysis by baseline and smoothing value. Rank-normalized $\widehat R$ and bulk ESS for the treatment effect inside the subgroup $\delta+\gamma$ and for the heterogeneity contrast $\gamma$, four dispersed chains per fit.}
  \label{tab:premier_primary_diag}
  \begin{adjustbox}{max width=0.85\textwidth}
  \begin{tabular}{lrrrrrr}
    \toprule
    & \multicolumn{6}{c}{Smoothing value $\tau$}\\
    \cmidrule(lr){2-7}
    & $0.005$ & $0.01$ & $0.02$ & $0.035$ & $0.05$ & $0.1$\\
    \midrule
    Parametric, $\widehat R$ of $\delta+\gamma$ & $1.005$ & $1.006$ & $1.001$ & $1.004$ & $1.001$ & $1.002$\\
    Parametric, ESS of $\delta+\gamma$ & $1164$ & $1077$ & $1178$ & $1059$ & $2906$ & $2639$\\
    Parametric, $\widehat R$ of $\gamma$ & $1.027$ & $1.030$ & $1.017$ & $1.010$ & $1.004$ & $1.007$\\
    Parametric, ESS of $\gamma$ & $215$ & $157$ & $189$ & $203$ & $257$ & $190$\\
    BART, $\widehat R$ of $\delta+\gamma$ & $1.009$ & $1.010$ & $1.008$ & $1.003$ & $1.010$ & $1.001$\\
    BART, ESS of $\delta+\gamma$ & $514$ & $289$ & $610$ & $2020$ & $550$ & $1036$\\
    BART, $\widehat R$ of $\gamma$ & $1.009$ & $1.017$ & $1.034$ & $1.007$ & $1.015$ & $1.023$\\
    BART, ESS of $\gamma$ & $221$ & $130$ & $110$ & $339$ & $130$ & $95$\\
    \bottomrule
  \end{tabular}
  \end{adjustbox}
\end{table}

\end{document}